\documentclass[smallcondensed,nospthms,draft]{svjour3}
\usepackage[T1]{fontenc}
\usepackage[utf8]{inputenc}
\usepackage[english]{babel}

\usepackage{mathptmx}      
\usepackage{amsmath,amssymb,amsthm}
\smartqed  
\usepackage[only,bigsqcap,llbracket,rrbracket,llparenthesis,rrparenthesis]{stmaryrd}
\usepackage{wasysym}	
\usepackage{graphicx}
\usepackage{thm-restate}
\usepackage{subfig}
\usepackage{xspace}
\usepackage{ifthen}
\usepackage{enumitem}
\usepackage{array}
\usepackage{booktabs}
\usepackage{longtable}
\usepackage{colortbl}
\usepackage{multirow}
\usepackage[babel]{csquotes}
\usepackage{cite}
\usepackage{doi}
\usepackage{tabularx}
\usepackage{tabu}
\usepackage[linesnumbered,algoruled,vlined,algosection]{algorithm2e}
\DontPrintSemicolon
\usepackage{tikz}
\usetikzlibrary{positioning,calc,fit,shapes, shapes.multipart, arrows,matrix,decorations.pathreplacing}

\pgfmathtruncatemacro\distance{1}

\usepackage[textwidth=2cm,textsize=scriptsize]{todonotes}
\presetkeys{todonotes}{color=blue!20!white}{}
\newcommand{\mygray}[1]{{\color{darkgray} #1 }}

\setlist[description]{labelwidth=0pt,style=unboxed}
\setlist[enumerate]{leftmargin=*}

\newcommand\Tstrut{\rule{0pt}{2.6ex}}         
\newcommand\Bstrut{\rule[-1ex]{0pt}{0pt}}   
\newcommand*\tabtopline{\hline\Tstrut}
\newcommand*\tabbotline{\Bstrut\\\hline}
\newcommand*\tabmidline{\Bstrut\\\hline\Tstrut}

\newlength\reptabvspace
\newlength\inreptabvspace
\newlength\tabmedlinesep
\definecolor{TabGray}{gray}{0.9}
\definecolor{AlgoGray}{gray}{0.88}
\definecolor{AlgoLightGray}{gray}{0.95}
\definecolor{Gray}{gray}{0.9}

\definecolor{LightGray}{gray}{0.95}
\newcolumntype{g}{>{\columncolor{Gray}}l}

\newcommand{\tablesize}{\scriptsize}
\newcommand{\source}[1]{\linebreak\tablesize\mygray{#1}}

\colorlet{blcolor}{gray!80}

\newcommand{\mycolorbox}[2]{%
\begin{tikzpicture}[baseline=(O.base)]%
	\node[fill=#1,inner sep=1pt](O){#2\vphantom{$A^I_J$}};%
\end{tikzpicture}%
}

\newcommand{\ie}{i.e.,\ }
\newcommand{\eg}{e.g.,\ }
\newcommand{\wrt}{w.r.t.\ }
\newcommand{\cf}{cf.\ }

\newcommand{\TTypes}{\ensuremath{\mathbf{T}}\xspace} 

\newcommand{\atuple}{\ensuremath{\mathsf{t}}\xspace} 

\newcommand{\aboxtype}[1]{\ensuremath{\Amc^+_{#1}}\xspace}
\newcommand{\typeai}[2]{\ensuremath{\smash{\Tmc^{#1}_{#2}}}\xspace}
\newcommand{\ttypea}[1]{\ensuremath{\smash{\tau^\atuple_{#1}}}\xspace}
\newcommand{\tabox}[2]{\ensuremath{\Amc^{#1}_{#2}}\xspace}

\newcommand{\cond}{Condition~}
\newcommand{\rccond}{C}

\newcommand{\citethm}{Thm.~}
\newcommand{\citelem}{Lem.~}

\newcommand{\citecor}{Cor.~}

\newcommand{\algofont}[1]{\texttt{#1}}

\newcommand{\algoTesting}{RSATISFIABLE}

\newcommand{\algoATMRecursion}{ATMRECURSION}
\newcommand{\algoATMFinal}{ATMFINAL}

\newcommand{\fnow}{{curr}}
\newcommand{\fnext}{{next}}
\newcommand{\fs}{s}

\newcommand{\dllhhornQUnsat}[1]{\ensuremath{\llbracket \bot\rrbracket}\xspace}
\newcommand{\dllhhornPerfRef}[2]{\ensuremath{\llbracket #1\rrbracket}\xspace}

\newcommand{\asymbol}{\ensuremath{X}\xspace}
\newcommand{\vtwo}{\ensuremath{y}\xspace}

\newcommand{\tone}{\ensuremath{s}\xspace}
\newcommand{\ttwo}{\ensuremath{t}\xspace}

\newcommand{\amonoid}{\ensuremath{M}\xspace}

\newcommand{\domain}{\ensuremath{\mathsf{domain}}\xspace}
\newcommand{\range}{\ensuremath{\mathsf{range}}\xspace}

\newcommand{\pa}[1]{\ensuremath{#1^{\smash{\mathrm{pa}}}}\xspace}
\newcommand{\pba}[1]{\ensuremath{#1^{\smash{\mathrm{ba}}}}\xspace}

\newcommand{\ahom}{\ensuremath{\pi}\xspace}
\newcommand{\ahomtwo}{\ensuremath{\pi'}\xspace}

\newcommand{\changed}[1]{{\color{black}#1}}
\newcommand{\ctwo}[1]{{\color{black}#1}}

\renewcommand{\iota}{\lambda}
\renewcommand{\omega}{\nu}
\newcommand{\eps}{\epsilon}
\newcommand{\true}{\ensuremath{\mathit{true}}\xspace}
\newcommand{\false}{\ensuremath{\mathit{false}}\xspace}

\newcommand{\atopform}{\ensuremath{f}\xspace}

\newcommand{\atype}{\ensuremath{T}\xspace}
\newcommand{\atypeset}[1][\Pmc]{\ensuremath{\mathsf{Typ}(#1)}\xspace}
\newcommand{\Clo}{\ensuremath{\mathsf{Clo}}\xspace}
\newcommand{\atmfut}[2]{\ensuremath{\mathrm{Fut}_{#2}}\xspace}

\newcommand{\afformset}{\ensuremath{\mathcal{F}}\xspace}
\newcommand{\apformset}{\ensuremath{\mathcal{P}}\xspace}

\newcommand{\C}[1]{\ensuremath{[#1]}\xspace}

\newcommand{\adlsignature}{\ensuremath{\Sigma}\xspace}
\newcommand{\aname}{\ensuremath{X}\xspace}

\newcommand{\indone}{\ensuremath{a}\xspace}
\newcommand{\indtwo}{\ensuremath{b}\xspace}
\newcommand{\el}{\ensuremath{e}\xspace}
\newcommand{\eladd}{\ensuremath{d}\xspace}
\newcommand{\elone}{\ensuremath{d}\xspace}
\newcommand{\eltwo}{\ensuremath{e}\xspace}
\newcommand{\aaxiom}{\ensuremath{\alpha}\xspace}

\newcommand{\aont}{\ensuremath{\mathcal{O}}\xspace}

\newcommand{\afb}{\ensuremath{\mathcal{A}}\xspace}

\newcommand{\akb}{\ensuremath{\mathcal{K}}\xspace}
\newcommand{\aint}{\ensuremath{\mathcal{I}}\xspace}

\newcommand{\adom}[1][]{\ensuremath{\Delta^{{#1}}}\xspace}

\newcommand{\un}{\ensuremath{\mathsf{anon}}\xspace}
\newcommand{\apath}{\ensuremath{\varrho}\xspace}
\newcommand{\apathtwo}{\ensuremath{\sigma}\xspace}

\newcommand{\ontsub}[1][\aont]{\ensuremath{\mathbf{S}(\aont)}\xspace}

\newcommand{\DB}{\ensuremath{\mathsf{DB}}\xspace}
\newcommand{\TDB}{\ensuremath{\mathsf{TDB}}\xspace}

\newcommand{\apropset}{\ensuremath{\mathcal{P}}\xspace}
\newcommand{\altlsignature}{\apropset}
\newcommand{\altlform}{\ensuremath{\varphi}\xspace}
\newcommand{\altlformtwo}{\ensuremath{\psi}\xspace}
\newcommand{\altlint}{\ensuremath{\Wmf}\xspace}

\newcommand{\aquery}{\ensuremath{\varphi}\xspace}
\newcommand{\aquerytwo}{\ensuremath{\psi}\xspace}

\newcommand{\acq}{\aquery}
\newcommand{\acqtwo}{\aquerytwo}

\newcommand{\aatom}{\ensuremath{\aaxiom}\xspace}

\newcommand{\Q}{\ensuremath{\mathcal{Q}}\xspace}

\newcommand{\cqinst}[1]{\ensuremath{\afb_{#1}}\xspace} 

\newcommand{\Var}{ \NV}
\newcommand{\Ind}{\NI}

\newcommand{\afbs}{\ensuremath{\mathfrak{A}}\xspace}
\newcommand{\atkb}{\ensuremath{\mathcal{K}}\xspace}
\newcommand{\atcq}{\ensuremath{\Phi}\xspace}
\newcommand{\atcqtwo}{\ensuremath{\Psi}\xspace}

\newcommand{\atint}{\ensuremath{\mathfrak{I}}\xspace}

\newcommand{\as}{\ensuremath{\mathcal{W}}\xspace}
\newcommand{\ax}{\ensuremath{W}\xspace}

\newcommand{\pv}{{\ensuremath{\{p_1,\dots,p_m\}}}\xspace}

\newcommand{\acqalpha}{\ensuremath{\acq}\xspace}
\newcommand{\acqalphatwo}{\ensuremath{\acq'}\xspace}
\newcommand{\aalpha}{\acqalpha}

\newcommand{\tcqcqs}[1]{\ensuremath{\mathcal{Q}_{#1}}\xspace} 

\newcommand{\awit}{\ensuremath{\mathcal{B}\rig}\xspace}
\newcommand{\wits}[2]{\ensuremath{\mathsf{Wit}(#2)}\xspace}
\newcommand{\abcset}{\ensuremath{{\mathcal{B}\rig}}\xspace}

\newcommand{\rigcons}[1]{\ensuremath{\Amc\rig_{#1}}\xspace}

\newcommand{\NIA}{\ensuremath{\NI^{\mathsf{aux}}}\xspace}
\newcommand{\NIT}{\ensuremath{\NI^{\mathsf{tree}}}\xspace}
\newcommand{\NITm}{\ensuremath{\NI^{\mathsf{tree}-}}\xspace}
\newcommand{\NIp}{\ensuremath{\NI^*}\xspace}

\newcommand{\AR}{\ensuremath{\Amc\rig}\xspace}
\newcommand{\QR}{\ensuremath{\Q^+}\xspace} %
\newcommand{\QRn}{\ensuremath{\Q^-}\xspace} %
\newcommand{\RF}[1][]{\ensuremath{\Smc_{#1}}\xspace}
\newcommand{\KR}[1][]{\ensuremath{\akb_{\mathsf{rc}}^{#1}}\xspace}

\newcommand{\RC}[2]{\ensuremath{{#1}^{#2^\mathsf{R}}}\xspace}

\newcommand{\adomp}[2][]{\ensuremath{\adom[#1]_{\smash{#2}}}\xspace}

\newcommand{\aux}{\ensuremath{\mathsf{aux}}\xspace}
\newcommand{\tree}{\ensuremath{\mathsf{tree}}\xspace}
\newcommand{\nam}{\ensuremath{\mathsf{n}}\xspace}

\newcommand{\ael}[2][]{\ensuremath{a_{\smash{#2}}^{\smash{#1}}}\xspace}
\newcommand{\aiel}[1]{\ensuremath{a^{\smash i}_{\smash{#1}}}\xspace}

\newcommand{\uel}[2][]{\ensuremath{u_{\smash{#2}}^{\smash{#1}}}\xspace}
\newcommand{\uiel}[1]{\ensuremath{u^{\smash i}_{\smash{#1}}}\xspace}
\newcommand{\ujel}[1]{\ensuremath{u^{\smash j}_{\smash{#1}}}\xspace}

\newcommand{\aautom}[1][]{\ensuremath{\mathfrak{M}_{#1}}\xspace}
\newcommand{\atm}[1][]{\ensuremath{\mathfrak{M}_{#1}}\xspace}

\newcommand{\automQ}{\ensuremath{Q}\xspace}
\newcommand{\automA}{\ensuremath{\Sigma}\xspace}
\newcommand{\automT}{\ensuremath{\Delta}\xspace}
\newcommand{\automqs}{\ensuremath{q_0}\xspace}
\newcommand{\automF}{\ensuremath{F}\xspace}

\newcommand{\rigctypes}{\ensuremath{\mathbf{F}}\xspace}
\newcommand{\rigctype}{\ensuremath{F}\xspace}

\newcommand{\tdboa}{\TDB\xspace}
\newcommand{\tdbfbs}{\tdb{\afbs}}
\newcommand{\tdb}[1]{\ensuremath{\TDB(#1)}\xspace}
\newcommand{\dboa}{{\DB}\xspace}
\newcommand{\db}[1]{\ensuremath{\DB(#1)}\xspace}%

\newcommand{\B}{\ensuremath{\mathsf{B}}\xspace}
\newcommand{\R}{\ensuremath{\mathsf{R}}\xspace}

\newcommand{\Bdb}{\ensuremath{\mathsf{B^\dboa}}\xspace}
\newcommand{\Rdb}{\ensuremath{\mathsf{R^\dboa}}\xspace}

\newcommand{\KRis}[1][i]{\ensuremath{\widetilde{\Kmc}_{\mathsf{rc}}^{#1}}\xspace}
\newcommand{\KRisp}[1][i]{\ensuremath{\widetilde{\Kmc}_{\mathsf{rc,pos}}^{#1}}\xspace}
\newcommand{\AKRs}{\ensuremath{{\widetilde{\Amc}_{\mathsf{rc}}}}\xspace}
\newcommand{\Bphi}{\ensuremath{\Bmc}\xspace}

\newcommand{\QRs}{\ensuremath{\widetilde{\Q}^+}\xspace} %
\newcommand{\QRsn}{\ensuremath{\widetilde{\Q}^-}\xspace} %
\newcommand{\ARs}{\ensuremath{\widetilde{\Amc}\rig}\xspace}
\newcommand{\NBC}{\ensuremath{r}\xspace}
\newcommand{\RFs}[1][]{\ensuremath{\widetilde{\Smc}_{#1}}\xspace}
\newcommand{\RFaux}{\ensuremath{\RFs[\auxp]}\xspace}
\newcommand{\RFphi}{\ensuremath{\RFs[\atcq]}\xspace}
\newcommand{\RFother}{\ensuremath{\RFs[\mathsf{o}]}\xspace}

\newcommand{\auxp}{\ensuremath{\mathsf{aux}}\xspace}

\newcommand{\prefo}[1][\acq]{\ensuremath{\llbracket #1\rrbracket}\xspace}
\newcommand{\prefth}[2][\acq]{\ensuremath{\llbracket #1|\ARs_{#2}\rrbracket}\xspace}

\newcommand{\NIP}{\ensuremath{\NI^{\mathsf{pro}}}\xspace}
\newcommand{\rewQUnsat}[1][i]{\ensuremath{\llbracket\bot|\AKRs\rrbracket(#1)}\xspace}
\newcommand{\rewPRef}[2][i]{\ensuremath{\llbracket #2|\AKRs\rrbracket(#1)}\xspace}

\newcommand{\filter}{\ensuremath{\mathsf{filter}}\xspace}

\newcommand{\repn}[1]{\ensuremath{\llbracket #1\rrbracket}\xspace}
\newcommand{\repo}[1]{\ensuremath{\repn{#1}_\mathsf{rep}}\xspace}
\newcommand{\rep}[2][i]{\ensuremath{\repo{#2}(#1)}\xspace}
\newcommand{\repone}[2][i]{\ensuremath{\repn{#2}_1(#1)}\xspace}
\newcommand{\reptwo}[2][i]{\ensuremath{\repn{#2}_2(#1)}\xspace}
\newcommand{\repthree}[2][i]{\ensuremath{\repn{#2}_3(#1)}\xspace}

\newcommand{\rsat}[1][\ax]{\ensuremath{\mathsf{rSat}_{#1}}\xspace}
\newcommand{\cons}{\ensuremath{f_{\text{\rccond1}}}\xspace}
\newcommand{\answit}{\ensuremath{f_{\text{\rccond5}}}\xspace}
\newcommand{\ansx}{\ensuremath{f_{\text{\rccond2}}}\xspace}

\newcommand{\ground}{\ensuremath{\gamma}\xspace}

\newcommand{\Amc}{\ensuremath{\mathcal{A}}\xspace}
\newcommand{\Bmc}{\ensuremath{\mathcal{B}}\xspace}
\newcommand{\Cmc}{\ensuremath{\mathcal{C}}\xspace}
\newcommand{\Dmc}{\ensuremath{\mathcal{D}}\xspace}

\newcommand{\Fmc}{\ensuremath{\mathcal{F}}\xspace}

\newcommand{\Hmc}{\ensuremath{\mathcal{H}}\xspace}
\newcommand{\Imc}{\ensuremath{\mathcal{I}}\xspace}
\newcommand{\Jmc}{\ensuremath{\mathcal{J}}\xspace}
\newcommand{\Kmc}{\ensuremath{\mathcal{K}}\xspace}

\newcommand{\Mmc}{\ensuremath{\mathcal{M}}\xspace}

\newcommand{\Omc}{\ensuremath{\mathcal{O}}\xspace}
\newcommand{\Pmc}{\ensuremath{\mathcal{P}}\xspace}
\newcommand{\Qmc}{\ensuremath{\mathcal{Q}}\xspace}
\newcommand{\Rmc}{\ensuremath{\mathcal{R}}\xspace}
\newcommand{\Smc}{\ensuremath{\mathcal{S}}\xspace}
\newcommand{\Tmc}{\ensuremath{\mathfrak{T}}\xspace}

\newcommand{\Vmc}{\ensuremath{\mathcal{V}}\xspace}

\newcommand{\Ymc}{\ensuremath{\mathcal{Y}}\xspace}

\newcommand{\Wmf}{\ensuremath{\mathfrak{W}}\xspace}

\newcommand{\Xbb}{\ensuremath{\mathbf{X}}\xspace}

\newcommand{\Asf}{\ensuremath{\mathsf{A}}\xspace}

\newcommand{\EL}{\ensuremath{\mathcal{E\!L}}\xspace}
\newcommand{\ALC}{\ensuremath{\mathcal{ALC}}\xspace}

\newcommand{\ALCI}{\ensuremath{\mathcal{ALCI}}\xspace}

\newcommand{\SHQ}{\ensuremath{\mathcal{SHQ}}\xspace}
\newcommand{\SHOQ}{\ensuremath{\mathcal{SHOQ}}\xspace}

\newcommand{\DLLite}{\textit{DL-Lite}\xspace}

\newcommand{\DLLitehcore}{\ensuremath{\textit{DL-Lite}^\Hmc_{\textit{core}}}\xspace}
\newcommand{\DLLitecore}{\ensuremath{\textit{DL-Lite}_{\textit{core}}}\xspace}
\newcommand{\DLLitebool}{\ensuremath{\textit{DL-Lite}_{\textit{bool}}}\xspace}
\newcommand{\DLLitehbool}{\ensuremath{\textit{DL-Lite}^\Hmc_{\textit{bool}}}\xspace}
\newcommand{\DLLitehkrom}{\ensuremath{\textit{DL-Lite}^\Hmc_{\textit{krom}}}\xspace}
\newcommand{\DLLitekrom}{\ensuremath{\textit{DL-Lite}_{\textit{krom}}}\xspace}
\newcommand{\DLLitehhorn}{\ensuremath{\textit{DL-Lite}^\Hmc_{\textit{horn}}}\xspace}
\newcommand{\DLLitehorn}{\ensuremath{\textit{DL-Lite}_{\textit{horn}}}\xspace}

\newcommand{\ExpTime}{\ensuremath{\textsc{ExpTime}}\xspace}
\newcommand{\NExpTime}{\ensuremath{\textsc{NExpTime}}\xspace}
\newcommand{\TwoExpTime}{\ensuremath{\textsc{2-ExpTime}}\xspace}
\newcommand{\NP}{\ensuremath{\textsc{NP}}\xspace}
\newcommand{\coNP}{\ensuremath{\textsc{co-NP}}\xspace}
\newcommand{\coNExpTime}{\ensuremath{\textsc{co-NExpTime}}\xspace}
\newcommand{\PTime}{\ensuremath{\textsc{P}}\xspace}
\newcommand{\PSpace}{\ensuremath{\textsc{PSpace}}\xspace}

\newcommand{\ExpSpace}{\ensuremath{\textsc{ExpSpace}}\xspace}

\newcommand{\ACzero}{\ensuremath{\textsc{AC}^0}\xspace}

\newcommand{\NCone}{\ensuremath{\textsc{NC}^1}\xspace}
\newcommand{\NLogSpace}{\ensuremath{\textsc{NLogSpace}}\xspace}

\newcommand{\ALogTime}{\ensuremath{\textsc{ALogTime}}\xspace}
\newcommand{\DLogTime}{\ensuremath{\textsc{LogTime}}\xspace}
\newcommand{\LogTime}{\ensuremath{\textsc{LogTime}}\xspace}

\newcommand{\Next}{\ensuremath{\ocircle_F}\xspace}

\newcommand{\Previous}{\ensuremath{\ocircle_P}\xspace}

\DeclareMathOperator*{\Since}{\mathcal{S}}
\DeclareMathOperator*{\Until}{\mathcal{U}}

\newcommand{\Diamondf}{\ensuremath{\Diamond_F}\xspace}
\newcommand{\Boxf}{\ensuremath{\Box_F}\xspace}
\newcommand{\Diamondm}{\ensuremath{\Diamond_P}\xspace}
\newcommand{\Boxm}{\ensuremath{\Box_P}\xspace}

\newcommand{\rig}{^{\mathsf{R}}}
\newcommand{\fle}{^{\mathsf{F}}}
\newcommand{\NC}{\ensuremath{\mathbf{C}}\xspace}
\newcommand{\NRC}{\ensuremath{\NC\rig}\xspace}
\newcommand{\NFC}{\ensuremath{\NC\fle}\xspace}
\newcommand{\NRM}{\ensuremath{\mathbf{R}}\xspace}
\newcommand{\NR}{\ensuremath{\mathbf{P}}\xspace}
\newcommand{\NRR}{\ensuremath{\NR\rig}\xspace}
\newcommand{\NFR}{\ensuremath{\NR\fle}\xspace}
\newcommand{\NRRM}{\ensuremath{\NRM\rig}\xspace}
\newcommand{\NFRM}{\ensuremath{\NRM\fle}\xspace}
\newcommand{\NI}{\ensuremath{\mathbf{I}}\xspace}
\newcommand{\NV}{\ensuremath{\mathbf{V}}\xspace}
\newcommand{\NT}{\ensuremath{\mathbf{T}}\xspace}

\newcommand{\BC}{\ensuremath{\mathbf{B}}\xspace}
\newcommand{\BCr}{\ensuremath{\BC\rig}\xspace}
\newcommand{\BCf}{\ensuremath{\BC\fle}\xspace}
\newcommand{\AS}{\ensuremath{\mathbf{A}}\xspace}
\newcommand{\ASr}{\ensuremath{\AS\rig}\xspace}
\newcommand{\ASf}{\ensuremath{\AS\fle}\xspace}

\DeclareMathSymbol{\Gamma}{\mathalpha}{operators}{0}
\DeclareMathSymbol{\Delta}{\mathalpha}{operators}{1}
\DeclareMathSymbol{\Theta}{\mathalpha}{operators}{2}
\DeclareMathSymbol{\Lambda}{\mathalpha}{operators}{3}
\DeclareMathSymbol{\Xi}{\mathalpha}{operators}{4}
\DeclareMathSymbol{\Pi}{\mathalpha}{operators}{5}
\DeclareMathSymbol{\Sigma}{\mathalpha}{operators}{6}
\DeclareMathSymbol{\Upsilon}{\mathalpha}{operators}{7}
\DeclareMathSymbol{\Phi}{\mathalpha}{operators}{8}
\DeclareMathSymbol{\Psi}{\mathalpha}{operators}{9}
\DeclareMathSymbol{\Omega}{\mathalpha}{operators}{10}

\DeclareMathAlphabet{\mathcal}{OMS}{lmsy}{m}{n}
\DeclareMathAlphabet{\mathbcal}{OMS}{lmsy}{b}{n}

\usepackage{hyperref}
\hypersetup{final}

\theoremstyle{plain}
\newtheorem{theorem}{Theorem}[section]
\newtheorem{lemma}[theorem]{Lemma}
\newtheorem{corollary}[theorem]{Corollary}
\newtheorem{fact}[theorem]{Fact}
\newtheorem{claim}[theorem]{Claim}

\theoremstyle{definition}
\newtheorem{definition}[theorem]{Definition}
\newtheorem{example}[theorem]{Example}

\begin{document}

\title{Temporal Conjunctive Query Answering in the Extended \DLLite Family
}


\author{Stefan Borgwardt    \and
        Veronika Thost
}


\institute{Stefan Borgwardt \at
Institute of Theoretical Computer Science\\
Technische Universit\"at Dresden, Germany\\
\email{{stefan.borgwardt@tu-dresden.de}}
           \and
           Veronika Thost \at
           MIT-IBM Watson AI Lab\\
           IBM Research\\
           \email{{veronika.thost@ibm.com}}
}

\date{Received: date / Accepted: date}


\maketitle
\begin{abstract}
	Ontology-based query answering (OBQA) augments classical query answering in databases by domain knowledge encoded in an ontology. Systems for OBQA use the ontological knowledge to infer new information that is not explicitly given in the data. Moreover, they usually employ the open-world assumption, which means that knowledge \ctwo{that is not stated explicitly in the data and that is not inferred} is not assumed to be true or false. Classical OBQA however considers only a snapshot of the data, which means that information about the temporal evolution of the data is not used for reasoning and hence lost.
	
	We investigate temporal conjunctive queries (TCQs) that allow to access temporal data through classical ontologies. In particular, we study combined and data complexity of TCQ entailment for ontologies written in description logics from the extended \DLLite family. Many of these logics allow for efficient reasoning in the atemporal setting and are successfully applied in practice. 
	We show comprehensive complexity results for temporal reasoning with these logics.
\end{abstract}




\section{Introduction}\label{ch:introduction}

Ontologies play a central role in various applications: by linking data from
heterogeneous sources to high-level concepts and relations, they are used for
automated data integration and processing.
In particular, queries formulated in the abstract vocabulary of the ontology can
then be answered over all the linked datasets.
We focus on lightweight description logics as ontology languages, which are known to allow for efficient reasoning in the classical setting \ctwo{\cite{dllite05,dllrelations,KLTWZ-KR10:dll-combined,Ontop16}} and are successfully applied in practice \changed{\cite{KHS+-JWS17,KMM+-JWS17}}.
%
Well-known medical domain ontologies like 
GALEN%
\footnote{\url{http://www.co-ode.org/ontologies/galen}}
may, for example, capture the facts
that the varicella zoster virus~(VZV) is a virus,
that chickenpox is a VZV infection, and that a negative allergy test implies 
that no allergies are present,
by terminological axioms called \emph{concept inclusions~(CIs)}:
$$\mathsf{VZV}\sqsubseteq\mathsf{Virus},\
\mathsf{Chickenpox}\sqsubseteq\mathsf{VZVInfection},\
\mathsf{NegAllergyTest}\sqsubseteq\lnot\exists\mathsf{AllergyTo}.$$
Here, $\mathsf{Virus}$ is a \emph{concept name} that represents the set of all
viruses, and $\mathsf{AllergyTo}$ is a \emph{role name} that represents a binary
relation connecting patients to allergenes; $\exists\mathsf{AllergyTo}$
refers to the domain of this relation, \ie all patients with allergies.
A possible data source storing patient data is depicted in Figure~\ref{fig:patient-data}.
\begin{figure}%
	\centering
  \resizebox*{\textwidth}{!}{%
	\begin{tabular}{|c|c|}
		\hline &\\[-1.0em]
		PID&Name\\\hline &\\[-1.0em]
		1&Ann \\
		2&Bob \\
		3&Chris \\
		\hline
	\end{tabular}\hspace*{1pt}
	\begin{tabular}{|c|c|c|}
		\hline &&\\[-1.0em]
		PID&AllergyTest&Date\\\hline &&\\[-1.0em]
		1&neg&$16.01.2011$ \\
		2&pos&$06.01.1970$\\
		3&neg&$01.06.2015$ \\
		\hline
	\end{tabular}\hspace*{1pt}
	\begin{tabular}{|c|c|c|}
		\hline &&\\[-1.0em]
		PID&Finding&Date\\\hline &&\\[-1.0em]
		1&Chickenpox&$13.08.2007$\\
		2&VZV-Infection&$22.01.2010$\\
		3&
		VZV-Infection&$01.11.2011$\\
		\hline
	\end{tabular}}
	\caption{Example patient data}
	\label{fig:patient-data}
\end{figure}%
The data is linked to the ontology by mappings~\cite{PCDLLR-JoDS08:linking};
in our example, the tuple $(1,\text{Chickenpox},13.08.2007)$ can be encoded into the facts $\mathsf{HasFinding}(1,\texttt{x})$ and
$\mathsf{Chickenpox}(\texttt{x})$, \changed{where \texttt{x} is a fresh symbol representing the finding and $\mathsf{Chickenpox}$ is the type of this finding, which may be contained in a fact base $\afb_{13.08.2007}$ (i.e., the individual time points are days).}


Ontology-based query answering (OBQA) can then assist in finding appropriate participants for a clinical study,
\ctwo{by formulating the eligibility criteria} as queries over the mapped patient data.
The following are examples of inclusion and exclusion conditions for an existing clinical trial:%
\footnote{\url{https://clinicaltrials.gov/ct2/show/NCT01953900}}
\begin{itemize} 
	\item The patient should have been previously infected with VZV or previously
	vaccinated with VZV vaccine.
	\item The patient should not be allergic to VZV vaccine.
\end{itemize} 
Considering the first condition, OBQA augments standard query answering (\eg in
SQL) in that not only Bob and Chris, but also Ann would be considered as an
appropriate candidate.
However, in standard OBQA, we can neither express negation (\emph{not}) nor relate
several points in time (\emph{previously}), both of which are needed to
faithfully represent the given criteria.
In this article, we study temporal OBQA and allow negation in our query
language.


We consider the \emph{temporal conjunctive queries}~(TCQs) proposed by~\cite{BaBL-CADE13,BaBL-JWS15}, which combine conjunctive queries (CQs)
via the operators of
propositional linear temporal logic LTL~\cite{Pnueli77}.
\ctwo{The flow of time is represented by the sequence of natural numbers, \ie every \emph{point in time} (also \emph{time point} or \emph{moment}) is represented by one number.}
For example, the above criteria can be specified via the following TCQ,
to obtain all eligible patients~$x$:
\begin{align*}
& 
\left(\Diamondm \big(\exists 
y.\mathsf{HasFinding}(x,y)\land\mathsf{VZVInfection}(y)\big)\lor{} \right. 
\\&\left.\phantom{(}
\Diamondm \big(\exists 
y.\mathsf{VaccinatedWith}(x,y)\land\mathsf{VZVVaccine}(y)\big)\right)\land{}\\& 
\lnot \big(\exists 
y.\mathsf{AllergyTo}(x,y)\land\mathsf{VZVVaccine}(y)\big).
\end{align*}
We here use the temporal operator $\Diamondm$ (\enquote{at some time in the past}) and
consider the symbols $\mathsf{AllergyTo}$ and $\mathsf{VZVVaccine}$ to be
\emph{rigid}, which means that their interpretation does not change over time.
Hence, we assume someone having an allergy to VZV vaccine to
have this allergy for his or her life.

We focus on the problem of evaluating a TCQ w.r.t.\ a \emph{temporal knowledge
base} (TKB), which contains the domain ontology and a finite \emph{sequence}
of fact bases. Each fact base contains the data associated to a specific
point in time---from the past until the \emph{current time point} $n$ (``now'').
In contrast, the domain knowledge is assumed to hold \emph{globally}, meaning at every point in time.
In this setting, the information within the ontology and the fact bases does not
explicitly refer to the temporal dimension, but is written in a
\emph{classical} (atemporal) description logic (DL); only the query is temporal.

\subsection{\changed{Related Work}}


%
\begin{figure}[tb]
 \centering
 \begin{tikzpicture}[level distance=6ex]
	 \tikzstyle{level 1}=[sibling distance=10em]
	 \tikzstyle{level 2}=[sibling distance=8em]
	 \node (alci) {$\mathcal{ALC[H]I}$}
		 child{node (alc) {$\ALC[\Hmc]$}
			 child{node (el) {$\EL[\Hmc]$}}
		 }
		 child{node (dllb) {$\DLLite^{[\Hmc]}_\textit{bool}$}
			 child{node (dllh) {$\DLLite^{[\Hmc]}_\textit{horn}$}}
			 child{node (dllk) {$\DLLite^{[\Hmc]}_\textit{krom}$}}
		 };
	 \node[below=9ex of dllb] (dllc) {$\DLLite^{[\Hmc]}_\textit{core}$};
	 \draw
		 (dllc) -- (dllk)
		 (dllc) -- (dllh);
 \end{tikzpicture}
 \caption{Hierarchy of expressivity of several description logics considered in this article.
 Each logic can be augmented with role inclusions, denoted by~\Hmc.}
 \label{fig:dls}
\end{figure}
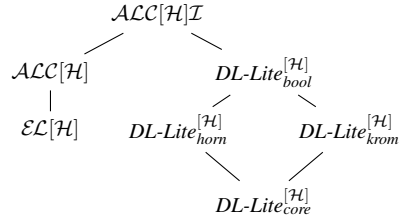

There are various ways to represent time in DL modeling; for example, by considering time points as concrete datatypes~\cite{BaH-IJCAI91:cds,Lutz-IJCAI01:concdomintervals} or formalisms inspired
by action logics~\cite{AF-JAIR98:actions,HCMGMF-JAIR13:actions}. Good overviews of different such approaches are provided in~\cite{AF-AMAI00:survey,AF-Handbook05}. 
We focus on \emph{temporal description logics} that are two-dimensional combinations of standard temporal logics with DLs, which is nowadays the common approach, 
though there is no formal definition of what a temporal description logic should look like.%
\footnote{Temporal description logics represent special kinds of combinations of DLs with modal logics~\cite{GKWZ03:manydimmodal}.} 
Such combinations still offer a wealth of degrees of freedom; for instance \wrt the base DL and temporal logic considered.
Figure~\ref{fig:dls} depicts various description logics that are relevant for this article, and their relations in terms of expressivity.
Earlier works investigate temporal versions of standard reasoning problems \wrt combined complexity and target applications such as terminologies with temporal aspects or temporal conceptual modeling~\cite{LuWZ-TIME08:survey}. In contrast, most recent investigations focus on temporal OBQA with the goal of accessing temporal data and also consider data complexity~\cite{AKK+-TIME17}. 

Schild proposed the first combination of a DL and a point-based temporal logic based on a two-dimensional semantics~\cite{Schild93}.
Subsequent studies have focused on classifying different combinations of LTL and (extensions of) \ALC according to expressivity and complexity, with results mostly in the range of \ExpSpace, if rigid roles are disregarded~\cite{GKWZ03:manydimmodal,AF-Handbook05, dlhandbook07}.
An important outcome of that research is the observation that rigid roles and other forms of temporal roles usually lead to undecidability.
Because decidability represents one major feature of DLs, research since then has been dedicated to the study of decidable temporal DLs. Lower and, in particular, tractable complexities are obtained by restricting the temporal logic or the DL component. 


\ctwo{In contrast to TCQs (which are temporalized queries), many of these logics support temporalizing either concepts, CIs, or facts.
Temporalized concepts allow to describe the temporal development of individuals, such as the collection of individuals that were vaccinated and, at the next ($\Next$) time point, had some allergic reaction: $$\exists\mathsf{VaccinatedWith}\sqcap\Next \exists\mathsf{HasFinding}.\mathsf{AllergicReaction}.%
\footnote{\ALC allows to qualify the existential restriction, \eg to restrict the range of the relation $\mathsf{HasFinding}$. Most logics of the \DLLite family can only express qualified existential restrictions on the right-hand side of CIs.}
$$
Temporalized facts can express, for example, that the patient with ID~$1$ did not have an allergic reaction since ($\Since$) the last vaccination: $$\big((\lnot\exists\mathsf{HasFinding}.\mathsf{AllergicReaction})(1)\big)\Since\big((\exists\mathsf{VaccinatedWith})(1)\big).$$
Temporalized CIs describe concept inclusions that only hold temporarily (instead of globally, \ie at all time points), \eg to describe changing policies using the operators \enquote{always in the past} ($\Boxm$) and \enquote{always in the future} ($\Boxf$):
\begin{align*}
	&\Boxm\big(\exists\mathsf{HasFinding}.\mathsf{Pneumonia}\sqsubseteq\exists\mathsf{Recommend}.\mathsf{Fluoroquinolone}\big)\land{}\\
	&\Next\Boxf\big(\exists\mathsf{HasFinding}.\mathsf{Pneumonia}\sqsubseteq\lnot\exists\mathsf{Recommend}.\mathsf{Fluoroquinolone}\big).
\end{align*}
}%
Without rigid symbols, temporalizing \ALC concepts or CIs does not lead to an increase in complexity from the \ExpTime complexity of satisfiability in \ALC.
%
However, rigid names lead to \ExpSpace-completeness~\cite{LuWZ-TIME08:survey}.


\ctwo{The logics} \ALC-LTL~\cite{BaGL-TOCL12}, \SHOQ-LTL \cite{phdmarcel}, \EL-LTL \cite{BoT-IJCAI15}, and \DLLite-LTL \cite{AKLWZ-TIME07:temporalising,phdthost} allow for combining DL axioms \ctwo{(CIs and facts)} via LTL operators, but no temporal concepts.
%
The setting where the concept inclusions are required to hold globally represents a simplified variant of the TCQ answering scenario, since assertions can be seen as simple CQs (without variables) and the fact bases are empty.
Even more, negated CQs can simulate non-global CIs in \ALC, one main feature of \ALC-LTL.
The satisfiability problem in \ALC-LTL has the same complexity as in the atemporal case, but rigid concepts and roles lead to \NExpTime and \TwoExpTime-completeness, respectively, and thus to a considerable increase in complexity. This increase can be overcome, at least for rigid concepts, if CIs are only allowed to occur globally \cite{BaGL-TOCL12}.
These results have been extended to \SHOQ-LTL~\cite{phdmarcel}.
Rigid names lead to \NExpTime-completeness even in \EL-LTL and \DLLite-LTL \cite{BoT-IJCAI15,phdthost}.

In another line of work, tractable DLs in combination with (subsets of) LTL have been investigated \cite{AKLWZ-TIME07:temporalising,AKRZ-TOCL14:cookbook}.
The KB consistency problem is investigated for temporal extensions of \DLLite that allow for temporalizing concepts, including rigid roles, with
%
%
positive results including containment in \NLogSpace or \PTime, obtained by reduction to 
fragments of propositional temporal logic, but only for formalisms strongly constrained on the temporal side \cite{AKRZ-TOCL14:cookbook}. In most cases, more variety in that direction leads to \NP-completeness and, if more expressive role expressions are allowed (e.g., arbitrary role inclusions), even to undecidability \cite{AKRZ-TOCL14:cookbook}.
%
For both \DLLite and \EL, the integration of temporalized concepts, axioms, and rigid roles yields complexities such as \PSpace for \DLLitekrom; \ExpSpace for \DLLitehorn and \DLLitebool; and even undecidability for \EL \cite{AKLWZ-TIME07:temporalising}.
%
These results are particularly interesting because, in the atemporal case, complexity results for \DLLitekrom are usually worse than corresponding ones for \DLLitehorn \cite{dllrelations}.

There are also recent works that temporalize concepts and/or axioms with temporal logic operators more expressive than LTL, such as from computation tree logic CTL \cite{GuJuL-ECAI12:ctl,GuJuS-KR14:marriage,GuJuS-IJCAI15:tdls-tboxes}
and metric temporal logics \cite{GuJO-ECAI16,BBK+-FroCoS17,Thost-KR18}.
However, this mostly results in quite high complexities and is out of the scope of our work.

Query answering with the goal of retrieving data 
is the focus of recent research in DLs in general, and this is reflected in the latest explorations on reasoning about temporal knowledge. The different works can be classified \wrt the considered ontologies, depending on whether they are also temporal 
or written in a classical DL. 
The former approaches offer more expressivity but, on the other hand, tend to lead to higher reasoning complexities. For this reason, they are usually studied w.r.t.\ lightweight DLs.


Different temporal extensions of \DLLite have been investigated with
the goal of first-order rewritability results for temporal OBQA \cite{AKKRWZ-IJCAI15:omtqs}. The queries are arbitrary combinations of temporalized concept and role atoms using the operators of first-order temporal logic and thus very expressive, which is why epistemic semantics is employed.\footnote{Queries with negation are not tractable, even in the atemporal setting~\cite{GIKK-JWS15}.} 
Although the setting studied is very restricted (roles are disregarded, only a subset of LTL is applied, the ABoxes contain a single individual name only), the results can, amongst others, be applied to show first-order rewritability of instance query answering in temporal \DLLite.
The rewritings are constructed based on temporal canonical interpretations for TKBs in these logics.
%
Also for TKBs based on \EL and allowing for both \Next and \Previous there are canonical models, which can be used to show that the satisfiability problem is tractable w.r.t.\ data complexity if rigid roles are disallowed \cite{GuJuK-IJCAI16:tel}. The latter paper also identifies a certain \emph{periodicity} of the ontology to ensure decidability and proposes several acyclicity notions for ontologies in temporal \EL that yield tractable combined and data complexity.
%

Research in the second direction (temporal query answering with classical ontologies) has been first considered in a general way in 
\cite{GuKla-RR12:approach} for expressive operators on both the temporal and the DL side.
A particular result shows first-order rewritability for \DLLitecore, but this is achieved by considering epistemic semantics for CQs \cite{Klarman-LPAR13:sql,KlaM-RR14:sql}.
{Temporal conjunctive queries} have first been studied for \ALC~\cite{BaBL-CADE13} and later for expressive extensions such as \SHOQ~\cite{BaBL-JWS15,BaBL-AI15}, focusing on the complexity of the entailment problem under open-world semantics, resulting in very high combined complexities. Moreover, for many of the considered DLs this problem is \coNP-complete in data complexity, as in the atemporal case, even in the presence of rigid concept names. 
Subsequent works have investigated TCQ answering and entailment w.r.t.\ the most prominent lightweight description logics \cite{BoLT-JWS15,THO-DL15,BoT-IJCAI15,BoT-GCAI15,Thos-WSP17}.
%
%
There are also some works on non-standard reasoning problems for TCQs that are different from query answering and entailment, such as ABox {abduction}~\cite{KlaM-LPAR13:streams,KlaM-DL14:abduction}, \ctwo{but these are less relevant to our work}. 
%


There are also recent proposals of interval temporal description logics that allow for tractable OBQA \cite{ABM+-ECAI14,
	AKRZ-AAAI15,KPP+-IJCAI16}, but 
their setting is rather different from the one we consider since there the basic units are intervals instead of time points. 

This article can be seen as an extension of the studies on TCQs \cite{BaBL-CADE13,BaBL-JWS15,BaBL-AI15,BoT-GCAI15} to the DLs of the extended \DLLite family.
In particular, it enhances previous results \cite{BoT-GCAI15,Thos-WSP17} by an improved presentation and full and revised versions of all proofs.
From another point of view, we extend a previously considered query language over temporal \DLLite knowledge bases \cite{BoLT-JWS15} with unrestricted negation in TCQs (and consider infinite instead of finite temporal semantics).
Our introductory example clearly shows that, even though most \DLLite logics can only express a weak form of negation, \ie disjointness constraints like $\mathsf{NegAllergyTest}\sqsubseteq\lnot\exists\mathsf{AllergyTo}$, the negation in the query language can be meaningfully used to query for negative information entailed by the ontology.
The impact of negation inside of CQs, \ie nested inside existential quantification, has been shown to cause a large increase in complexity even for atemporal \DLLite, even making it undecidable in very restricted cases \cite{GIKK-JWS15}.
In contrast, our extension is much more well-behaved (see Section~\ref{sec:contrib}), though it clearly does not come for free.
We also study the impact of allowing full negation in the ontology language, \ie the very expressive logic \DLLitebool.
An alternative, more tractable approach is to consider epistemic semantics for negation \cite{GuKla-RR12:approach,KlaM-RR14:sql}, similar to approaches for closed-world reasoning in DLs \cite{LuSW-IJCAI13,AhOS-IJCAI16}. This semantics makes quite strong assumptions on the input data, by presupposing that all missing information is indeed false, which is at odds with the open-world nature of standard DL semantics.
However, both closed-world and open-world semantics (or a combination of both) may be more or less suitable, depending on the application scenario.

\subsection{Contributions}
\label{sec:contrib}

In this article, we study TCQs over the lightweight DLs of the extended \DLLite family, which are depicted in Figure~\ref{fig:dls} and were tailored
for efficient (atemporal) query answering~\cite{dllfamily,Ontop16}. 
Of particular interest in this setting is 
the question to what extent ontology-based temporal query answering is \emph{first-order rewritable}, which means that the
queries can be rewritten into FO queries (\eg in SQL) over a database, and then can be executed using standard database systems. This is possible in the atemporal case for the logic \DLLitehhorn \cite{dllfamily,CDLLR-KR06:dcqadls,CDD+-AI13}.
\changed{Note that $\DLLite_\textit{R}$, the DL closest to the OWL~2 QL profile~\cite{owl2-profiles}, extends
	\DLLitehcore (a subset of \DLLitehhorn) only by disjointness axioms for roles. 
}
\changed{Although TCQ entailment over \DLLitehhorn ontologies turns out to be not first-order rewritable, certain parts of this problem can be solved using FO rewritings (see Section~\ref{sec:dlltcqs-r-sat-rew}).}
We also study related, but more expressive logics such as \DLLitehbool, where reasoning becomes harder~\cite{dllrelations}.

We investigate both combined and data complexity of TCQ entailment and, as usual, distinguish three different settings for the rigid symbols:
\begin{enumerate}[label=(\roman*)]
  \item no symbols are allowed to be rigid,
  \item only rigid concept names are allowed, and
  \item both concept names and role names can be rigid.
\end{enumerate}
As in \ALC-LTL~\cite{BaGL-TOCL12}, the fourth case is irrelevant since rigid
concepts~$C$ can be simulated by rigid roles~$R_C$ via two CIs 
$C\sqsubseteq\exists R_C$ and $C\sqsupseteq\exists R_C$.
Tables~\ref{tab:cc} and~\ref{tab:dc} summarize our results and compare them to
the baseline complexity of atemporal query answering \changed{(the logics are ordered by complexity)}.
\begin{table}[tb]
	\renewcommand*{\arraystretch}{1.2}
	\centering
	\tabulinesep=.7mm
		\caption{Combined complexity of TCQ entailment compared to standard CQ entailment. New results are highlighted
			in gray. The complexity of CQ entailment in
			$\DLLite^{[\Hmc]}_{\textit{krom}}$ is still open (marked by ?).}
		\label{tab:cc}
	\begin{tabu} to \linewidth {cX[1c]X[1c]X[1c]X[1c]} 
		\toprule
		& CQ & (i) & (ii) & (iii) \\
		\midrule
		$\DLLite_{[core\mid horn]}^{[\ |\Hmc]}$
		& \NP
			\source{\cite{dllfamily}}
		& \cellcolor{TabGray} \PSpace
			\source{$\ge$ \cite{SiCl85:ltlpspace}}
		& \cellcolor{TabGray} \PSpace
		& \cellcolor{TabGray} \PSpace
			\source{$\le$ Thm.~\ref{thm:dlltcqs-cc}}
		\\
		\EL
		& \NP
			\source{\cite{Rosa-DL07:elcqa}}
		& \PSpace
			\source{$\ge$ \cite{SiCl85:ltlpspace}}
		& \PSpace
			\source{$\le$ \cite{BoT-IJCAI15}}
		& \coNExpTime
			\source{\cite{BoT-IJCAI15}}
		\\
		$\ALC[\ |\Hmc]$
		& \ExpTime
			\source{\cite{Lutz-IJCAR08,EiOS-JCSS12}}
		& \ExpTime
		\source{$\le$ \cite{BaBL-JWS15}}
		& \coNExpTime
		\source{\cite{BaBL-JWS15}}
		& \TwoExpTime
		\source{\cite{BaBL-JWS15}}
		\\
		$\DLLite_{[krom\mid bool]}$
		& ? $|$ \ExpTime
			\source{\cite{BoMMP-TODS16}}
		& \cellcolor{TabGray} \ExpTime
			\source{$\ge$ Cor.~\ref{cor:dlltcqs-krom-cc:lb-w/o-rigid}}
			\source{$\le$ Thm.~\ref{thm:dlltcqs-krombool-cc:ub-w/o-rigid}}
		& \cellcolor{TabGray} \coNExpTime
			\source{$\ge$ Thm.~\ref{thm:dlltcqs-krom-cc:lb-rigid-concepts}}
			\source{$\le$ Thm.~\ref{thm:dlltcqs-krombool-cc:ub-rigid-concepts}}
		& \cellcolor{TabGray} \TwoExpTime
			\source{$\ge$ Thm.~\ref{thm:krom-cc:lb-rigid-roles}}
		\\
		$\DLLite_{[krom\mid bool]}^\Hmc$
		& ? $|$ \TwoExpTime
			\source{\cite{BoMMP-TODS16}}
		& \cellcolor{TabGray} \TwoExpTime
			\source{$\ge$ Cor.~\ref{cor:dlltcqs-hkrom-cc:lb-w/o-rigid}}
		& \cellcolor{TabGray} \TwoExpTime
		& \cellcolor{TabGray} \TwoExpTime
			\source{$\le$ \cite{BaBL-AI15}}
		\\
		$\ALC[\ |\Hmc]\Imc$
		& \TwoExpTime
			\source{$\ge$ \cite{Lutz-IJCAR08}}
			\source{$\le$ \cite{CaEO-IC14}}
		& \TwoExpTime
		& \TwoExpTime
		& \TwoExpTime
			\source{$\le$ \cite{BaBL-AI15}}
		\\
		\bottomrule
	\end{tabu}
\end{table}
\begin{table}[tb]
	\renewcommand*{\arraystretch}{1.2}
	\centering
	\tabulinesep=.7mm
		\caption{Data complexity of TCQ entailment compared to standard CQ entailment. New results are highlighted in
			gray.}
		\label{tab:dc}
	\begin{tabu} to \linewidth {cX[1c]X[1c]X[1c]X[1c]} 
		\toprule
		& CQ & (i) & (ii) & (iii) \\
		\midrule
		$\DLLite_{[core\mid horn]}^{[\ |\Hmc]}$
		& \ACzero \source{\cite{CDD+-AI13}}
		& \cellcolor{TabGray} \ALogTime
		\source{$\ge$ Thm.~\ref{thm:dlltcqs-dc:lb}}
		& \cellcolor{TabGray} \ALogTime
		& \cellcolor{TabGray} \ALogTime
		\source{$\le$ Thm.~\ref{thm:dlltcqs-dc-alogtime}}
		\\
		\EL
		& \PTime
		\source{$\ge$ \cite{CDD+-AI13}}
		\source{$\le$ \cite{KrLu-LPAR07:eldc,Rosa-ICDT07}}
		& \PTime
		\source{$\le$ \cite{BoT-IJCAI15}}
		& \coNP
		\source{$\ge$ \cite{BoT-IJCAI15}}
		& \coNP
		\source{$\le$ \cite{BoT-IJCAI15}}
		\\
		$\DLLite_{[krom\mid bool]}^{[\ |\Hmc]}$
		& \coNP
		\source{$\ge$ \cite{CDD+-AI13}}
		& \cellcolor{TabGray} \coNP
		& \cellcolor{TabGray} \coNP
		& \cellcolor{TabGray} \coNP
		\\
		$\ALC[\ |\Hmc][\ |\Imc]$
		& \coNP
		\source{$\ge$ \cite{CDD+-AI13}}
		\source{$\le$ \cite{OrCE-JAR08}}
		& \coNP
		& \coNP
		\source{$\le$ \cite{BaBL-AI15}}
		& \cellcolor{TabGray} \coNP
		\source{$\le$ Thm.~\ref{thm:dc-dllitehkrom}}
		\\
		\bottomrule
	\end{tabu}
\end{table}
On the one hand, for expressive members of the extended \DLLite family, we
obtain
complexities similar to those for very expressive DLs such as $\mathcal{ALCHI}$. 
%
In data complexity, there is not even a difference between the lightweight DL \EL and 
$\mathcal{ALCHI}$ if rigid symbols are considered.
On the other hand, for the logics below \DLLitehhorn, we get results that are even
better than those for \EL; interestingly, here rigid names do not affect the
complexity.
The \ALogTime-hardness result for the data complexity of TCQ
entailment in \DLLitecore shows however that it is not possible to find a (pure)
first-order rewriting of TCQs in this setting. 
%
%
Nevertheless, our analysis gives hope for an efficient implementation of temporal query answering in \DLLitehhorn.

The article is structured as follows.
In Section~\ref{ch:prelims}, we recall the preliminaries and describe a general approach for solving TCQ entailment~\cite{BaBL-CADE13}; it is based on splitting the problem into separate problems in LTL and in (atemporal) DLs.
In Section~\ref{sec:dlltcqs-r-sat-charact}, we propose a characterization of the DL part of the TCQ satisfiability problem that is tailored to \DLLitehhorn.
We use this characterization in Section~\ref{sec:dlltcqs-cc} to obtain the \PSpace combined complexity result.
Based on our characterization, we then also show that parts of the  TCQ entailment problem are first-order rewritable (see Section~\ref{sec:dlltcqs-r-sat-rew}).
This, in turn, allows us to develop an algorithm that proves membership in \ALogTime in data complexity in Section~\ref{sec:dlltcqs-dc}.
Section~\ref{sec:exdlltcqs} covers \DLLite logics beyond the Horn fragments.

%


\section{Preliminaries}
\label{ch:prelims}
	
%
We first recall description logics of the \DLLite family, conjunctive
queries, linear temporal logic, and their combination into temporal conjunctive
queries. Then, we discuss approaches for solving the classical
(\emph{atemporal}) and the temporal query entailment problems.

\subsection{\DLLite}
\label{sec:prelims-dls}
\label{sec:prelims-dls-syntax}


Description logics can be seen as fragments of first-order logic.
A DL \emph{signature} $(\NI,\NC,\NR)$ contains three sorts of
non-logical symbols representing constants, and unary and binary predicates,
respectively: \emph{individual names}~\NI, \emph{concept names (primitive concepts)}~\NC, and
\emph{role names (primitive roles)}~\NR, which are countably infinite, non-empty, pairwise disjoint sets.
In the following, we fix such a signature.

\begin{definition}[Syntax of \DLLite]\label{def:dll-syntax}
\emph{Roles} and \emph{(basic) concepts} are
defined, respectively, by the following rules, where $A\in\NC$ and $P\in\NR$:
\begin{align*}
   R &::=P \mid P^-,
 & B &::= A \mid \exists R.
\end{align*}
The sets of all roles and basic concepts are denoted by~\NRM and~\BC,
respectively.

\emph{Axioms} are the following kinds of expressions:
\emph{concept inclusions}~(CIs) of the form
\begin{equation}
\label{eq:ci}
B_1\sqcap\dots\sqcap B_m\sqsubseteq B_{m+1}\sqcup\dots\sqcup B_{m+n}
\end{equation}
where $B_1,\ldots,B_{m+n}\in\BC$;
\emph{role inclusions}~(RIs) of the form $S\sqsubseteq R$, where $R,S\in\NRM$;
and
\emph{assertions} of the form $B(a)$, $\lnot B(a)$, $R(a,b)$, or $\lnot R(a,b)$,
where $B\in\BC$, $R\in\NRM$, and $a,b\in\NI$.
An \emph{ontology} is a finite set of concept and role inclusions, and an
\emph{ABox} is a finite set of assertions.
Together, an ontology~\aont and an ABox~\afb form a \emph{knowledge base}~(KB)
$\akb:=\aont\cup\afb$, also written as $\akb=\langle\aont,\afb\rangle$.
The set of all assertions is denoted by~\AS.
\end{definition}

We distinguish several members of the extended \DLLite family as presented in~\cite{dllrelations}.
For $c\in\{\textit{core},\textit{horn},\textit{krom},\textit{bool} \}$, we
denote by \textit{DL-Lite}$^\Hmc_c$ the logic that restricts the concept inclusions~\eqref{eq:ci} as follows:
\begin{itemize}
\item if $c=\textit{core}$, then $m+n\le2$ and $n\le 1$;
\item if $c=\textit{horn}$, then $n\le 1$;
\item if $c=\textit{krom}$, then $m+n\le2$;
\item if $c=\textit{bool}$, there are no restrictions.
\end{itemize}
We also consider the sublogics \textit{DL-Lite}$_c$ that further disallow role
inclusions.
We use the term \DLLite for a generic member of this family of logics.
%

We use the generic notation $\Xbb(\changed{\Ymc})$ to denote those elements of~\Xbb that
can be built from the names occurring in~\changed{\Ymc}.
For example, if \Xbb is the set~\BC of basic concepts and \changed{\Ymc} is a KB~\Kmc,
then $\BC(\Kmc)$ denotes the basic concepts that can be built from the concept
and role names occurring in~\Kmc (either in concept inclusions or assertions).
Similarly, $\NI(\Amc)$ simply denotes the set of individual names occurring in
the ABox~\Amc;
we apply this notation more generally also to ontologies and single axioms, and
later to queries, temporal queries, and so on.
%
%
%


The semantics is specified in a model-theoretic way, based on interpretations.
\begin{table}[tb]
\centering
\caption{Semantics of roles, concepts, and axioms for an interpretation $\aint=(\Delta^\aint,\cdot^\aint)$.}
\label{tab:dl-semantics}
\begin{tabu}{|l|l|l|}
	\tabtopline
	Name&Syntax&Semantics
	\tabmidline
	inverse role\index{operator!DL|textbf}\index{inverse role|textbf}&$R^-$&$\{(y,x)\in\Delta^\aint\times\Delta^\aint\mid(x,y)\in R^\aint\}$
	\tabmidline
	existential restriction
	&$\exists R$&$\{x\in\Delta^\aint \mid\exists y\in \Delta^\aint:(x,y)\in R^\aint\}$
	\tabmidline
	concept inclusion&$B_1\sqcap\dots\sqcap B_m\sqsubseteq B_{m+1}\sqcup\dots\sqcup B_{m+n}$&$B_1^\aint\cap\dots\cap B_m^\aint\subseteq B_{m+1}^\aint\cup\dots\cup B_{m+n}^\aint$\\
	role inclusion&$R\sqsubseteq S$&$R^\aint\subseteq S^\aint$
	\tabmidline
	concept assertion&$B(a)$&$a^\aint\in B^\aint$\\
   negated concept assertion&$\lnot B(a)$&$a^\aint\notin B^\aint$\\
	role assertion&$R(\indone,\indtwo)$&$(\indone^\aint, \indtwo^\aint)\in R^\aint$ \\
   negated role assertion&$\lnot R(\indone,\indtwo)$&$(\indone^\aint, \indtwo^\aint)\notin R^\aint$
	\tabbotline
\end{tabu}
\end{table}

\begin{definition}[Semantics of \DLLite]\label{def:semantics}
An \emph{interpretation} $\aint=(\Delta^\aint,\cdot^\aint)$ consists of a
non-empty set~$\Delta^\aint$, the \emph{domain} of~\aint, and an
\emph{interpretation function}~$\cdot^\aint$, which assigns
to every $A\in\NC$ a set $A^\aint\subseteq\Delta^\aint$,
to every $P\in\NR$ a binary relation
$P^\aint\subseteq\Delta^\aint\times\Delta^\aint$, and
to every $a\in\NI$ an element
$a^\aint\in\Delta^\aint$ such that, for all $a,b\in\NI$ with $a\neq b$, we have
$a^\aint\neq b^\aint$ (\emph{unique name assumption}; UNA).
This function is extended to all roles and concepts as described in the first
part of Table~\ref{tab:dl-semantics}.
An interpretation~\aint is a \emph{model} of an axiom~\aaxiom, if the
corresponding condition given in Table~\ref{tab:dl-semantics} is satisfied.
It is a \emph{model} of a knowledge base~\akb, if it is a model of all axioms
contained in it.
\end{definition}
Following the standard notation for first-order logic, we denote the fact that
\aint is a model of~\akb by $\aint\models\akb$, and in this case also say that
\aint \emph{satisfies}~\akb.
Further, \akb is \emph{consistent} (or \emph{satisfiable}) if it has a model,
and \emph{inconsistent} (or \emph{unsatisfiable}) otherwise.
\akb \emph{entails} an axiom~\aaxiom, written $\akb\models\aaxiom$, if all
models of~\akb also satisfy~\aaxiom.
%
%
Two KBs are \emph{equivalent} if they have the same models.
This terminology and notation for \akb is extended to axioms, ontologies, and
ABoxes by viewing each as a (singleton) KB.
Moreover, we freely apply these terms to any \enquote{model of} relation
that we define in the following (for queries, temporal queries, etc.).

Given two domain elements $\elone, \eltwo$, a role $R$, and an interpretation~\aint such
that $(\elone,\eltwo)\in R^\aint$, we say that \elone is an
$R$-predecessor\index{predecessor} of~\eltwo, and \eltwo is an
$R$-successor\index{successor} of~\elone.
We use the terms \enquote{(domain) elements} and \enquote{individuals}
interchangeably for the elements of~$\Delta^\aint$. We call them \enquote{named}
if they are used to interpret individual names.
%
In concept inclusions of the form~\eqref{eq:ci} (see Definition~\ref{def:dll-syntax}), we
denote the empty conjunction by~$\top$ and the empty disjunction by~$\bot$, \ctwo{which are interpreted as $\Delta^\Imc$ and $\emptyset$, respectively} (\cf Table~\ref{tab:dl-semantics}).
We use the abbreviation
$\bigsqcap\Bmc$ for the conjunction $B_1\sqcap\dots\sqcap B_m$ if $\Bmc=\{B_1,\dots,B_m\}$, and set $(P^-)^-:=P$ for all $P\in\NR$.
We assume every KB to be such that all concept and role names occurring in the ABox also occur in the ontology.
\begin{definition}[Syntax of CQs]\label{def:cq-syntax}
	Let \NV be a \ctwo{countably infinite set of \emph{variables} disjoint from \NI, \NC, and \NR}, and $\NT:=\NI\cup\NV$ be the set of \emph{terms}.
	A \emph{conjunctive query}~(CQ) is of the form
   $\acq=\exists y_1,\ldots,y_m.\psi$, where $y_1,\dots,y_m\in\NV$ and $\psi$ is
   a (possibly	empty) finite conjunction ($\land$) of atoms of the form:
	\begin{itemize}
		\item $A(t)$ (\emph{concept atom}\index{atom}) with $A\in\NC$ and $t\in\NT$; or
		\item $P(\tone,\ttwo)$ (\emph{role atom}\index{atom}) with $P\in\NR$ and $\tone,\ttwo\in\NT$.
	\end{itemize}
   A \emph{union of conjunctive queries}\index{union of conjunctive queries!DL} (UCQ)\index{UCQ|see{union of conjunctive queries}} is a 
		disjunction ($\lor$) of CQs \ctwo{with the same free variables}. 
\end{definition} \changed{
In general, CQs may contain free variables, also called \emph{answer variables}.
However, without loss of generality, and unless stated otherwise, in the following we} assume all (U)CQs to be Boolean, \ie that all
variables are existentially quantified in the CQs. \changed{We sometimes stress this again, but actually make only one exception (in Section~\ref{sec:rcomplete-rewriting-details}). }
For ease of presentation, we sometimes treat a CQ as a set, thereby meaning the set of all of its atoms.

\begin{definition}[Semantics of CQs] \label{def:cqs-semantics}
	A mapping $\pi\colon\NT(\acq)\to\Delta^\aint$ is a homomorphism\index{homomorphism}
	of a 
	CQ~\acq into an interpretation \aint if
	\begin{itemize}
		\item $\pi(a)=a^\aint$ for all $a\in\Ind(\acq)$,
		\item $\pi(t)\in A^\aint$ for all concept atoms $A(t)$ in $\acq$, and
		\item $(\pi(\tone),\pi(\ttwo))\in P^\aint$ for all role atoms $P(\tone,\ttwo)$ in~$\acq$.
	\end{itemize}
	An interpretation~\aint is a \emph{model} of~\acq if there is such a homomorphism, and is a \emph{model} of a UCQ if it satisfies one of its disjuncts.
\end{definition}

We also allow basic concept atoms of the form $\exists R(x)$ to occur in CQs,
since such an atom can be expressed via a role atom $R(x,y)$ using a fresh,
existentially quantified variable~$y$. Similarly, inverse role atoms $R^-(x,y)$
can be expressed by $R(y,x)$.

\changed{
For the interested reader, in Appendix~\ref{app:prelim} we introduce additional notions about \DLLite and the temporal logic LTL that are relevant for our proofs.
}

\subsection{Temporal Conjunctive Queries}
\label{sec:tqa}
\label{sec:tqa-tqs}
\emph{Temporal conjunctive queries} are a temporal query language introduced in \cite{BaBL-CADE13}. 
%
They are basically formulas of LTL, but the variables are replaced by CQs
%
and the semantics is suitably lifted from sequences of propositional worlds to sequences of DL interpretations.
The flow of time is represented by the sequence of natural numbers, \ie every \emph{point in time} (also \emph{time point} or \emph{moment}) is represented by one number.
%
We additionally assume that a subset of the concept and role names is designated as being \emph{rigid}.
The intuition is that the interpretation of rigid names does not change over time. All individual names are implicitly assumed to be rigid, \ie to refer to the same domain element at all time points.
If a concept (axiom) contains only rigid symbols, then we call it a \emph{rigid} concept (axiom).
We hence extend the signature \adlsignature by two sets $\NRC\subseteq\NC$ and $\NRR\subseteq\NR$ of rigid concept names and rigid role names, respectively. The elements of $\NFC:=\NC\setminus\NRC$ and $\NFR:=\NR\setminus\NRR$ are called \emph{flexible}.
We denote the sets of rigid roles, basic concepts, and assertions by~\NRRM,
\BCr, and~\ASr, respectively, and their complements by~\NFRM, \BCf, and~\ASf,
respectively.

In this temporal setting, the knowledge base contains a global ontology that holds at all time points, as well as a series of ABoxes describing a finite sequence of initial time points~\cite{BaBL-CADE13}.

\changed{
\begin{definition}[Syntax of TKBs]
	%
	A \emph{temporal knowledge base}~(TKB)
	$\atkb=\langle\aont,(\afb_i)_{0\le i\le n}\rangle$ consists of an
	ontology~\aont and a non-empty, finite sequence of ABoxes~$\afb_i$,
	$i\in[0,n]$.
	%
\end{definition}

As mentioned above, the semantics is given by sequences of DL interpretations.

\begin{definition}[Semantics of TKBs]
	\label{def:tqa-tint}\label{def:tkb-syntax}\label{def:tkb-semantics}
	An infinite sequence $\atint=(\aint_i)_{i\ge0}$ of interpretations
	$\aint_i=(\Delta,\cdot^{\aint_i})$ is a  \emph{DL-LTL structure} if it
	\emph{respects rigid names}, \ie we have
	$\asymbol^{\aint_i}=\asymbol^{\aint_j}$ for all
	$\asymbol\in\NI\cup\NRC\cup\NRR$ and $i,j\ge0$.
	%
	%
   %
	Such an interpretation
	\atint is a \emph{model} of a TKB \atkb if we have $\Imc_i\models\Omc$ for all
	$i\ge0$, and $\Imc_i\models\Amc_i$ for all $i\in[0,n]$.
\end{definition}
}
%
Observe that the interpretations in a DL-LTL structure share a single domain
(\emph{constant domain assumption}).
\ctwo{Similarly, we say that} any finite collection of interpretations $\aint_1,\dots,\aint_\ell$
\emph{respects rigid names} if they have the same domain and agree on the
interpretation of all rigid symbols.
As with atemporal KBs, we assume all concept and role names occurring in some
ABox of a TKB to also occur in its ontology.


As outlined above, TCQs combine conjunctive queries via LTL operators. We again restrict our focus to Boolean queries.
\begin{definition}[Syntax of TCQs]\label{def:tqs-syntax}
	The set of \emph{temporal conjunctive queries} (TCQs) is defined as follows, where $\acq$ is a CQ: $$\atcq,\atcqtwo::= \acq\mid 
	\neg\atcq\mid\atcq\land\atcqtwo\mid
	\Next\atcq\mid\Previous\atcq\mid
	\atcq\Until\atcqtwo\mid\atcq\Since\atcqtwo.
	$$
	A \emph{CQ literal} is of the form \acq (\emph{positive} CQ literal) or $\lnot\acq$
	(\emph{negative} CQ literal), for a CQ \acq.
\end{definition}
\begin{table}[tb]
	\centering
   \caption{Definitions of derived temporal operators.}
	\label{tab:tcqs-derived-ops}
	\begin{tabu}{|l|l|l|}
		\tabtopline
		Operator &Definition&Name 
		\tabmidline
      \true & $\acq\lor\lnot\acq$ for some CQ~\acq & tautology \\
      \false & $\lnot\true$ & contradiction \\
  		$\atcq\lor\atcqtwo$ & $\lnot(\lnot\atcq\land\lnot\atcqtwo)$ & disjunction\\ $\atcq\rightarrow\atcqtwo$ & $\lnot\atcq\lor\atcqtwo$ & implication \\
      $\atcq\leftrightarrow\atcqtwo$ & $(\atcq\rightarrow\atcqtwo)\land(\atcqtwo\rightarrow\atcq)$ & bi-implication \\
  		$\Next^i\atcq$ & $\Next\dots\Next\atcq$ ($i$ times) & iterated next \\
  		$\Previous^i\atcq$ & $\Previous\dots\Previous\atcq$ ($i$ times) & iterated previous \\
  		$\Diamondf\atcq$ & $\true\Until\atcq$ & eventually (at some time in the future) \\
  		$\Boxf\atcq$ & $\lnot\Diamondf\lnot\atcq$ & globally (always in the future) \\
  		$\Diamondm\atcq$ & $\true\Since\atcq$ & once (at some time in the past) \\
  		$\Boxm\atcq$ & $\lnot\Diamondm\lnot\atcq$ & historically (always in the past)
		\tabbotline
	\end{tabu}
\end{table}
%
The operators \Next and \Previous are called ``next'' and ``previous'', respectively.
The formula $\atcq\Until\atcqtwo$ stands for ``$\atcq$ until $\atcqtwo$'', and $\atcq\Since\atcqtwo$ is read ``$\atcq$ since $\atcqtwo$''.
The operators \Next and $\Until$ are the \emph{future operators}, and \Previous and $\Since$ are the \emph{past operators}. Together, they represent the \emph{temporal operators}.
Further, derived operators are defined in Table~\ref{tab:tcqs-derived-ops}.


%
\begin{definition}[Semantics of TCQs] \label{def:tcqs-semantics}
%
%
	For a given DL-LTL structure $\atint=(\aint_i)_{i\ge0}$, an $i\ge0$, and a
  Boolean TCQ~\atcq, the satisfaction relation $\atint,i\models\atcq$ is as defined in Table~\ref{tab:tcq-semantics}.

\changed{
  \atcq is \emph{satisfiable} w.r.t.\ a TKB $\atkb=\langle\aont,(\afb_i)_{0\le i\le n}\rangle$ if there is a model~\atint of~\atkb such that $\atint,n\models\atcq$, \ie where \atcq is satisfied at the \emph{current time point}~$n$. 
  Similarly, \atcq is \emph{entailed} by~\atkb{} \ctwo{if we have $\atint,n\models\atcq$ for all models $\atint\models\atkb$.}
}
\end{definition}
%
%
\begin{table}[tb]
	\centering
	\caption{Semantics of TCQs for a DL-LTL structure
    $\atint=(\aint_i)_{i\ge0}$.}
	\label{tab:tcq-semantics}
	\begin{tabu}{|l|l|}
		\tabtopline
		TCQ \atcq & Condition for $\atint,i\models\atcq$
		\tabmidline
		CQ \acq&$\aint_i\models\acq$\\
		$\neg\atcq$&$\atint,i\not\models\atcq$\\
		$\atcq\land\atcqtwo$&$\atint,i\models\atcq$ and $\atint,i\models\atcqtwo$\\
		$\Next\atcq$&$\atint,i+1\models\atcq$\\
		$\Previous\atcq$&$i>0$ and $\atint,i-1\models\atcq$\\
		$\atcq\Until\atcqtwo$&
		there is a $k\ge i$, such that $\atint,k\models\atcqtwo$ and, for all $j$, $i\le j< k$, we have
		$\atint,j\models\atcq$
		\\
		$\atcq\Since\atcqtwo$&there is a $k$, $0\le k\le i$, such that $\atint,k\models\atcqtwo$ and, for all $j$, $k< j\le i$, we have
		$\atint,j\models\atcq$
		\tabbotline
	\end{tabu}
\end{table}

\changed{
The satisfiability and entailment problems are mutually reducible.
Indeed, \atcq is \emph{not} entailed by \atkb iff the TCQ $\lnot\atcq$ is satisfiable w.r.t.~\atkb, and vice versa.
Hence, the complexity of both problems is always complementary.
Since the techniques we use apply to satisfiability, we will consider this problem for most of the technical development.
However, we formulate our results in terms of entailment since that problem is more interesting from a practical point of view.}
Hence, in the main part of this article, we investigate the \emph{satisfiability
problem of TCQs \wrt TKBs}, \ie the problem of finding a common model of both.
%

We denote by~\tcqcqs\atcq the CQs in the Boolean TCQ~\atcq and assume
without loss of generality that these CQs use disjoint sets of variables.
We further assume that TCQs contain only individual names that occur in the
ABoxes, and only concept and role names that occur in the ontology; this is
clearly without loss of generality, especially because of the assumption that
all concept and role names occurring in a TKB occur in its ontology.
We further assume all CQs~\acq to be connected, \ie that all
elements in $\NT(\acq)$ can be reached from the others via a series of role atoms in~\acq. \changed{This assumption is also without loss of generality, because a disconnected CQ can be split into a conjunction of several CQs, which is a special kind of TCQ; see~\cite{BaBL-CADE13,BaBL-JWS15} for details.}

We often consider TCQs~\atcq that do not contain temporal operators; for example,
UCQs or conjunctions of CQ literals. In this case, the satisfaction of \atcq in
a DL-LTL structure $\atint= (\aint_i)_{i\ge0}$ at time point~$i$ only depends on
the interpretation~$\aint_i$. For simplicity, we then often write
$\aint_i\models\atcq$ instead of $\atint,i\models\atcq$.
In this context, it is also sufficient to consider classical knowledge
bases~$\langle\aont,\afb\rangle$, which can be viewed as TKBs with a single
ABox.

We investigate the influence of the different \DLLite logics on the combined and data complexity of TCQ entailment. For combined complexity, the size of all the input is taken into consideration (i.e., the size of both the query and the entire KB), while data complexity only refers to the size of the data \cite{Vardi-STOC82:dc} (\ie the number and size of the ABoxes).

\subsection{Reasoning with TCQs}\label{sec:on-ubs}

We recall the general approach to decide TCQ satisfiability from~\cite{BaBL-CADE13}, where the problem of finding a common model $(\Imc_i)_{i\ge 0}$ of a TCQ~\atcq and a TKB~$\atkb=\langle\aont,(\afb_i)_{0\le i\le n}\rangle$ is split into an LTL satisfiability problem and several DL satisfiability problems \cite[Lemma~4.7]{BaBL-JWS15}.

The former considers the \emph{propositional abstraction}~$\pa{\atcq}$ of~\atcq, which is obtained from~\atcq by replacing the CQs $\aalpha_1,\dots,\aalpha_m\in\tcqcqs{\atcq}$ by propositional variables $p_1,\dots,p_m$, respectively.
The idea is that the world~$w_i$ in an LTL model $(w_i)_{i\ge 0}$ of~$\pa{\atcq}$ determines which CQs from~$\tcqcqs\atcq$ should be satisfied at time point~$i$.
To obtain an interpretation~$\Imc_i$ from~$w_i$, we have to check the satisfiability of the CQ literals 
that are induced by~$w_i$, where $\overline{w_i}:=\{p_1,\dots,p_m\}\setminus w_i$ denotes the complement of~$w_i$.
For this, it is enough to consider the atemporal KB
$\langle\aont,\afb_i\rangle$, where we assume that $\afb_i=\emptyset$ whenever
$i>n$.
However, the problem is that independent satisfiability tests for each time point are not enough, since the DL interpretations should also respect the rigid names.

Hence, we need to connect the LTL and DL satisfiability problems more closely.
For this, we consider a set $\as=\{\ax_1,\dots,\ax_k\}\subseteq 2^{\pv}$ of possible worlds,%
\footnote{\ctwo{In the following, we denote propositional worlds both by lower case~$w$ and by upper case~$\ax$ (with indices). Usually, the (finitely many) elements of \as are denoted by~$\ax_i$, and the (infinitely many) elements of an LTL structure by~$w_i$.}}
which is given as input to both problems.
Intuitively, these are the worlds that are allowed to occur in the LTL model.
Moreover, for each of the initial time points $0,\dots,n$, we fix one of these worlds, via a mapping $\iota\colon[0,n]\to[1,k]$.
%
The idea is to simultaneously look for models of the conjunctions
\[ \chi_{\iota(i)} := \bigwedge_{p_j\in\ax_{\iota(i)}}\acq_j \land
  \bigwedge_{p_j\in\overline{\ax_{\iota(i)}}}\lnot\acq_j \]
\wrt the atemporal KBs $\langle\aont,\afb_i\rangle$.
For the time points after~$n$, we do not have to consider ABoxes; however, for ease of presentation, we set $\afb_{n+i}:=\emptyset$ and $\iota(n+i):=i$ for all $i\in[1,k]$, which means that we artificially extend the ABox sequence to cover $n+k+1$ \enquote{time points}: $n+1$ time points for the worlds associated to an input ABox, and $k$ additional time points for all worlds in the model that are not influenced by the input ABoxes.

The LTL part is characterized by \emph{temporal satisfiability
(t-satisfiability)}, and \emph{rigid satisfiability (r-satisfiability)}
summarizes the DL part.
\begin{definition}[t-satisfiable]
	\label{def:tcqs-t-sat}\label{def:tqa-t-sat}
	The LTL formula $\pa{\atcq}$ is \emph{t-satisfiable}\index{t-satisfiability}
	w.r.t.~\as and~$\iota$ if there is an LTL structure $\atint=(w_i)_{i\geq 0}$
	such that $w_i\in\as$ for all $i\ge 0$, $w_i=\ax_{\iota(i)}$ for all
	$i\in[0,n]$, and $\atint,n\models\pa{\atcq}$.
\end{definition}
\begin{definition}[r-satisfiable]
	\label{def:tqa-r-sat}
	The set \as is \emph{r-satisfiable}\index{r-satisfiability} w.r.t.\ $\iota$
	and~\atkb iff there are interpretations $\Jmc_0,\ldots,\Jmc_{n+k}$ as follows:
	\begin{itemize}
		\item the interpretations share the same domain and respect rigid names,
		\item for all $i\in[0,n+k]$, $\Jmc_i$ is a model of~\aont, $\afb_i$,
		and~$\chi_{\iota(i)}$.
	\end{itemize}
\end{definition}
%

The satisfiability of~\atcq w.r.t.~\atkb can then be decided by combining the
above definitions.
\begin{lemma}[{see \cite[\citelem4.7]{BaBL-JWS15}}]
	\label{lem:tcq-sat-iff}
	A TCQ~\atcq has a model w.r.t.\ a TKB~\atkb iff there exist a set
	$\as=\{\ax_1,\dots,\ax_k \}\subseteq 2^{\pv}$ and a mapping
	$\iota\colon[0,n] \to[1,k]$ such that
	$\pa{\atcq}$ is t-satisfiable w.r.t.~\as and~$\iota$, and
	\as is r-satisfiable w.r.t.~$\iota$ and~\atkb.
\end{lemma}
The original proof in~\cite{BaBL-JWS15} considers the DL \SHQ, but it is independent of the description logic under consideration, and hence also applies in our setting.
%
%
This result shows that TCQ satisfiability can be split into the three
subproblems of
\begin{enumerate}[label=\textnormal{(\roman*)}]
	\item\label{def:tqa-task-obtaining}
	obtaining \as and $\iota$,
	\item\label{def:tqa-task-ltl}
	solving the LTL satisfiability test (t-satisfiability), and
	\item\label{def:tqa-task-dl} solving the DL satisfiability test(s)
	(r-satisfiability).
\end{enumerate}
However, for solving these problems in \DLLitehhorn, we cannot in general apply the existing methods from~\cite{BaBL-CADE13,BaBL-JWS15,BaBL-AI15} since we want to show considerably lower complexity bounds, namely \ALogTime in data complexity and \PSpace in combined complexity.
The linear size of~$\iota$ is an obstacle for designing algorithms of sublinear
complexity, and the exponential size of \as makes it impossible to guess (and
store) this set using only a polynomial amount of space. Further, known
results only allow to solve Problems~\ref{def:tqa-task-ltl} and
\ref{def:tqa-task-dl} in \ExpTime in combined complexity.

Our first step is to provide a new characterization of r-satisfiability that is
tailored to \DLLitehhorn and can be decided using only polynomial space (see
Section~\ref{sec:dlltcqs-r-sat-charact}). It also allows us to show that
r-satisfiability is first-order rewritable (see
Section~\ref{sec:dlltcqs-r-sat-rew}). In Sections~\ref{sec:dlltcqs-cc} and
\ref{sec:dlltcqs-dc} we then determine the combined and data complexity of
satisfiability in \DLLitehhorn, respectively. There we integrate our
characterization of r-satisfiability, which solves
Problem~\ref{def:tqa-task-dl}, into algorithms that additionally solve
Problems~\ref{def:tqa-task-obtaining} and~\ref{def:tqa-task-ltl} while satisfying the corresponding resource constraints.


\section{Characterizing r-Satisfiablility in \texorpdfstring{\DLLitehhorn}{DL-LiteHhorn}}
\label{sec:dlltcqs-r-sat-charact}



We consider a \DLLitehhorn TKB $\atkb=\langle\aont,(\afb_i)_{0\le i\le
n}\rangle$, a Boolean TCQ~\atcq, and a set~\as and a mapping~$\iota$ as in the
previous section.
\textcolor{blue}{
We consider these objects fixed, and do not always explicitly include them in the notation we use in the following.}
The goal is to find interpretations~$\Jmc_i$ satisfying
$\langle\aont,\afb_i\rangle$ and $\chi_{\iota(i)}$, $i\in[0,n+k]$, that respect the rigid names.

To solve this problem in \PSpace (and later in \ALogTime in data complexity), the idea is to guess a polynomial amount of additional information that allows
us to split the above tests into \emph{independent} satisfiability tests.
The additional information enforces a certain connection between these tests, which
simulates the following effects of the shared domain:
\begin{enumerate}[label=\textnormal{(F\arabic*)}]
	\item\label{eq:tqa-domfunctinds} The interpretation of rigid names over the named individuals \ctwo{is synchronized}.
	\item\label{eq:tqa-domfunctchis} The satisfiability of~$\chi_{\iota(i)}$ at time point~$i$ \ctwo{cannot be} contradicted by the interpretation of the rigid names at the other time points.
\end{enumerate}
In this section, we first describe the precise form that this additional
information takes, then present the characterization itself, and lastly prove
its correctness.



We augment the consistency tests for the KBs $\langle\aont,\afb_i\rangle$ and for the conjunctions $\chi_{\iota(i)}$ by additional ABoxes and conditions that encode the information that is shared between the time points.
More formally, the additional information is a tuple $(\AR,\QR,\QRn,\RF)$, where
\begin{itemize}
	\item \AR is a set of rigid assertions over the names occurring in~\Kmc,
	  specifying the behaviour of the rigid concept and role names on all named
		individuals;
	\item $\QR\subseteq\tcqcqs\atcq$ contains those CQs that are
	  satisfied in at least one of the interperetations;
	\item $\QRn\subseteq\tcqcqs\atcq$ specifies which CQs are \emph{not} satisfied by at least one of the
	  interpretations; and
	\item \RF is a set of assertions of the form $\exists S(b)$, where $S$ is a
	  flexible role name; such an assertion encodes the information that $b$ has
		an $S$-successor at some point in time, which means that the influence of
		this successor on the interpretation of the rigid names needs to be taken
		into account also at the other time points.
\end{itemize}
Roughly speaking, \AR simulates the effect~\ref{eq:tqa-domfunctinds} of the
common domain, while \QR, \QRn, and~\RF express~\ref{eq:tqa-domfunctchis}.
The additional information thus consists of a number of assertions and queries
that is polynomial in the size of~\atcq.
Each of these four items gives rise to (i) additional ABoxes and/or (ii)
external conditions that have to be satisfied by~\aont, $\afb_i$,
and~$W_{\iota(i)}$.
In the following subsections, we define those precisely.

\subsection{Rigid ABox Type}
\label{sec:ar}

The set~\AR is a so-called \emph{rigid ABox type}, which completely fixes the
interpretation of the rigid names on the individual names.
%
\begin{definition}[Rigid ABox Type]\label{def:dlltcqs-aboxtype}
  A \emph{rigid ABox type} for~\atkb is a set $\AR\subseteq\ASr(\atkb)$ such that, for all $\lnot\aaxiom\in\ASr(\atkb)$, we have $\lnot\aaxiom\in\AR$ iff $\aaxiom\notin\AR$.
\end{definition}

We require that all interpretations~$\Imc_i$ must satisfy the assertions in~\AR.

\subsection{Rigid Consequences}
\label{sec:qr}

The second set, \QR, contains all CQs~$\acq_j\in\tcqcqs\atcq$
for which $p_j$ occurs in some~$W_{\iota(i)}$ with $i\in[0,n+k]$ (which means that $\acq_j$ occurs
positively in~$\chi_{\iota(i)}$, and hence must be satisfied by~$\Imc_i$).
We keep track of these CQs because their satisfaction
implies the presence of certain rigid structures at all time points.
To explicitly refer to these structures, we instantiate all variables in CQs by
fresh individual names.
Formally, given a set \ctwo{$\Q\subseteq\tcqcqs\atcq$} (\eg \QR), the ABox~$\cqinst{\Q}$
is obtained by
replacing every variable~$x$ in~$\Q$ with a fresh individual name~$\ael{x}$, and
viewing the resulting (ground) CQ as a set of assertions.

\begin{definition}[Rigid Consequences]\label{def:dlltcqs-consequences}
	The set $\rigcons{\Q}$ of \emph{rigid consequences} of a set of CQs~\Q (w.r.t.~\aont) contains exactly those assertions $\aaxiom\in\ASr(\langle\Omc,\cqinst{\Q}\rangle)$ that are entailed by $\langle\Omc,\cqinst{\Q}\rangle$.
\end{definition}

The set~\QR contributes the assertions in $\rigcons{\QR}$ to the individual
consistency tests, as well as the additional condition mentioned above; \ie
every CQ that is satisfied at some time point must be included in~\QR.

\subsection{Rigid Witnesses}
\label{sec:qrn}

For the set~\QRn of all CQs from \tcqcqs\atcq
 that are \emph{not}
satisfied at some time point, we have to enforce a dual condition to the one
above.
That is, we have to \emph{disallow} the presence of any rigid structures that
imply the satisfaction of such a CQ.
We consider first the case that a CQ is satisfied by the unnamed part of an
interpretation~$\Imc_i$.
The case where the CQ is satisfied (partly) by named individuals is captured
by the set~\RF. 

\begin{definition}[Rigid Witness Query]\label{def:dlltqcs-rigid-witness-query}
	A CQ~$\psi$ is a \emph{rigid witness query} of a set of CQs~\Q if there exists $\acq\in\Q$ such that
	\begin{itemize}
		\item $\langle\aont,\cqinst{\{\psi\}}\rangle\models\acq$, \ie whenever $\psi$ is satisfied, the same must hold for~$\acq$;
		\item \ctwo{$\NC(\psi)\cup\NR(\psi)\subseteq\NRC(\aont)\cup\NRR(\aont)$, \ie $\psi$ uses only rigid names; and}
		\item \ctwo{$|\NT(\psi)|\le|\NT(\acq)|$, \ie the size of~$\psi$ is bounded by a polynomial in the sizes of~$\acq$ and~$\aont$.}
	\end{itemize}
\end{definition}

The condition imposed by~\QRn requires that no witness of~\QRn is
satisfied at any time point, because that would imply that an element of~\QRn
has to be satisfied at \emph{every} time point, contradicting the purpose
of~\QRn.

\subsection{Flexible Successors of Named Elements}
\label{sec:rf}

%
The set~\RF represents the last part of the additional information we have to guess. It contains information about flexible role successors of named
individuals, 
to capture possible effects of RIs involving
both rigid and flexible roles, as sketched in the following example.

\begin{example}
At $n=1$, the TCQ
$$\atcq:=\big(\Previous A(a)\big)\land \lnot\big(\exists x.B(a)\land R(a,x)\land T(a,x)\big)$$ is not satisfiable w.r.t.\ the TKB $\atkb~\langle\aont,(\emptyset,\{B(a) \})\rangle$, where \aont contains CIs $A\sqsubseteq \exists S$, $S\sqsubseteq R$, and $S\sqsubseteq T$, and  $R$ and $T$ are the only rigid symbols.
This is because every model $\atint=(\aint_i)_{i\ge 0}$ of \atkb and \atcq must satisfy $\aint_0\models\aont$, $\aint_1\models\aont$, $\aint_0\models \acq_1$, $\aint_1\models B(a)$, and $\aint_1\not\models\acq_2$, where $\acq_1=A(a)$ and
$\acq_2=\exists x.B(a)\land R(a,x)\land T(a,x)$.
Thus, there has to be an element~\el such that the tuple $(a,\el)$ is contained in $S^{\aint_0},R^{\aint_0},$ and $T^{\aint_0}$.
Since \atint respects the rigid names, this means that this tuple is also contained in $R^{\aint_1}$ and $T^{\aint_1}$. Hence, $\aint_1\models \exists x.B(a)\land R(a,x)\land T(a,x)$, which is a contradiction.

Rigid ABox types, consequences, and witnesses, however, do not help in this case, since \atkb and \atcq imply that a named individual has a \emph{flexible} role successor~$e$ that implies several rigid relations to the same element~$e$.
To see this, consider 
 $\ax_1=\{p_1\}$, $\ax_2=\emptyset$,
a mapping $\iota=\{ 0\mapsto1,1\mapsto2 \}$, and arbitrary interpretations
$\Jmc_0,\Jmc_1,\Jmc_2,\Jmc_3$ as in Definition~\ref{def:tqa-r-sat} except
that they do not share one domain.
We must then have $\QR=\{\acq_1\}$ and $\QRn=\{\acq_1,\acq_2\}$.
The corresponding rigid ABox type $\AR=\{ \exists R(a), \exists T(a) \}$
captures the existence of the rigid relations, but not the fact that they refer
to the same domain element.
That is, even if we require $\Jmc_0,\Jmc_1,\Jmc_2,\Jmc_3$ to all satisfy~\AR,
this only means that $a$ must have one $R$-successor and one $T$-successor
(except in $\Jmc_0$ and $\Jmc_2$, where they must be the same element).
Likewise, the only rigid consequences of~$\acq_1$ are
$\{ \exists R(a), \exists T(a) \}$.
There are no rigid witnesses for~$\acq_2$ since this query is not entailed by a
combination of rigid names.
Hence, the additional information in the rigid ABox type, consequences, and witnesses
alone cannot detect the contradiction in~$\Imc_1$ as described above.
\end{example}

Formally, \RF contains assertions of the form $\exists S(b)$ with
$S\in\NFRM(\aont)$ and $b\in\NI(\atkb)\cup\NIA$, \ctwo{where
$\NIA\subseteq\NI$ contains all individual names~$\ael{x}$ we
introduced in Section~\ref{sec:qr} for the variables~$x$ occurring in the CQs
of~\tcqcqs\atcq}. Because of our assumption that the CQs have no variables in common, each
$\ael{x}\in\NIA$ can be associated to the unique CQ containing~$x$.

The set \RF captures information about which named elements can have which kinds
of flexible role successors, which gives rise to the (rigid) ABox
$\Amc_{\RF}:=\bigcup_{\exists S(b)\in\RF}\Amc_{\exists S(b)}$,
where $\Amc_{\exists S(b)}$ is constructed as follows.
\begin{enumerate}
  \item For every domain element $\uel{b\apath}$ of the canonical interpretation
    $\Imc_{\langle\aont,\{\exists S(b)\}\rangle}$ where the length of~\apath is
    at most $\max\{|\Var(\acqalpha)| \mid \acqalpha\in\tcqcqs\atcq\}$, introduce
    a new individual name $\ael{b\apath}$.
		These new individual names are collected in the set~\NIT.
	\item For every $\ael{\apath S}\in\NIT$, add the following rigid assertions to
	  the set $\Amc_{\exists S(b)}$, which is initially empty:
		\begin{itemize}
			\item for every $B\in\BCr(\aont)$ with
				$\aont\models\exists S^-\sqsubseteq B$, the concept assertion
				$B(\ael{\apath S})$; 
			\item for every $R\in\NRRM(\aont)$ with $\aont\models S\sqsubseteq R$, the
				role assertion $R(\ael{\apath},\ael{\apath S})$ if
				$\apath\not \in\NI(\atkb)$, otherwise $R({\apath},\ael{\apath S})$.
		\end{itemize}
\end{enumerate}
Observe that these consequences only have to be considered up to a depth which
ensures that possible matches of the CQs in \tcqcqs\atcq can be fully
characterized.
The ABox~$\Amc_{\RF}$ is of exponential size,
but this does not affect the complexity results (see Sections~\ref{sec:dlltcqs-cc} and \ref{sec:dlltcqs-dc}).

\subsection{R-complete Tuples}
\label{sec:r-complete-def}

Based on the tuple $(\AR,\QR,\QRn,\RF)$, we hence consider the ABoxes \AR, $\rigcons{\QR}$, and $\Amc_{\RF}$ described above.
%
In addition, 
at each $i\in[0,n+k]$, we check consistency also \wrt $\Amc_i$ and
$\Amc_{\Q_{\iota(i)}}$, where $\Q_{\iota(i)}$ contains all CQs~$\acq_j$ that occur
positively in the conjunction $\chi_{\iota(i)}$, \ie for which we have
$p_j\in W_{\iota(i)}$.
This gives rise to the knowledge bases
\[ \KR[i]:=\langle\aont,
\AR\cup\rigcons{\QR}\cup\afb_{\Q_{\iota(i)}}\cup\afb_{\RF}\cup\afb_i\rangle\ 
\]
which we need to check for consistency.
Observe that $\rigcons{\QR}$, $\Amc_{\Q_i}$, and~$\Amc_{\RF}$ may share individual
names from~$\NI(\atkb)$ and~\NIA. Different ABoxes~$\Amc_{\Q_i}$
and~$\Amc_{\Q_j}$ share individual names from~\NIA for the CQs in ${\Q_i}\cap{\Q_j}$, while individual names from~\NIT only occur in~$\Amc_{\RF}$.

We finally summarize the conditions outlined above in the following
definition.


\begin{definition}[r-complete]\label{def:dlltcqs-r-complete}
A tuple $(\AR,\QR,\QRn,\RF)$ as above is \emph{r-complete}
(w.r.t.~\as and~$\iota$) if the following hold for all $i\in[0,n+k]$.
\begin{enumerate}[label=\textnormal{(\rccond\arabic*)}]
\item\label{def:dlltcqs-rc:consistent}
$\KR[i]$ is consistent.
\item\label{def:dlltcqs-rc:negcqs}
For all $p_j\in\overline{\ax_{\iota(i)}}$, we have
$\KR[i]\not\models\acqalpha_j$.
\item\label{def:dlltcqs-rc:cqcons}
For all $p_j\in\ax_{\iota(i)}$, we have $\acqalpha_j\in\QR$.
\item\label{def:dlltcqs-rc:qrn}
For all $p_j\in\overline{\ax_{\iota(i)}}$, we have $\acqalpha_j\in\QRn$.
\item\label{def:dlltcqs-rc:witnesses}
For all CQs $\acqalpha\in\QRn$ and \ctwo{rigid}
witness queries~\acqtwo for~\acqalpha w.r.t.~\aont, we have
$\KR[i]\not\models\acqtwo$.
\item\label{def:dlltcqs-rc:rf}
For all $S\in\NFRM(\aont)$ and $b\in\NI(\atkb)\cup\NIA$, we have
$\exists S(b)\in\RF$ iff there is an index $j\in[0,n+k]$ such that
$\langle\aont,\AR\cup\rigcons{\QR}\cup\Amc_{Q_{\iota(j)}}\cup\Amc_j\rangle
\models \exists S(b)$.
\end{enumerate}
\end{definition}

Note that all conditions except~\ref{def:dlltcqs-rc:rf} refer only to a single
index~$i$.
Conditions~\ref{def:dlltcqs-rc:consistent} and~\ref{def:dlltcqs-rc:negcqs}
ensure that we can actually satisfy $\langle\Omc,\Amc_i\rangle$ and
$\chi_{\iota(i)}$ together with the additional ABoxes.
As described in Sections~\ref{sec:qr} and~\ref{sec:qrn},
Conditions~\ref{def:dlltcqs-rc:cqcons} and~\ref{def:dlltcqs-rc:qrn} make sure
that only the queries from~\QR (resp., \QRn) can occur in some~$\ax\in\as$
(resp., $\overline{\ax}$).
\cond\ref{def:dlltcqs-rc:witnesses} checks that the queries corresponding to
some $\overline{\ax}$ are not entailed because of the rigid names, by requiring
that the KBs~$\KR[i]$ do not entail any of their witnesses.
And the last condition ensures that \RF is minimal; that is, that it contains only
those assertions $\exists S(b)$ that are required by one of the KBs~\KR[i]
(excluding~$\Amc_{\RF}$).


\begin{lemma}\label{lem:dlltcqs-iff-s-r-consistent}
\as is r-satisfiable w.r.t.~$\iota$ and~\atkb iff there is an r-complete 
tuple w.r.t.~\as and~$\iota$.
\end{lemma}

A full proof of this lemma can be found in Appendix~\ref{app:r-complete}; we
only sketch the main ideas here.
For the ``only if''-direction, let $\Jmc_0,\ldots,\Jmc_{n+k}$ be interpretations over a common domain~\adom, which exist according to the r-satisfiability of~\as (see Definition~\ref{def:tqa-r-sat}).
Finding $(\AR,\QR,\QRn,\RF)$ is then straightforward: since the interpretations respect rigid names, \AR is uniquely defined, $\QR,\QRn$ are determined by~\as (see Conditions~\ref{def:dlltcqs-rc:cqcons} and~\ref{def:dlltcqs-rc:qrn}), and \RF is then given by \cond\ref{def:dlltcqs-rc:rf}.
It is rather easy to show that each of the knowledge bases~\KR[i] has a model (\cond\ref{def:dlltcqs-rc:consistent}) which satisfies neither a CQ from $\tcqcqs{\atcq}\setminus\Q_{\iota(i)}$ (\cond\ref{def:dlltcqs-rc:negcqs}) nor a witness of the queries in \QRn (\cond\ref{def:dlltcqs-rc:witnesses}).
However, special attention needs to be given to the UNA and \cond\ref{def:dlltcqs-rc:rf}.
The problem is that the homomorphisms witnessing $\Jmc_i\models\chi_{\iota(i)}$
may map several variables to the same domain element, but the new individual
names $a_x\in\NIA$ in~$\rigcons{\QR}$ and $\Amc_{Q_{\iota(i)}}$ are required to be
different by the UNA.
Similar considerations apply to the individual names \NIT in $\Amc_{\RF}$.
The solution is to introduce enough copies of these domain elements in order to
satisfy the UNA.

For the \enquote{if}-direction,
we assume an r-complete tuple~$(\AR,\QR,\QRn,\RF)$ to be given, and need to find interpretations $\Jmc_0,\ldots,\Jmc_{n+k}$ that satisfy the requirements of Definition~\ref{def:tqa-r-sat}.
The idea is to start with the canonical interpretations of the KBs given by
\cond\ref{def:dlltcqs-rc:consistent}, and then merge them to obtain a common
domain~$\Delta$ while satisfying the rigid names.

%
We introduce a few auxiliary notions, for $i\in[0,n+k]$:
\begin{itemize}
\item \ctwo{$\aint_i$ is an abbreviation for the canonical interpretation~$\aint_{\KR[i]}$}
of~$\KR[i]$ as specified in Definition~\ref{def:dll-io}.

\item We rename every element $\ael{x}\in\NIA\cap\Delta^{\Imc_i}$ to~$\aiel{x}$,
and collect all these elements in the set~$\adomp[\aint_i]\aux$. We similarly
define $\adomp[\aint_i]\tree$ and $\adomp[\aint_i]\un$ based on the elements
of~$\NIT\cap\Delta^{\Imc_i}$ and the unnamed domain elements of~$\Imc_i$,
respectively.
\item 
\emph{Witnesses} are similar to rigid witness queries, but simpler. A witness is a set of rigid basic concepts that, if satisfied in a model of \aont, implies the existence of a particular role chain.
We use them to detect whether we need to include certain anonymous domain elements from~$\Imc_i$ also at other time points $j\neq i$.
\end{itemize}
\begin{definition}[Witness]\label{def:dlltcqs-witness}
	Let \aint be the canonical interpretation for a knowledge base $\langle\aont,\afb\rangle$.
	A set $\awit\subseteq\BCr(\aont)$ is a \emph{witness} of
	$\uel{\apath}\in\adomp[\aint]\un$ w.r.t.\ 
	$\langle\aont,\afb\rangle$ if $\apath=\apathtwo R_0\ldots R_\ell$ is such that
	$\aont\models\bigsqcap
	\awit\sqsubseteq\exists R_0$
	and either $\uel{\apathtwo}{}\in(\bigsqcap\awit)^\aint$ or
	$\apathtwo\in\NI(\afb)\cap(\bigsqcap\awit)^\aint$.
	The set of all witnesses of $\uel{\apath}$ is denoted
	by~$\wits\aont{\uel{\apath}}$.
	For all $\el\in\adom[\aint]\setminus \adomp[\aint]\un$, we define
	$\wits\aont\el:=\emptyset$.
\end{definition}
Hence, the domain of~$\Imc_i$ is composed of the pairwise disjoint sets
$\NI(\atkb)$, $\adomp[\aint_i]\aux$, $\adomp[\aint_i]\tree$, and
$\adomp[\aint_i]\un$.
Moreover, the domains of different interpretations~$\Imc_i,\Imc_j$ only overlap
in~$\NI(\Kmc)$.
The common domain is now defined as
$\adom := \NI(\atkb)\cup\ \bigcup_{i=0}^{n+k}\big( \adomp[\Imc_i]\aux\cup \adomp[\Imc_i]\tree\cup \adomp[\Imc_i]\un\big)$.

Next, we construct the interpretations $\Jmc_0,\ldots,\Jmc_{n+k}$ over~$\Delta$
as required by Definition~\ref{def:tqa-r-sat}.
The canonical interpretation~$\Imc_i$ represents the parts specific to~$\Jmc_i$
and, for the interpretation of the rigid names in~$\Jmc_i$, all~$\Imc_j$ with
$j\in[0,n+k]$ are considered. The interpretation of the flexible names then
can obviously not be solely based on~$\Imc_i$, but has to be adjusted.
\changed{Intuitively, we include in~$\Jmc_i$ all consequences of the rigid names in~$\Imc_j$.}
%
%
Formally, for all $i\in[0,n+k]$, \ctwo{we define $\Jmc_i$ as follows.}
\begin{itemize}
\item \ctwo{$\Delta^{\Jmc_i}:=\Delta$.}
\item For all $a\in\NI(\atkb)$, $a^{\Jmc_i}:=a.$
\item For all rigid concept names~$A$,
$A^{\Jmc_i} := \bigcup_{j=0}^{n+k} A^{\Imc_j}.$ 
\item For all flexible concept names~$A$, \changed{$A^{\Jmc_i} := A^{\Imc_i}
\cup\bigcup_{j=0}^{n+k} \RC{A}{\Imc_j}$, where}%
\begin{align*}
\changed{\RC{A}{\Imc_j}} & \changed{:= \big\{\ujel{\apath R}\in\adomp[\Imc_j]\un \mid \aont\models \exists R^-\sqsubseteq A,\ \wits\aont{\ujel{\apath R}}\neq\emptyset\big\} \cup{}} \\
&\phantom{:={}} \ctwo{\bigcup
\big\{\big(\bigsqcap\abcset\big)^{\Imc_j} \mid \abcset\subseteq\BCr(\aont),\ \aont\models \bigsqcap\abcset\sqsubseteq A \big\}}
\end{align*}
captures the flexible consequences of the rigid names in~$\Imc_j$.
\item For all rigid role names~$R$,
$R^{\Jmc_i} := \bigcup_{j=0}^{n+k} R^{\Imc_j}.$ 
\item For all flexible role names~$R$, \changed{$R^{\Jmc_i} := R^{\Imc_i} \cup \bigcup_{j=0}^{n+k}\RC{R}{\Imc_j}$, where}
\begin{align*}
\changed{\RC{R}{\Imc_j}} & \changed{:=
\big\{(\elone,\eltwo)\in R^{\Imc_j} 
\mid\wits\aont\elone\neq\emptyset\text{ or }\wits\aont\eltwo\neq\emptyset\big\} \cup{}} \\
&\phantom{:={}} \ctwo{\bigcup\big\{
S^{\Imc_j} \mid S\in\NRRM(\aont),\ \aont\models S\sqsubseteq R \big\}}.
\end{align*}
\end{itemize}
The interpretations $\Jmc_0,\dots,\Jmc_{n+k}$ share the same domain and respect the rigid names. 
%
%
%
We next point out an important characterization of $B^{\Jmc_i}$ for all basic concepts~$B$, in terms of the original interpretations $\Imc_0,\dots,\Imc_{n+k}$.
\begin{restatable}{lemma}{lemdlltcqsIdef}\label{lem:dlltcqs-idef}
	For all $i,j\in[0,n+k]$ and basic concepts $B\in\BC(\aont)$, the
	following hold.
	\begin{enumerate}[label=\textnormal{\alph*)}]
		\item\label{lem:dlltcqs-idef:ind}
		For all $\el\in\NI(\atkb)$, we have $\el\in B^{\Jmc_i}$ iff $\el\in B^{\Imc_i}$.
		\item\label{lem:dlltcqs-idef:auxr} 
		If $B$ is rigid, then, for every
		$\el\in\adomp[\Imc_j]\aux\cup\adomp[\Imc_j]\tree\cup\adomp[\Imc_j]\un$, we have $\el\in B^{\Jmc_i}$ iff $\el\in B^{\Imc_j}$.
		\item\label{lem:dlltcqs-idef:auxf} 
		If $B$ is flexible, then, for every
		$\el\in\adomp[\Imc_j]\aux\cup\adomp[\Imc_j]\tree\cup\adomp[\Imc_j]\un$, we have $\el\in B^{\Jmc_i}$ iff
		\begin{enumerate}[label=\textnormal{(\roman*)}]
			\item\label{lem:dlltcqs-idef:auxf1} $i=j$ and $\el\in B^{\Imc_i}$, or
			\item\label{lem:dlltcqs-idef:auxf2} there is a $\abcset\subseteq\BCr(\aont)$ with $\el\in(\bigsqcap\abcset)^{\Imc_j}$ and
			$\aont\models \bigsqcap\abcset\sqsubseteq B$, or
			\item\label{lem:dlltcqs-idef:auxf3} $\el\in B^{\Imc_j}\cap\adomp[\Imc_j]\un$ and 
			$\wits\aont\el\not=\emptyset$.
		\end{enumerate}
	\end{enumerate}
\end{restatable}
Based on Lemma~\ref{lem:dlltcqs-idef}, we can show that $\Jmc_i$ is a model of~$(\aont,\afb_i)$. 
The main part of the proof, however, is to show that $\Jmc_i$ satisfies the corresponding conjunction $\chi_{\iota(i)}$ of CQ literals. For the positive literals, this is easy given that the ABox $\Amc_{\Q_{\iota(i)}}$ contains an instantiation of all these CQs and is satisfied by $\Imc_i$.
For the negative literals $\lnot\acqalpha$, we show that, if $\Jmc_i$ satisfies \acqalpha via a homomorphism~$\pi$, then one of the following cases must apply:
%
\begin{enumerate}[label=(\Roman*)]
  \item $\pi$ maps all terms to unnamed domain elements of a single~$\Imc_j$, and a rigid witness query of \acqalpha is satisfied in~$\Imc_j$.
	\item The image of~$\pi$ includes named elements, and either it maps directly into $\Imc_i$ or we can construct such a homomorphism. The latter holds because, if $\pi$ maps some terms to named domain elements from~$\Imc_j$, $j\neq i$, corresponding rigid knowledge on the named elements must be contained in the additional ABoxes and thus also be satisfied in~$\Imc_i$.
\end{enumerate}
Thus, the first case contradicts Condition~\ref{def:dlltcqs-rc:witnesses}, and
the second case is impossible due to Condition~\ref{def:dlltcqs-rc:negcqs}.
This concludes the proof of Lemma~\ref{lem:dlltcqs-iff-s-r-consistent}.
%


\section{Combined Complexity}\label{sec:dlltcqs-cc}

Our characterization shows that there is no need to store the exponentially large set~\as in order to check r-satisfiability.
That is, given an r-complete tuple $(\AR,\QR,\QRn,\RF)$ and a time point~$i$
with associated world $W:=W_{\iota(i)}$ and ABox $\Amc:=\Amc_i$, the conditions
of Definition~\ref{def:dlltcqs-r-complete} (except for one direction
of~\ref{def:dlltcqs-rc:rf}) can be checked independently for the KB
$\langle\aont,\AR\cup\rigcons{\QR}\cup\afb_{\Q_\ax}\cup\afb_{\RF}\cup\afb\rangle$ if
we define $\Q_\ax:=\{\acq_j\in\tcqcqs{\atcq}\mid p_j\in\ax\}$.
This allows us
to show that TCQ satisfiability (and hence
also entailment) is in \PSpace w.r.t.\ combined complexity, which matches the
complexity of satisfiability in LTL (\cf Lemma~\ref{lem:ltl-periodic-model}).

We adapt the procedure for LTL~\cite{SiCl85:ltlpspace} as described in
Algorithm~\ref{algo:tq-sat}.
It constructs the propositional types $T_0,\dots,T_s,\dots,T_{s+p}$
one after the other, without storing the whole sequence.
It keeps in memory two types~$\atype_{\fnow}$ and~$\atype_{\fnext}$ for
the current and next time point, respectively, and checks whether these sets are
t-compatible \changed{(see Appendix~\ref{sec:prelims-ltl} for the definition of t-compatibility)}.
The algorithm additionally guesses the start~$s$ and length~$p$ of the period
(Lines~\ref{ls}, \ref{lp}), stores the type $\atype_{\fs}$ (Line~\ref{lfs}), and
checks if the period is valid by comparing $\atype_{\fnext}$ to $\atype_{\fs}$
at time point $s+p$ (Line~\ref{lfinal}).
It also ensures that all $\Until$-formulas are satisfied within the period
(Lines~\ref{lfs}, \ref{lfu}, \ref{lfinal}).
\begin{algorithm}[tb]
	\fontsize{8pt}{11pt}\selectfont
	\KwIn{TCQ~\atcq, TKB~$\atkb~\langle\aont,(\afb_i)_{0\le i\le n}\rangle$}
	\KwOut{\true if \atcq is satisfiable w.r.t.\ \atkb, otherwise \false}
	\mycolorbox{AlgoGray}{Guess $(\AR,\QR,\QRn,\RF)$ and set $\RF':=\RF$}\;
	\label{lguess}
	$i:=0$\;
	$s:=$ Guess a number between \mycolorbox{AlgoGray}{n} and $2^{|\pa\atcq|\smash{\mycolorbox{AlgoGray}{$\scriptstyle+|\Kmc|$}}}$\;
	\label{ls}
	$p:=$ Guess a number between $0$ and $4^{|\pa\atcq|\smash{\mycolorbox{AlgoGray}{$\scriptstyle+|\Kmc|$}}}$\;
	\label{lp}
	$\atype_{\fnext}:=\emptyset$, 
	$\atype_{\fs}:=\emptyset$, 
	$\atype_{\Until}:=\emptyset$\;
	%
	$\atype_{\fnow}:=$ Guess an element of $\atypeset[\{\pa\atcq\}]$\;
	\lIf{\upshape
			$\atype_{\fnow}$ is not initial
	}{\Return{\false}}
	\While{$i< s+p$}{
		$\atype_{\fnext}:=$ Guess an element of $\atypeset[\{\pa\atcq \}]$\;
		
		
		\lIf{\upshape
		$(\atype_{\fnow},\atype_{\fnext})$ is not t-compatible
		}{\Return{\false}}
	\lIf{$i=s$}{
		$\atype_s:=\atype_{\fnow}$;\ \
		$\atype_{\Until}:=\{\altlform_1\Until\altlform_2\in \atype_s\}$}\label{lfs}
	\lIf{$i\ge s$}{
		$\atype_{\Until}:=\{\altlform_1\Until\altlform_2\in \atype_{\Until} \mid \altlform_2\not\in \atype_{\fnow}\}$
		\label{lfu}
	}
	\lIf{\upshape\mycolorbox{AlgoGray}{$i=n$} \textbf{and} $\pa\atcq\not\in\atype_{\fnow}$}{\Return{\false}}
	\label{lphi}
	\mycolorbox{AlgoGray}{$\ax:=\atype_{\fnow}\cap\pv$}\;
	\label{lw}
	\mycolorbox{AlgoGray}{\lIf{\upshape\textbf{not}  {\textnormal{\algofont{\algoTesting}($\atcq,\aont,\Amc_i,(\AR,\QR,\QRn,\RF),\ax$)}}}{
		{\Return{\false}}}}\;
	\label{lrsat}
	\mycolorbox{AlgoGray}{$\RF':=\{\exists S(b)\in\RF' \mid \langle\Omc,\AR\cup\rigcons{\QR}\cup\Amc_{Q_\ax}\cup\Amc_i\rangle\not\models \exists S(b)\}$}\;
	\label{lrfthree}
	$i:=i+1$\;
	$\atype_{\fnow}:=\atype_{\fnext}$\;
}
\lIf{\upshape
  $\atype_{\Until}=\emptyset$ \textbf{and}
  \mycolorbox{AlgoGray}{$\RF'=\emptyset$} \textbf{and}
  $(\atype_{\fnow},\atype_{s})$ is t-compatible
}{
\Return{\true}}
\label{lfinal}
\Return{\false}

\caption{\PSpace procedure for deciding TCQ satisfiability}
\label{algo:tq-sat}
\end{algorithm}

Our modifications (highlighted in gray) ensure that we consider at least $n$ time points (Line~\ref{ls}).
For t-satisfiability, we check that $\pa\atcq$ is satisfied at~$n$ instead of at~$0$ (Line~\ref{lphi}).
Now, r-satisfiability can be tested in a modular fashion (Line~\ref{lrsat})
in the procedure \algofont{\algoTesting} (see Algorithm~\ref{alg:rsat}),
given a tuple $(\AR,\QR,\QRn,\RF)$ guessed in the beginning (Line~\ref{lguess})
and the current world~$W$ (Line~\ref{lw}).
The additional set~$\RF'$ checks the global part of Condition~\ref{def:dlltcqs-rc:rf} by ensuring that all elements $\exists R(b)\in\RF$ are entailed by one of the KBs encountered by the algorithm.
%
We hence integrate the r-satisfiability and t-satisfiability tests from Lemma~\ref{lem:tcq-sat-iff} such that $\as$ and $\iota$ are implicitly represented by the worlds induced by the sequence of guessed types.

\begin{algorithm}[tb]
	\fontsize{8pt}{11pt}\selectfont
	\caption{The procedure \algofont{\algoTesting}}
	\label{alg:rsat}
	\KwIn{TCQ~\atcq, ontology~\Omc, ABox~$\Amc_i$, tuple $(\AR,\QR,\QRn,\RF)$, world~\ax}
	\KwOut{\true if Conditions \ref{def:dlltcqs-rc:consistent}--\ref{def:dlltcqs-rc:rf} hold for time point~$i$, otherwise \false}

$\KR:={}\langle\aont,\AR\cup\rigcons{\QR}\cup\Amc_{Q_\ax}\cup\Amc_i\rangle$\;

	\lIf{\upshape \KR is inconsistent}{\Return{\false}}
	\label{lcone}
	
	\For{$p_j\in\overline{\ax}$}{
		Guess a set $\Amc'_{\RF}\subseteq\Amc_{\RF}$ of size polynomial in~$|\acq_j|$\;
		\label{larfprime}
		\lIf{$\KR\cup\Amc'_{\RF}\models\acq_j$}{\Return{\false}}
		\label{lctwo}
		\lIf{$\acq_j\notin\QRn$}{\Return{\false}}
		\label{lcfour}
	}
	
	\For{$p_j\in\ax$}{
		\lIf{$\acqalpha_j\notin\QR$}{\Return{\false}}
		\label{lcthree}
	}
	
	\For{\upshape all rigid witness queries~$\psi$ of~\QRn}{
		\label{lrwq}
		\lIf{$\KR\models\psi$}{\Return{\false}}
		\label{lcfive}
	}

	\For{\upshape all $S\in\NFRM(\Omc)$ and $b\in\NI(\Kmc)\cup\NIA$}{
	\label{lrfone}
		\lIf{\upshape $\KR\models\exists S(b)$ \textbf{and} $\exists S(b)\notin\RF$}{\Return{\false}}
	}
	\label{lrftwo}
	
	%
	%
	%
	%
	
	%
	%
	%
	%

	\Return{\true}
\end{algorithm}
\begin{lemma}\label{lem:dlltcqs-algo-ext}
Algorithm~\ref{algo:tq-sat} decides TCQ satisfiability using only polynomial space.
\end{lemma}
\begin{proof}
We consider the conditions in Lemma~\ref{lem:tcq-sat-iff}.
Let the set $\as=\{\ax_1,\dots,\ax_k \}$ be defined as the set of all worlds~\ax encountered during a run of the procedure. The mapping $\iota\colon[0,n]\to[1,k]$ is defined as 
$\iota(i):=\ell$ if $W_\ell$ is the world encountered at time point~$i$.
Regarding t-satisfiability (see Definition \ref{def:tcqs-t-sat}), it is easy to see that the above definitions of~\as and~$\iota$ fulfill the first two conditions.
The last condition follows from the correctness of the original LTL satisfiability algorithm \cite[\citethm4.1, 4.7]{SiCl85:ltlpspace}, which is not affected by our restriction that $s>n$ nor by the other extensions.

It thus remains to show that \as is r-satisfiable iff these extensions do not cause the algorithm to return \false. By Lemma~\ref{lem:dlltcqs-iff-s-r-consistent}, we can consider Conditions \ref{def:dlltcqs-rc:consistent}--\ref{def:dlltcqs-rc:rf} from Definition~\ref{def:dlltcqs-r-complete}. Conditions~\ref{def:dlltcqs-rc:cqcons}--\ref{def:dlltcqs-rc:witnesses} are obviously captured by \algofont{\algoTesting}.

For~\ref{def:dlltcqs-rc:consistent}, observe that \algofont{\algoTesting} only checks the consistency of the knowledge base \KR, which does not include~$\Amc_{\RF}$.

However, this exponentially large ABox can be ignored for
this consistency test since, once \cond\ref{def:dlltcqs-rc:rf} is verified, we know that for each $\Amc_{\exists S(b)}\subseteq\Amc_{\RF}$ there is at least
one index $j\in[0,n+k]$ for which the existence of the elements described in
$\Amc_{\exists S(b)}$ follows from the KB
$\langle\aont,\AR\cup\rigcons{\QR}\cup\ctwo{\Amc_{Q_{\iota(j)}}}\cup\Amc_{j}\rangle$.
%
Hence, the rigid consequences of the assertion $\exists S(b)$ with $b\in\NI(\atkb)\cup\NIA$, which are represented by $\Amc_{\exists S(b)}$, must follow from $\AR$ or $\rigcons{\QR}$ (depending on the kind of~$b$). We can thus disregard the assertions from $\Amc_{\exists S(b)}$ including the elements in \NIT since any model of the KB we consider must have such domain elements.


For~\ref{def:dlltcqs-rc:negcqs}, we have to check whether $\KR\cup\Amc_{\RF}\not\models\acqalpha_j$ holds for each $p_j\in\overline{\ax}$. Considering the nondeterministic variant of the algorithm in~\cite{BotAC-DL10:dllhorn}, it is easy to see that, in order to check for a homomorphism from~$\acqalpha_j$, it suffices to consider only a nondeterministically chosen part of~$\Amc_{\RF}$ of size polynomial in~$|\acqalpha_j|$, the cardinality of $\acqalpha_j$. Additionally, we have to check if there is a named individual from which we can reach this part, but this can also be done while using only polynomial space. 

Finally, we consider~\ref{def:dlltcqs-rc:rf}. The ``if''-direction of the
equivalence is captured by Lines~\ref{lrfone}--\ref{lrftwo} in Algorithm~\ref{alg:rsat}. The other direction of \cond\ref{def:dlltcqs-rc:rf} is checked globally in Line~\ref{lrfthree} in Algorithm~\ref{algo:tq-sat}.
Observe that our global condition
corresponds to the extension of~\atcq with linearly many additional conjuncts of the form $\Diamondm\Diamondf\exists S(b)$, which
may require us to look for an LTL structure with a longer period.
However, the required period is still exponential in the input.


We analyze the complexity. For the original parts of Algorithm~\ref{algo:tq-sat}, we refer to \cite{SiCl85:ltlpspace}. The nondeterministic guessing of the polynomially large sets $\AR,\QR,\QRn,$ and \RF can be done using polynomial space only.
The set $\rigcons{\QR}$ can be computed in polynomial time since it involves only
a polynomial number of \PTime subsumption test in
\DLLitehhorn~\cite[\citethm8.2]{dllrelations}.
Morevoer, \KR is of polynomial size (recall that we drop $\Amc_{\RF}$) and hence
can be tested for consistency in \PTime~\cite[\citethm8.2]{dllrelations}.
The various UCQ entailment tests can be done in \NP by the nondeterministic variant of the algorithm in~\cite{BotAC-DL10:dllhorn} (see the sketch after Theorem~12 in that paper).
The guess in Line~\ref{larfprime} of Algorithm~\ref{alg:rsat} is clearly also possible in polynomial space, and we can enumerate all rigid witness queries in Line~\ref{lrwq} of Algorithm~\ref{alg:rsat} in polynomial space since their size is bounded by the size of the largest CQ in~\tcqcqs{\atcq}.
\end{proof}

Since the nondeterminism is not relevant for \PSpace complexity according to the
well-known result of Savitch~\cite{savitch}, we obtain the desired
complexity result.

\begin{theorem}\label{thm:dlltcqs-cc:wo-rigid-roles}\label{thm:dlltcqs-cc}
TCQ entailment in~\DLLitehhorn is in \PSpace in combined complexity, even if $\NRR\neq\emptyset$.
\end{theorem}

\section{First-Order Rewriting of r-Satisfiability}
\label{sec:dlltcqs-r-sat-rew}

Towards our goal of obtaining a low data complexity for TCQ entailment in
\DLLitehhorn, we first reconsider the r-completeness conditions from
Definition~\ref{def:dlltcqs-r-complete}, and show that they are partially
first-order rewritable.
As before, we consider a TCQ~\atcq and a \DLLitehhorn TKB
$\atkb=\langle\aont,\afbs\rangle$ with $\afbs=(\afb_i)_{0\le i \le n}$.
Since we focus on data complexity, we disregard the impact
of~\aont and~\atcq on the computational resources in the following.
In particular, the size of the set $\as=\{\ax_1,\dots,\ax_k \}\subseteq2^\pv$ is
constant; we consider a fixed such set for now.
However, the same does not hold for the mapping~$\iota$, which depends on the
length of~\afbs.
In addition to~\as, we \ctwo{guess} a set
$\Bphi\subseteq\{B(a) \mid B\in\BC(\aont),\ a\in\NI(\atcq)\}$ of basic concept
assertions over individual names occurring in~\atcq.
%
The size of this set is also constant in data complexity, \ctwo{and hence we also assume $\Bphi$ to be fixed throughout this section.
The set~$\Bphi$ captures additional basic concept assertions that are not contained in the ABoxes but, due to their consequences, critical for determining r-satisfiability.
} 

This allows us to show that:
\begin{itemize}
	\item to verify the r-satisfiability of~\atcq w.r.t.~\atkb, it suffices to check r-completeness of a representative tuple $(\ARs,\QRs,\QRsn,\RFs)$ that depends on~\as and~\Bphi; and
	\item the r-completeness conditions for this tuple can be encoded into (linearly many) FO formulas that are evaluated over a fixed, finite structure \tdbfbs (\ctwo{constructed based on~\afbs in Definition~\ref{def:tdb-int}
	}).
\end{itemize}
These two steps are described in Sections~\ref{sec:dlltcqs-r-sat-rew-tuple} and~\ref{sec:rcomplete-rewriting-details}, respectively.

\subsection{A Tuple for Testing r-Satisfiability}
\label{sec:dlltcqs-r-sat-rew-tuple}

We first describe the tuple $(\ARs,\QRs,\QRsn,\RFs)$ and the corresponding KBs that
are relevant for the r-completeness tests.
We define the sets such that they are minimal \wrt the r-completeness
conditions.
The definition of~$\QRs$ and~$\QRsn$ is straightforward:
\begin{align*}
  \QRs  &:= \{\acqalpha_j\in\tcqcqs\atcq\mid \ax\in\as,\ p_j\in\ax\}, \\
  \QRsn &:= \{\acqalpha_j\in\tcqcqs\atcq\mid \ax\in\as,\ p_j\not\in\ax\}.
%
\intertext{%
To define \ARs, we can use the sets~\Bphi and $\rigcons{\QRs}$ (restricted to
$\NI(\atkb)$), but we also have to consider the rigid consequences of the input
ABoxes~$\Amc_i$.
We give an inductive construction of these consequences that allows us to
consider the ABoxes~$\Amc_i$ in isolation.
We define sets of (positive) rigid assertions inductively as follows \ctwo{for $j\geq0$:}}
%
  \ARs_0 &:= \ASr(\atkb) \cap \big(\Bphi\cup\rigcons{\QRs}\big), \\
  \ARs_{j+1} &:= \{ \alpha\in\ASr(\atkb) \mid \text{ there is } i\in[0,n]  \text{ with }
    \langle\aont,\ARs_j\cup\Amc_i\rangle\models\alpha \}.
\end{align*}
%
After at most \ctwo{$\NBC:=|\BCr(\aont)|$} iterations, this computation becomes stable, \ie
we do not add any more assertions, because
\begin{itemize}
  \item by Definition~\ref{def:dll-io} and Lemma~\ref{lem:dll-iomodel}, all role
    assertions about $\NI(\atkb)$ follow in \emph{one} entailment step from some
    role inclusions in~\aont and a role assertion in~$\ARs_0$ or~$\Amc_i$;
  \item entailment of basic concept assertions~$B(a)$ does not depend on
    basic concept assertions on individual names other than~$a$, and so all
    possible assertions about~$a$ are added after at most $\NBC$ steps.
\end{itemize}
%

\ARs is now defined as the union of~$\ARs_{\NBC}$ and the set of all negative assertions
$\lnot\aaxiom\in\ASr(\atkb)$ for which $\aaxiom\not\in\ARs_{\NBC}$.
The following is a direct consequence of this definition.

\begin{lemma}
\label{lem:ars-other-consequences}
  For $\alpha\in\ASr(\atkb)$, we have $\alpha\in\ARs$ iff there is
  $i\in[0,n]$ with $\langle\aont,\ARs\cup\Amc_i\rangle\models\alpha$.
\end{lemma}

It remains to define the last component, \RFs.
%
Recall that this set of flexible assertions of the form $\exists S(b)$ can refer
to the individual names~$b$ in $\NI(\atkb)$ and~\NIA; moreover we want to
define it as the minimal such set that satisfies the r-completeness conditions
(in particular~\ref{def:dlltcqs-rc:rf}).
With respect to the elements of~\NIA, the set \RFs thus only depends on the fixed query~\atcq; in contrast, the names in
$\NI(\Kmc)$ may occur in the input ABoxes~$\Amc_i$ as well as in~\atcq.
For the rewriting, it is important to separate these cases since
\begin{itemize}
  \item the parts about \NIA are known at the time of the rewriting (depending
    on~\as);
  \item the parts about $\NI(\Kmc)\setminus\NI(\atcq)$ depend on the input
    ABoxes, so they are not fixed; and
  \item the parts about $\NI(\atcq)$ depend on the input ABoxes as well as on the
    TCQ~\atcq.
\end{itemize}
Hence, we define \RFs as the disjoint union of the three sets \RFaux, \RFphi,
and \RFother ($\mathsf{o}$ for \enquote{other}), referring only to names
from \NIA, $\NI(\atcq)$, and $\NI(\atkb)\setminus\NI(\atcq)$, respectively:
\begin{itemize}
	\item 
	\RFaux is constant and hence its size is not relevant. In particular, the elements of \NIA occur neither in the rigid ABox type~\ARs, nor in the input ABoxes~$\Amc_i$, and are uniquely associated to one of the CQs in \atcq. For constructing \RFaux in line with \cond\ref{def:dlltcqs-rc:rf}, it is thus sufficient to focus on the consequences of these CQs (see Lemma~\ref{lem:dll-io-elements}):
	\[ \RFaux :=
		\{ \exists S(a_x) \mid
			S\in\NFRM(\aont),\ a_x\in\NIA,\ \ax\in\as,\ 
			\langle\aont,\Amc_{Q_\ax}\rangle\models \exists S(a_x) \}.
	\]
	\item 
	The set \RFphi is also of constant size. However, the elements of $\NI(\atcq)$ may occur in the input ABoxes, and hence their behavior cannot be fully determined by a computation that is independent of the input. To tackle this problem, we assume the set $\Bphi$ of assertions 
	to be given first, and postpone the test whether \RFphi actually satisfies \cond\ref{def:dlltcqs-rc:rf} to a later point (see Lemma~\ref{lem:fo-r-sat1}).
	For now, we simply set
	\[ \RFphi:=\{ \exists S(a)\in\Bphi \mid S\in\NFRM(\aont) \}. \]
	\item 
	The set \RFother concerns the remaining individual names from
	$\NI(\atkb)\setminus\NI(\atcq)$. Here, we can refer to the flexible
	consequences of the ABoxes~$\Amc_i$ together with~\ARs.
	Since the elements under consideration do not occur in~\atcq, this computation
	actually does not depend on~\as or~\Bphi:
  \[ \RFother:= \{ \exists S(a) \mid
    S\in\NFRM(\aont),\ a\notin\NI(\atcq)\text{ there is } i\in[0,n] \text{ with }
    \langle\aont,\ARs\cup\Amc_i\rangle\models \exists S(a) \}.
  \]
\end{itemize}

This finishes the definition of $(\ARs,\QRs,\QRsn,\RFs)$.
Observe that, apart from $\Amc_{\RFother}$ and \ARs, which depend on the input
ABoxes~\afbs, all of the ABoxes induced by this tuple (see
Section~\ref{sec:dlltcqs-r-sat-charact}) are constant.
Moreover, the tuple is indeed as intended.

\begin{restatable}{lemma}{lemFOrSattwo}\label{lem:fo-r-sat2}
  For all $\as=\{\ax_1,\dots,\ax_k\}\subseteq 2^\pv$ and
  $\iota\colon[0,n]\to[1,k]$, there is an r-com\-plete tuple w.r.t.~\as
  and~$\iota$ iff there is a set
  $\Bphi\subseteq\{B(a) \mid B\in\BC(\aont),\ a\in\NI(\atcq)\}$ such that $(\ARs,\QRs,\QRsn, \RFs)$ is r-complete w.r.t.~\as
  and~$\iota$.
\end{restatable}
\begin{proof}[Proof sketch]
  Given an r-complete tuple $(\AR,\QR,\QRn,\RF)$, we define
  \[ \Bphi:= \{ B(a)\in\AR\cup\RF \mid B\in\BC(\aont),\ a\in\NI(\atcq) \} \]
  and show that $(\ARs,\QRs,\QRsn,\RFs)$ is r-complete as well.
  The KBs used in the r-completeness tests in Definition~\ref{def:dlltcqs-r-complete} look as follows, for all $i\in[0,n+k]$:
  $$\KRis[i]:=\langle\aont,\ARs\cup\rigcons{\QRs}\cup\Amc_{Q_{\iota(i)}}\cup\Amc_{\RFs}\cup\Amc_i\rangle.$$
  Conditions~\ref{def:dlltcqs-rc:cqcons} and~\ref{def:dlltcqs-rc:qrn} are satisfied by construction. 
  %
  %
  For Conditions~\ref{def:dlltcqs-rc:consistent}, \ref{def:dlltcqs-rc:negcqs}, and~\ref{def:dlltcqs-rc:witnesses}, we 
  can find a model of \KRis that is homomorphically embeddable into the canonical interpretation of \KR[i] that is consistent by \cond\ref{def:dlltcqs-rc:consistent}. 
  This is because $\rigcons{\QRs}\subseteq\rigcons{\QR}$, $\Amc_{\RFaux}\cup\Amc_{\RFphi}\cup\Amc_{\RFother}\subseteq\Amc_{\RF}$, and
  $\ARs_{\NBC}\subseteq\AR$.
\end{proof}

\subsection{Rewriting Consistency and Entailment}
\label{sec:rcomplete-rewriting-details}

Now, we can focus on testing the r-completeness of a single, (mostly) fixed tuple in the r-completeness test.
Observe that the tests for r-completeness consist of consistency and non-entailment tests for atemporal KBs, which are standard. 

\ctwo{Many query answering problems in lightweight DLs can be encoded into first-order logic formulas, called \emph{rewritings}, which are then evaluated over the following structures, in which ABoxes are viewed under the closed-world assumption, \ie as databases.

\begin{definition}[$\db\afb$]\label{def:prelims-db-int}\label{def:db-int}
	For an ABox~\afb, the first-order structure $\db{\afb}=(\NI(\afb),\cdot^{\dboa})$ over the domain~$\NI(\afb)$ contains the following relations for all $B\in\BC(\afb)$ and 
	$R\in\NR(\afb)$:
	\begin{align*}
	\Bdb &:= \{a\mid  B(a)\in\Amc\}, &
	\Rdb &:= \{(\indone,\indtwo)\mid  R(\indone,\indtwo)\in\Amc\}.
	\end{align*}
\end{definition}

There are FO rewritings for KB inconsistency and for UCQ
	entailment in \DLLitehhorn, which we here denote by \dllhhornQUnsat\aont and \dllhhornPerfRef\acq\aont, respectively (see, \eg \cite{BotAC-DL10:dllhorn}). 
	These can be easily adapted to our slightly modified setting with assertions about (negated) basic concepts. 

%

\begin{lemma}
	\label{lem:dllhhorn-perfect-ref-correct}
	Let $\akb=\langle\aont,\afb\rangle$ be a \DLLitehhorn knowledge base and \acq
	be a Boolean 
	UCQ. Then \akb is inconsistent iff
	$\db{\afb}\models\dllhhornQUnsat\aont$. If \akb is consistent,
	then 
	$\akb\models\acq$ iff
	$\db{\afb}\models\dllhhornPerfRef\acq\aont$.
\end{lemma}}

The idea is to apply these UCQ rewritings $\dllhhornQUnsat\aont$ and  \changed{ 
$\dllhhornPerfRef\acq\aont$, which are evaluated over the FO structure $\db{\afb}$, where \afb is a single ABox~\Amc. 
} 
%
%
However, the conditions for r-completeness involve the \emph{sequence} of input ABoxes~\afbs, as well as additional ABoxes such as~\ARs. In this section, we describe how these ABoxes can be incorporated into the rewritings such that the resulting FO formulas can be answered over~\afbs alone.
\changed{Then, the r-completeness check for the tuple $(\ARs,\QRs,\QRsn,\RFs)$ can be reduced to the evaluation of (mutiple) FO formulas over~\afbs (see Lemma~\ref{lem:dlltcqs-fo-r-sat}).}
%
%

\ctwo{
To illustrate the main idea of how to extend the rewritings,
%
	consider an atemporal KB $\langle\aont,\afb_i\rangle$ formulated in \DLLitehhorn and a CQ~\acq.
  By~\cite{BotAC-DL10:dllhorn}, there is a rewriting $\dllhhornPerfRef\acq\aont$
  such that
  \[ \langle\aont,\afb_i\rangle\models\acq \text{ iff }
    \db{\afb_i}\models\dllhhornPerfRef\acq\aont. \]
  Assume that we want to incorporate the additional ABox~\ARs into the
  entailment test, without modifying $\db{\afb_i}$.
	Specifically, the goal is to extend $\dllhhornPerfRef\acq\aont$ to an FO formula \prefth{}
  (interpreted under the standard first-order semantics) such that
  \[ \langle\aont,\ARs\cup\afb_i\rangle\models\acq \text{ iff }
	  \db{\afb_i}\models\prefth{}. \]
	If \ARs consists of the single assertion $A(a)$, this can be achieved, for
	instance, by replacing every atom $A(x)$ in $\dllhhornPerfRef\acq\aont$ by the
	disjunction $((x=a)\vee A(x))$.
}

The various additional ABoxes we consider, such as $\Amc_{\RFs}$, contain individual names that do not occur in the input sequence~\afbs (namely those in $\NIA\cup\NIT$).
This makes the required adaptations of the rewritings even more complex, as our FO formulas have to quantify over elements that are not in the interpretation domain.


\ctwo{
We now present the main lemma that will be shown in this section. It characterizes the r-completeness of $(\ARs,\QRs,\QRsn,\RFs)$ by a series of FO-formulas.

\begin{restatable}{lemma}{LemFORSat}\label{lem:fo-r-sat1}\label{lem:dlltcqs-fo-r-sat}
  For all $\as=\{\ax_1,\dots,\ax_k\}\subseteq 2^\pv$, mappings
  $\iota\colon[0,n]\to[1,k]$, and sets
  $\Bphi\subseteq\{B(a)\mid B\in\BC(\aont),\ a\in\NI(\atcq)\}$,
  the tuple $(\ARs,\QRs,\QRsn,\RFs)$ is \mbox{r-complete} w.r.t.~\as and~$\iota$ iff
  the following hold:
  \begin{enumerate}[label=(\alph*)]
    \item\label{rew:a} For all $i\in[0,n]$, we have $\tdbfbs\models\rsat[\ax_{\iota(i)}](i)$. 
    \item\label{rew:b} For all $\ax\in\as$, we have $\tdbfbs\models\rsat(-1)$.
    \item\label{rew:c} For all $S\in\NFRM(\aont)$ and $a\in\NI(\atcq)$, we have $\exists S(a)\in\Bphi$ iff there is an $i\in[0,n]$ such that $\tdbfbs\models\rewPRef{\exists S(a)}$, where this rewriting is \wrt the world $\ax_{\iota(i)}$.
  \end{enumerate}
\end{restatable}

In the following subsections, we describe how to obtain the temporal database $\tdbfbs$ and the rewritings $\rsat(i)$ and $\rewPRef{\exists S(a)}$ used in this characterization.
}

\ctwo{So far, we have restricted our attention to Boolean queries.
However,} the queries that we rewrite in this section may also
be non-Boolean, \ie they may contain variables that are not existentially quantified, called \emph{free} variables. This \ctwo{is necessary for} the presentation of the intermediate queries we construct in the rewriting process.
The rewritings must thus preserve entailment \wrt all possible
groundings, as
defined next.
A \emph{grounding} of a UCQ~\acq \wrt an atemporal KB~\Kmc to be a
function~$\ground$ that maps the free variables of~\acq to individual names from
$\NI(\Kmc)$.
Such a grounding $\ground$ is a \emph{certain answer} to~\acq over~\akb if
$\akb\models\ground(\acq)$, where $\ground(\acq)$ denotes the Boolean UCQ resulting
from~$\acq$ by replacing all free variables according to~$\ground$.
Similarly, $\ground$ is an \emph{answer} to~\acq over the first-order structure
\db\afb if $\db\afb\models\ground(\acq)$.
%
The rewriting \dllhhornPerfRef\acq\aont 
can be assumed to be also 
correct for non-Boolean UCQs~\acq,
in the sense that it has the same free
variables as~\acq and that the certain answers to~\acq over~\akb coincide with
the answers to \dllhhornPerfRef\acq\aont over \db\afb; \ie we have
$\akb\models\ground(\acq)$ iff
$\db\afb\models\ground\big(\dllhhornPerfRef\acq\aont\big)$ (\cf
Lemma~\ref{lem:dllhhorn-perfect-ref-correct}) for every possible
grounding~\ground (see, \eg \cite{BotAC-DL10:dllhorn}).

\subsubsection{From $\Amc$ to $\Amc_i$}

Since we have a temporal semantics, as a first step, we need to lift the
definition of $\db\afb$ to the temporal sequence of ABoxes~\afbs, and to adapt the
rewritings $\dllhhornPerfRef\acq\aont$ and $\dllhhornQUnsat\aont$ accordingly.
In the following, we usually talk only about $\dllhhornPerfRef\acq\aont$, since
the procedures for $\dllhhornQUnsat\aont$ are analogous; however, note that
$\dllhhornQUnsat\aont$ is always Boolean.
For now, \acq is simply an arbitrary CQ, which we later instantiate with the
concrete CQs relevant for the r-completeness test, \eg with the rigid witness queries
for the CQs occurring in~\atcq.

\begin{definition}[$\tdb{\afbs}$]
	\label{def:tdb-int}
	For the ABox sequence $\afbs=(\afb_i)_{0\le i\le n}$, the two-sorted first-order structure $\tdb{\afbs}=(\NI(\afbs),[-1,n],\cdot^{\tdboa})$ over the object domain $\NI(\afbs)$ and temporal domain $[-1,n]$ contains the following relations, for all $B\in\BC(\afbs)$ and 
	$R\in\NR(\afbs)$:
	\begin{align*}
	\B^{\tdboa} &:= \{(a,i)\mid i\in[0,n],\ B(a)\in\Amc_i\}, & 
	\R^{\tdboa} &:= \{(\indone,\indtwo,i)\mid i\in[0,n],\ R(\indone,\indtwo)\in\Amc_i\}.
	\end{align*}
\end{definition}
We use the temporal domain element~$-1$ to describe the prototypical empty ABox $\afb_{-1}:=\emptyset$; all formulas of the form $\B(a,-1)$ and $\R(a,b,-1)$ thus evaluate to \false.
As before, we may use atoms of the form $\R^-(b,a,i)$ to refer to $\R(a,b,i)$.
The relations $\B^{\tdboa}$/$\R^{\tdboa}$ for symbols $B$/$R$ that do not occur
in~\afbs are considered to be empty; for simplicity, we do not explicitly
consider this case in the following.

We now define a first-order formula~$\prefo(i)$ with an additional argument~$i$ that allows us to explicitly refer to time points.
The formula $\prefo(i)$ adapts $\dllhhornPerfRef\acq\aont$ to \tdbfbs:
given $i\in[-1,n]$, it checks whether~\acq is entailed by
$\langle\aont,\afb_i\rangle$.
It is obtained from $\dllhhornPerfRef\acq\aont$ by simply replacing all atoms
\ctwo{$\B(t)$ and $\R(\tone,\ttwo)$} by $\B(t,i)$ and $\R(\tone,\ttwo,i)$, respectively.
Given Lemma~\ref{lem:dllhhorn-perfect-ref-correct}, it is easy to see that this
is correct in the following sense.
\begin{lemma}
	\label{lem:pref1}
	For all CQs~\acq, groundings~\ground, and $i\in[-1,n]$, we have
  $$\langle\aont,\Amc_i\rangle\models\ground(\acq) \text{ iff }
    \tdbfbs\models\ground\big(\prefo(i)\big).$$
\end{lemma}

\subsubsection{From $\Amc_i$ to $\ARs\cup\Amc_i$}

In the next step, we incorporate the inductive computation of~\ARs (see
Section~\ref{sec:dlltcqs-r-sat-rew-tuple}) into $\prefo(i)$,
yielding the FO formulas $\prefth{j}(i)$ for all $j\in[0,\NBC]$ \ctwo{(\ie $|\BCr(\aont)|+1$ formulas)}:
\begin{itemize}
  \item $\prefth{0}(i)$ is obtained from $\dllhhornPerfRef\acq\aont$ by
    replacing all rigid basic concept and role atoms $\alpha(\vec{t})$ (where
    $\vec{t}$ is either~$t_1$ or $(t_1,t_2)$, depending on the type of~$\alpha$)
    by
    \[ \Asf(\vec{t},i)\lor
    \bigvee_{\alpha(\vec{a})\in\ARs_0}
      \vec{t}=\vec{a},
    \]
    where, if $\alpha$ is a basic concept~$B$, then \ctwo{$\Asf$ denotes~\B}, and if $\alpha$ is a role~$R$, then \ctwo{$\Asf$ denotes~\R}.
    The big disjunction over the component-wise equality $\vec{t}=\vec{a}$ encodes that the atom~$\alpha(\vec{t})$ is satisfied by~$\ARs_0$.
  \item $\prefth{j+1}(i)$ for $j\in[0,\NBC-1]$ is obtained from
    $\dllhhornPerfRef\acq\aont$ by replacing all rigid atoms $\alpha(\vec{t})$
    by
    \[ \Asf(\vec{t},i)\lor\exists p.\prefth[\alpha(\vec{t})]{j}(p). \]

    The second disjunct encodes that
    $\langle\Omc,\ARs_{j}\cup\Amc_p\rangle\models\ground(\alpha(\vec{t}))$
    for some $p\in[0,n]$, which corresponds to the definition
    of~$\ARs_{j+1}$ in Section~\ref{sec:dlltcqs-r-sat-rew-tuple}.
\end{itemize}
The rewriting $\prefth{}(i):=\prefth{\NBC}(i)$ can be shown to be
correct by induction on~$j$.

%
%
%
\begin{restatable}{lemma}{LemPrefThreeARs}
	\label{lem:pref3-ars}
	For all CQs~\acq, groundings~\ground, and $i\in[-1,n]$, we have
  $$\langle\aont,\ARs\cup\Amc_i\rangle\models\ground(\acq) \text{ iff }
  \tdbfbs\models\ground\big(\prefth{}(i)\big).$$
\end{restatable}

\subsubsection{From $\ARs\cup\Amc_i$ to $\AKRs\cup\Amc_i$}

The final rewriting needs to consider an ABox of the form $\AKRs\cup\Amc_i$, where
\[
  \AKRs:=\ARs\cup\rigcons{\QRs}\cup\Amc_{Q_\ax}\cup\Amc_{\RFs}
\]
is the part that is independent of~$i$ (\cf Section~\ref{sec:r-complete-def}).
Note that we use~$Q_\ax$ instead of~$Q_{\iota(i)}$ since we do not explicitly
consider a mapping~$\iota$ yet; and, as in Section~\ref{sec:dlltcqs-cc}, we assume
for now that a world \ax (\eg $\ax_{\iota(i)}$) is given explicitly.
\textcolor{blue}{
It should be kept in mind that the rewriting depends on~\ax (as well as on \as and \Bphi), although we do not explicitly specify this in the notation. 
}

The rewritings introduced so far only cover the individual names in $\NI(\Kmc)$.
However, \AKRs also contains auxiliary individual names from $\NIA\cup\NIT$,
which do not occur in~\tdbfbs.
The set~\NIA and those individuals in~\NIT stemming from~$\Amc_{\RFphi}$
and~$\Amc_{\RFaux}$ do not depend on the input ABoxes, and hence are relatively
unproblematic.
However, $\Amc_{\RFother}$ contains the ABoxes~$\Amc_{\exists S(a)}$, whose
number is not bounded in the size of~\aont or~\atcq.
Our next goal is thus to separate $\Amc_{\RFother}$ as much as possible from the
input data.

We introduce \emph{prototypes},\footnote{Similar in function to, but not to be
confused with, the prototypical elements in canonical interpretations.} which
are fresh individual names ${[S]},\ael{[S]S},\ael{[S]S\apath},\dots$, with the
intention that $[S]$ is used to replace the concrete individual names from
$\NI(\atkb)$.
We collect all these new names except $[S]$ in the set~\NIP.
The ABoxes $\Amc_{\exists S}$ for all $S\in\NFRM(\aont)$ are prototypical
versions of $\Amc_{\exists S(b)}$ with $ b\in\NI(\atkb)$ (see
Section~\ref{sec:rf}), and are obtained from $\Amc_{\exists S(b)}$ by replacing
the individual name~$b$ everywhere by~$[S]$, \eg $\ael{b\apath}$ becomes
$\ael{[S]\apath}$.
In the rewriting, we can then use~$\Amc_{\exists S}$ to refer
to~$\Amc_{\exists S(b)}$ without mentioning~$b$ explicitly.
In the following, we denote by \NITm the restriction of \NIT to those individual
names that occur in~$\Amc_{\RFphi}$ and~$\Amc_{\RFaux}$, \changed{and we denote by $\NIp:=\NIA\cup\NITm\cup\NIP$ the set of all additional individual names we consider in the following (apart from the original ones in $\NI(\atkb)$).}

We can now continue to extend the rewriting $\prefth{}(i)$ to accommodate the
remaining parts of \AKRs, \ie $\rigcons{\QRs}$, $\Amc_{Q_\ax}$,
and~$\Amc_{\RFs}$.
As before, our goal is a rewriting $\rewPRef\acq$ that reflects entailment \wrt
$\langle\Omc,\AKRs\cup\Amc_i\rangle$.
%

We again start from the original rewriting $\dllhhornPerfRef\acq\aont$, which is
a UCQ.
Since $\tdbfbs$ only contains the individual names from
$\NI(\Kmc)$, we adapt the existential quantifiers in the CQs to simulate
quantification over the extended set $\NI(\Kmc)\cup\NIp$ as
follows.
We consider each CQ $\omega=\exists x_0. \dots \exists x_{\ell-1}.\acqtwo$ in
$\dllhhornPerfRef\acq\aont$ separately.
The idea is to expand~$\omega$ into a disjunction of~$2^\ell$ variants
$\omega_0,\ldots,\omega_{2^\ell-1}$ that cover all cases of the variables~$x_j$
being mapped either to $\NI(\Kmc)$ or to $\NIp$.
More formally, to define the disjunct~$\omega_k$, $k\in[0,2^\ell-1]$, we first
represent the number~$k$ by the binary vector
$(b_0,\dots,b_{\ell-1})\in\{0,1\}^\ell$, \ie such that
$k=b_0\cdot 2^0+\ldots+b_{\ell-1}\cdot 2^{\ell-1}$.
Then, for a variable~$x_j$, $j\in[0,\ell-1]$, we replace the original
quantifier~$\exists x_j$ from~$\omega$ by the expression~$\exists'x_j$, which is
defined as follows:
\begin{itemize}
  \item if $b_j=1$, it remains~$\exists x_j$; and
  \item if $b_j=0$, it is the disjunction
    $\bigvee_{x_j\in\NIp}$.
\end{itemize}
Observe that we use the symbol~$x_j$ now in two different ways. If $b_j=0$, then $x_j$ is a variable, as before. However, if $b_j=1$, then $x_j$ is an element of~$\NIp$; the big disjunction over all these elements simulates the quantification over the additional sets of individual names.
We now set
\[ \omega_k:=\exists'x_0.\dots\exists'x_{\ell-1}.
\rep{\acqtwo}\land\acqtwo_{\mathsf{filter}}, \]
where it remains to define the formula $\rep{\acqtwo}\land\acqtwo_{\mathsf{filter}}$.
The idea is that $\rep{\acqtwo}$ replaces the atoms of~$\psi$ in a similar way to the
rewritings considered before. Since atoms in~$\psi$ can refer to prototypes, the
additional formula $\acqtwo_{\mathsf{filter}}$ is needed to ensure that the
structure of the ABoxes $\Amc_{\exists S(b)}\subseteq\Amc_{\RFother}$ is
respected (see Example~\ref{exa:filter} below).
This is inspired by a technique described in~\cite{KLTWZ-IJCAI11:combined}.

The formula~$\rep{\acqtwo}$ is constructed by replacing every atom~\aatom
in~$\acqtwo$ by~$\rep{\aatom}$, depending on the form of~\aatom as described
below.
To simplify the notation, here, we do not mention the parameters \as, \ax,
\Bphi, and~$k$, on which the operation~$\rep{\cdot}$ implicitly depends.
First, we define the abbreviation
\[
\repo{\tone=\ttwo}:={}
\begin{cases}
\tone=\ttwo	& \text{if $\tone,\ttwo\notin\NIp$,}\\
\true & \text{if $\tone=\ttwo$,}\\
\false & \text{otherwise.}
\end{cases}
\]
for $s,t\in\NT(\omega)$, which allows us to express equality between two terms.

For all concept and role atoms~$\alpha(\vec{t})$, we now define
\[ \rep{\alpha(\vec{t})} := \repone{\alpha(\vec{t})}\lor\reptwo{\alpha(\vec{t})}\lor\repthree{\alpha(\vec{t})}. \]
The formulas $\repone{\alpha(\vec{t})}$ and $\reptwo{\alpha(\vec{t})}$
are the rewritings of~$\alpha(\vec{t})$ relative to the ABoxes 
$\ARs\cup\Amc_i$ and
$\rigcons{\QRs}\cup\Amc_{Q_\ax}\cup\Amc_{\RFaux}\cup\Amc_{\RFphi}$,
respectively, and are defined as below:
\begin{align*}
\repone{\alpha(\vec{t})}:={}&\begin{cases}
  \false & \text{if $\vec{t}$ contains elements of $\NIp$,}\\
  \prefth[\alpha(\vec{t})]{}(i) & \text{if $\alpha$ is rigid,}\\
  \Asf(\vec{t},i) & \text{if $\alpha$ is flexible,}\\
\end{cases}\\
\reptwo{\alpha(\vec{t})}:={}&\bigvee_{\alpha(\vec{a})\in\rigcons{\QRs}\cup\Amc_{Q_\ax}\cup\Amc_{\RFaux}\cup\Amc_{\RFphi}}\repo{\vec{t}=\vec{a}} .
\end{align*}
It remains to simulate the influence of~$\Amc_{\RFother}$ via
$\repthree{\alpha(\vec{t})}$.
We start with the formula
\[ \rep[x]{\exists S}:= \exists p.\prefth[\exists S(x)]{}(p) \land
  \bigwedge_{a\in\NI(\atcq)}(x\neq a) \]
which expresses that the variable~$x$ is bound to an individual name~$a$ with
$\exists S(a)\in\RFother$ (see Section~\ref{sec:dlltcqs-r-sat-rew-tuple}).
%
%
We now define
\begin{align*}
  \repthree{\alpha(\vec{t})} &:= \begin{cases}
    \exists x.\rep[x]{\exists S}
      & \text{if $t_2=a_{[S]\rho}$ and
        $\alpha(\vec{t})\in\Amc_{\exists S}$,} \\
    \rep[t_1]{\exists S}
      & \text{if $t_2=a_{[S]R}$,
        $\alpha([S],a_{[S]R})\in\Amc_{\exists S}$,
        $t_1\notin\NIp$} \\
    \false
      & \text{otherwise.}
  \end{cases}
\end{align*}
Intuitively, an assertion $\alpha(\vec{t})$ involving an individual name
from~\NIP holds whenever it directly follows from some
$\Amc_{\exists S}$ (which is independent of the input ABoxes),
or it is a role atom
$R(t_1,a_{[S]S})$ and $t_1$ is mapped to the root~$b$ of some
$\Amc_{\exists S(b)}\subseteq\Amc_{\RFother}$ such that
$R(b,a_{bS})\in\Amc_{\exists S(b)}$ (which corresponds to the assertion
$R([S],a_{[S]S})\in\Amc_{\exists S}$).
%
In both cases, the formula needs to check that a relevant ABox
$\Amc_{\exists S(b)}$ is actually part of $\Amc_{\RFother}$.

It remains to define $\psi_\filter$, whose purpose is to ensure that the
structure of $\Amc_{\RFother}$ is preserved, even though its elements cannot be
explicitly mentioned in the rewriting.

\begin{example}\label{exa:filter}
Consider the CQ $\omega=\exists x,y,z.\psi$, where $\psi=S(y,x)\land S(z,x)$ and
$S$ is rigid, and the disjunct
\[ \omega_k=\bigvee_{x\in\NIp}
	\exists y,z.\rep{\psi}\land\psi_\filter \]
of the rewriting.
In particular, we consider the disjunct of~$\omega_k$ where~$x$ is considered to
be equal to~$\uel{[S]S}\in\NIP$.
It addresses the case where both atoms in~$\psi$ are satisfied by an ABox
of the form~$\Amc_{\exists S(b)}$, by mapping~$x$ to~$\uel{b[S]}$ and both~$y$
and~$z$ to~$b$ (note that $y$ and~$z$ are still quantified over $\NI(\atkb)$).
Since $x=\uel{[S]S}$ is a prototype, we must have
$b\in\NI(\atkb)\setminus\NI(\atcq)$, and thus $\rep{S(y,x)\land S(z,x)}$ is
equal to $\rep[y]{\exists S}\land\rep[z]{\exists S}$.
This formula expresses that both~$y$ and~$z$ must be mapped to roots of an ABox
of the form $\Amc_{\exists S(b)}$, and hence neglects the fact that $y$ and~$z$
must actually be mapped to the \emph{same} individual name~$b$.
The formula $\psi_\filter$ addresses this issue by adding the atom~$y=z$ to the
rewriting.
\end{example}

Formally, we define
\[ \psi_\filter :=
  \ctwo{\bigwedge\big\{
  \repo{s=t} \mid R(x_j,s),\ S(x_j,t)\in\psi,\ x_j\in\NIP,\ s,t\notin\NIP \big\}} \]
(\cf \cite{KLTWZ-IJCAI11:combined}).
Hence, any two terms that are not prototypes and occur together with the same
prototype in role atoms of~$\omega$ must be mapped to the same individual name
in $\NI(\Kmc)\setminus\NI(\atcq)$.
As described above, we construct $\rewPRef\acqalpha$ by replacing each
CQ~$\omega$ in $\dllhhornPerfRef\acq\aont$ by
$\omega_0\lor\dots\lor\omega_{2^\ell-1}$.
In the same way, we obtain the formula $\rewQUnsat$
from $\dllhhornQUnsat\aont$.
The next lemma establishes the correctness of this translation.
\begin{restatable}{lemma}{lemdlltcqsRews}\label{lem:dlltcqs-rews}
  For all Boolean CQs~$\acqalpha$
  and $i\in[-1,n]$, we have:
	\begin{itemize}
		\item $\langle\aont,\AKRs\cup\Amc_i\rangle$ is inconsistent iff 
		$\tdbfbs\models\rewQUnsat$.
		\item $\langle\aont,\AKRs\cup\Amc_i\rangle\models\acqalpha$ iff 
		$\tdbfbs\models \rewPRef\acqalpha$.
	\end{itemize}
\end{restatable}

\subsection{Rewriting r-Satisfiability}
\label{sec:dlltcqs-r-sat-rew-rews2}

We can now use the above rewritings to capture r-satisfiability via the
r-completeness conditions.
Given $\as\subseteq 2^\pv$,
$\Bphi\subseteq\{B(a)\mid B\in\BC(\aont),\ a\in\NI(\atcq)\}$, and a single world $\ax\in\as$, define
$\rsat(i) := \cons(i)\land\ansx(i) \land\answit(i)$, where
\ctwo{%
\begin{align*}
	\cons(i) &:= \lnot\rewQUnsat, \\
	\ansx(i) &:= \bigwedge_{p_j\in\overline{\ax}} \lnot \rewPRef{\acqalpha_j}, \\
	\answit(i) &:=
	\bigwedge \big\{ \lnot \rewPRef{\acqtwo} \mid \acqtwo\text{ rigid witness query for }\QRsn \big\},
\end{align*}%
}%
\changed{in which all rewritings are \wrt $W$ (which is not mentioned explicitly in the notation).}

\ctwo{We can finally prove Lemma~\ref{lem:fo-r-sat1}, which we state here again for convenience.}
  \LemFORSat*
	\begin{proof}
		We consider Definition~\ref{def:dlltcqs-r-complete}.
		Conditions~\ref{def:dlltcqs-rc:cqcons} and~\ref{def:dlltcqs-rc:qrn} are
		trivially satisfied.
		Lemmas~\ref{lem:dlltcqs-rews} and \ref{lem:dllhhorn-perfect-ref-correct}
		show that \ref{rew:a} and~\ref{rew:b} take care of
    Conditions~\ref{def:dlltcqs-rc:consistent}, \ref{def:dlltcqs-rc:negcqs},
		and~\ref{def:dlltcqs-rc:witnesses}.
		It remains to prove Condition~\ref{def:dlltcqs-rc:rf}. For \RFaux and
		\RFother, we show this in the proof of Lemma~\ref{lem:fo-r-sat2} in the
		appendix, even independent of~\ref{rew:c}. For~\RFphi, we show that
		\ref{rew:c} is equivalent to the corresponding part of
		Condition~\ref{def:dlltcqs-rc:rf}. For this, consider any $a\in\NI(\atcq)$.

		($\Leftarrow$)
		If \ref{rew:c} holds, then the definition of \RFphi based on \Bphi and Lemmas~\ref{lem:dlltcqs-rews} and \ref{lem:dllhhorn-perfect-ref-correct} yield that $\exists S(a)\in\RFphi$ iff
		there is an $i\in[0,n]$ such that
		$$\langle\aont,\ARs\cup\rigcons{\QRs}\cup\Amc_{Q_{\iota(i)}}\cup
		\Amc_i\cup\Amc_{\RFphi}\rangle\models \exists S(a),$$
		since \RFaux and \RFother do not contain relevant assertions.
		However, by Lemma~\ref{lem:dll-io-elements}, all parts of $\Amc_{\RFphi}$ relevant to obtain the conclusion $\exists S(a)$ are contained in \ARs, given the definition of these ABoxes. Since $\Amc_{\RFphi}$ does not contain basic concept assertions over $\NI(\atcq)$, we obtain that $\exists S(a)\in\RFphi$ iff there is an $i\in[0,n]$ such that
		$\langle\aont,\ARs\cup\rigcons{\QRs}\cup\Amc_{Q_{\iota(i)}}\cup\Amc_i\rangle
		\models \exists S(a)$, as required.
		
		($\Rightarrow$)
		This follows from the definition of~\RFphi and Lemmas~\ref{lem:dllhhorn-perfect-ref-correct} and~\ref{lem:dlltcqs-rews}.
	\end{proof}

\section{Data Complexity} 
\label{sec:dlltcqs-dc}

Based on the FO rewritability of r-satisfiability, we now show that the low data complexity of query answering in \DLLite does not increase dramatically in our temporal setting and prove \ALogTime-completeness. Nevertheless, FO rewritability is lost.
The lower bound holds already for \DLLitecore without rigid names, which can be shown by reducing the word problem of deterministic finite automata to TCQ entailment, by translating the construction of \cite[\citethm9]{AKKRWZ-IJCAI15:omtqs} to our setting.

\begin{restatable}{theorem}{dlltcqsDCLB}\label{thm:dlltcqs-dc:lb}
 TCQ entailment in~\DLLitecore is \ALogTime-hard in data complexity,
 even if $\NRC=\emptyset$ and $\NRR=\emptyset$.
\end{restatable}
We show \ALogTime-membership by describing an alternating Turing machine that solves our problem in logarithmic time.
For t-satisfiability, we need additional notation and auxiliary results, which are described next.

\subsection{Separating the LTL Satisfiability Test}
\label{sec:ltl-sep}

As before, we consider a TCQ \atcq and a \DLLitehhorn TKB
$\atkb=\langle\aont,\afbs\rangle$ with $\afbs=(\afb_i)_{0\le i \le n}$.
Similar to Algorithm~\ref{algo:tq-sat}, we do not consider the t-satisfiability test for $\pa\atcq$ as a black box, but rather split it into multiple parts, which are then integrated with the test for r-satisfiability using the rewritings of the previous section.
By Lemma~\ref{lem:sep}, we can assume $\pa\atcq$ to be separated, \ie that no future operator occurs in the scope of a past operator and vice versa.
A subformula of $\pa\atcq$ is a \emph{top-level future formula} (\emph{top-level past formula}) if it is of the form $\Next\altlform$ or $\altlform\Until\altlformtwo$ ($\Previous\altlform$ or $\altlform\Since\altlformtwo$) and occurs in $\pa\atcq$ at least once in the scope of no other temporal operator; we denote the set of all such formulas and their negations by \afformset (\apformset), and assume without loss of generality that all propositional variables from~\pv occur in both \afformset and \apformset.
Since we require $\pa\atcq$ to be satisfied at time point~$n$, the crucial part
of this formula thus concerns the past formulas in~\apformset, whose
satisfaction depends on the number~$n$.
The goal is to separate this dependency as much as possible.

%
The \emph{Boolean abstraction}~$\pba\atcq$ of~$\pa\atcq$ is obtained by replacing the top-level future and past formulas $\atopform_1,\dots,\atopform_o$ of~$\pa\atcq$ by propositional variables $q_1,\dots,q_o$, respectively.
%
%
We consider the set \Vmc of all valuations $v\colon\{q_1,\ldots,q_o\}\to\{\true,\,\false\}$ of these variables for which $v(\pba{\atcq})\equiv\true$.
For $v\in\Vmc$, the set
\changed{$\Fmc^v := \{\atopform_i\in\Fmc \mid v(q_i)=\true\}\cup\{\lnot\atopform_i\in\Fmc \mid v(q_i)=\false\}$}
collects the induced future subformulae of~$\pa\atcq$, and $\Pmc^v$ can be defined similarly.
Given a set of worlds $\as\subseteq 2^\pv$ and $v\in\Vmc$, the set $\atmfut{\as}{v}\subseteq\as$ contains the worlds that can serve as the start of an LTL model of~$\Fmc^v$ (restricted to~\as):
\[ \atmfut{\as}{v} := \{ w_0 \mid
  \text{ there is }\altlint=(w_i)_{i\ge0}\text{ such that }\altlint,0\models\Fmc^v \text{ and, for all }
  i\ge 0, \ w_i\in\as \}. \]
%
%
%
%
%
All of these sets are independent of the data and can hence be considered constant.
%

\begin{lemma}\label{lem:dlltcqs-atm-split}
Let $\as=\{\ax_1,\dots,\ax_k\}\subseteq 2^\pv$ and $w_0,\dots,w_n\in\as$. The
following are equivalent.
\begin{enumerate}[label=(\alph*)]
	\item\label{split:a} There is an LTL structure~\altlint that only contains worlds from \as, starts with $w_0,\dots,w_n$, and satisfies $\altlint,n\models\pa{\atcq}$.
  \item\label{split:b} There is a valuation $v\in\Vmc$ such that $w_n\in\atmfut{\as}{v}$ and $(w_0,\dots,w_n,w_n,\dots),n\models\Pmc^v$.
\end{enumerate}
\end{lemma}
\begin{proof}
  ($\Rightarrow$) Given \altlint, $v$ can be obtained by checking which elements of $\{\atopform_1,\dots,\atopform_o\}$ are satisfied at~$n$; then, the LTL structure needed to justify $w_n\in\atmfut\as v$ is defined as the substructure of~\altlint that starts at~$n$. Since the satisfaction of the past formulas $\Pmc^v$ in the structure $(w_0,\dots,w_n, \dots)$ at~$n$ does not depend on any time point after~$n$, the remaining worlds can be chosen arbitrarily.

  ($\Leftarrow$) \altlint can be constructed by joining $(w_0,\dots,w_n)$ and the LTL structure obtained from the fact that $w_n\in\atmfut{\as}{v}$, since the satisfiability of past (future) subformulas at~$n$ is not affected by the worlds after (before) that time point.
\end{proof}

As mentioned above, the critical part is to test whether
$(w_0,\dots,w_n,w_n,\dots),n\models\Pmc^v$ holds, since that depends on~$n$.
However, we can employ Lemma~\ref{lem:ltl-periodic-model} to separate the time
points from each other.
Since $\Pmc^v$ does not contain any future operators, its satisfaction depends
only on the time points before~$n$, and we do not have to be concerned with
finding a period or satisfying $\Until$-formulas.
Our task is thus to find types $\atype_0,\dots,\atype_n\in\atypeset$ such that
\begin{itemize}
  \item $\atype_0$ is initial and $\Pmc^v\subseteq\atype_n$;
  \item for all $i\in[0,n-1]$, the pair $(\atype_i,\atype_{i+1})$ is compatible;
  \item for all $i\in[0,n]$, the world $\atype_i\cap\pv$ belongs to~\as.
\end{itemize}
We can then use the world $w_n=\atype_n\cap\pv$ induced by the last type to
check satisfiability of~$\Fmc^v$, \ie whether $w_n\in\atmfut{\as}{v}$.

\subsection{An Alternating Logarithmically Time-Bounded Turing Machine}

Based on this abstraction, we describe an alternating Turing machine (ATM)~\cite{alternation} that solves the TCQ satisfiability problem in \DLLitehhorn in logarithmic time, in the size of the ABox sequence~\afbs.
We use a random access model, where the read-only input tape is accessed by writing the address of the symbol to be read (in binary) on a specific address tape.
Next to those two tapes, the machine may use a constant number of work tapes.
If $l$ is the size of the input, such machines can add, subtract, compare, and compute the logarithm of numbers with $\mathcal{O}(\log l)$ bits \cite[\citelem7.1]{BaISt90:nc1uniformity}.

As usual, our ATM \atm deciding satisfiability of \atcq w.r.t.~\atkb is based on
Lemma~\ref{lem:tcq-sat-iff}, where for t-satisfiability of~\as and~$\iota$ we
only need to find types $\atype_0,\dots,\atype_n$
(Lemmas~\ref{lem:ltl-periodic-model} and~\ref{lem:dlltcqs-atm-split}), and the
r-satisfiability is checked via FO rewritings based on an additional set
$\Bphi\subseteq\{B(a)\mid B\in\BC(\aont),\ a\in\NI(\atcq)\}$
(see Lemmas~\ref{lem:dlltcqs-iff-s-r-consistent} and~\ref{lem:dlltcqs-fo-r-sat}).
The sets \as and \Bphi (of constant size) are guessed in the beginning, and the
mapping~$\iota$ can be obtained from the types $T_0,\dots,T_n$ guessed during
the computation of~\atm.

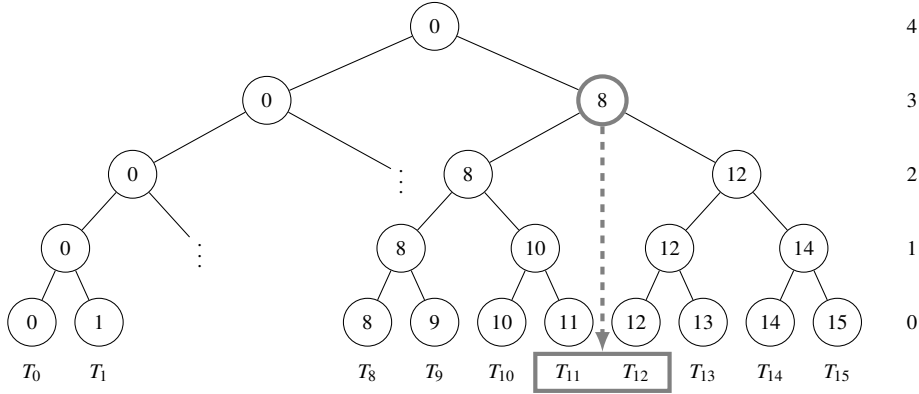
\begin{figure}[tb!]
	\centering
	\resizebox{\columnwidth}{!}{
		\begin{tikzpicture}[level distance=10mm]
		\tikzstyle{level 1}=[sibling distance=45mm]
		\tikzstyle{level 2}=[sibling distance=36mm]
		\tikzstyle{level 3}=[sibling distance=18mm]
		\tikzstyle{level 4}=[sibling distance=9mm]
    \tikzstyle{split}=[circle,draw,minimum size=2.3em]
    \tikzstyle{marked}=[ultra thick,gray,draw]
		\node [split] (z){$0$}
		child {node [split] (a) {$0$}
			child {node [split] (b) {$0$}
				child {node[split] {$0$}
					child {node [split,label={[label distance=.5em]below:$T_0$}] {$0$}}
					child {node [split,label={[label distance=.5em]below:$T_1$}] {$1$}}
				}
				child {node {$\vdots$}}
			}
			child {node {$\vdots$}}
		}
		child {node [split,marked] (j) {{\color{black}$8$}}
			child {node [split] (k) {$8$}
				child {node [split] {$8$}
					child {node [split,label={[label distance=.5em]below:$T_8$}] {$8$}}
					child {node [split,label={[label distance=.5em]below:$T_9$}] {$9$}}
				}
				child {node [split] {$10$}
					child {node [split,label={[label distance=.5em]below:$T_{10}$}] {$10$}}
					child {node [split,label={[label distance=.5em]below:$T_{11}$}] (rl) {$11$}}
				}
			}
			child {node [split] (l) {$12$}
				child {node [split] {$12$}
					child {node [split,label={[label distance=.5em]below:$T_{12}$}] (rr) {$12$}}
					child {node [split,label={[label distance=.5em]below:$T_{13}$}] {$13$}}
				}
				child {node[split] (c) {$14$}
					child {node [split,label={[label distance=.5em]below:$T_{14}$}] {$14$}}
					child {node [split,label={[label distance=.5em]below:$T_{15}$}] {$15$}
						child [grow=right] {node {$0$} edge from parent[draw=none]
							child [grow=up] {node {$1$} edge from parent[draw=none]
								child [grow=up] {node {$2$} edge from parent[draw=none]
									child [grow=up] {node {$3$} edge from parent[draw=none]
										child [grow=up] {node (u) {$4$} edge from parent[draw=none]
											child [grow=up] {node {$\ell$} edge from parent[draw=none]}}
									}
								}
							}
						}
					}
				}
			}
		};
		\path (rl) -- (rr) node [midway,below=0.4cm,marked,rectangle,minimum width=1.8cm,minimum height=0.5cm] (x) {};
		\path[marked,dashed,-latex] (j) -- (x);
		
		\end{tikzpicture} }
	\caption{A sketch of the computation of \atm for $n=15$.}
	\label{fig:atm}
\end{figure}
\begin{example}\label{ex:dlltcqs-atm}
Figure~\ref{fig:atm} gives an overview of the computation tree of \atm, given an ABox sequence with $n=15$.
The nodes represent points at which the alternating machine splits into two copies.
Each node is responsible for constructing a subsequence of~$\atype_0,\dots,\atype_{15}$ for which only the first and last types are given.
For the root node, this means that we initially guess $\atype_0$ and $\atype_{15}$ in order to start this process.
Given the label~$i$ of a node, and its \emph{level} $\ell$, it is responsible for the subsequence starting at index~$i$ and ending at $i+2^\ell-1$.

The root node is labeled by $i=0$ and has level $\ell=4$, which makes it responsible for the subsequence from~$\atype_0$ to~$\atype_{15}$, 
\ie the full sequence.
It then delegates this responsibility to its successors at level~$3$ in the following way: it guesses the types~$\atype_7$ and~$\atype_8$ in the exact middle of the sequence and verifies their t-compatibility, and then it splits the sequence in half. The machine splits into two copies, each of which is responsible for one half of the remaining computation.
The \emph{left} successor deals with the sequence from~$\atype_0$ to~$\atype_7$, 
where we already know the first type and the last type.
Correspondingly, the \emph{right} successor (marked in gray the figure) is labeled by $i=8$, because its designated subsequence starts at $\atype_8$ and ends at $\atype_{15}$. 
Again, we already know the types $\atype_8$, $\atype_{15}$ for the start and end points.
In turn, this copy of the machine then guesses a t-compatible pair $(\atype_{11},\atype_{12})$ (also marked in the figure), and splits the subsequence again into the two shorter sequences $\atype_7,\dots,\atype_{11}$ and $\atype_{12},\dots,\atype_{15}$.

Since all copies of \atm proceed in this way, those at level $\ell=1$ consider only two types $\atype_i$, $\atype_{i+1}$ that have already been guessed before. Each copy then verifies the t-compatibility of this pair of types.
Finally, the copies at level~$0$ each know only one type.
Throughout the whole computation, each type~$\atype_i$ is guessed only once, which prevents conflicting guesses for one time point. Moreover, the copies require no knowledge about what happens in other branches of the computation tree.
\end{example}

\begin{algorithm}[tb]
	\fontsize{8pt}{11pt}\selectfont
	\KwIn{TCQ~\atcq, TKB~$\langle\aont,(\afb_i)_{0\le i\le n}\rangle$}
	\KwOut{\true if \atcq is satisfiable w.r.t.\ \atkb (i.e., the ATM accepts the input), otherwise \false}
	\BlankLine
	$v :=$ Guess an element of $\Vmc$\;
	$\as:=$ Guess a subset of $ 2^\pv$\;
	$\Bphi:=$ Guess a subset of $\{B(a)\mid B\in\BC(\aont),\ a\in\NI(\atcq)\}$\;
	$T_l:=$ Guess an initial type in \atypeset \tcp*{Lemma~\ref{lem:dlltcqs-atm-split}}
	$T_r:=$ Guess a type in \atypeset that contains~$\Pmc^v$\;
  \lIf{\upshape $T_r\cap\{p_1,\dots,p_m\}\notin\atmfut{\as}{v}$}{\Return{\false}}
	\ForEach(\tcp*[f]{Lemma~\ref{lem:dlltcqs-fo-r-sat}\ref{rew:b}}){$\ax\in\as$}{
		\lIf{\upshape $\tdbfbs\not\models\rsat(-1)$
		}{{\Return{\false}}}
	}
	\tcp{start recursion with $i=0$, $\ell=\log(n+1)$, and $\Bphi'=\Bphi$}
	\Return{\textnormal{\algofont{\algoATMRecursion}%
			\big(\atcq, \atkb, $v$, \as, \Bphi, \Bphi, $T_l$, $T_r$, $\log(n+1)$, $0$\big)
	}}
	
	\caption{The ATM}
	\label{algo:atm1}
\end{algorithm}
The copies at level $\ell=0$ are each responsible only for one type $\atype_i$, which induces the world~$w_i=\atype_i\cap\pv$ that implicitly corresponds to $W_{\iota(i)}$ in Lemma~\ref{lem:dlltcqs-fo-r-sat}. We can thus apply this lemma to check the r-completeness conditions.
By this lemma, we have to check the satisfaction of FO formulas in \tdbfbs, which can be done in \ACzero \cite[\citethm9.1]{BaISt90:nc1uniformity}, a subclass of \LogTime.
There are two points that deserve special attention.
First, Condition~\ref{rew:b} in Lemma~\ref{lem:fo-r-sat1}, refers to satisfaction problems \wrt the empty ABox ($i=-1$) for all elements of \as. To this end, \atm splits into $|\as|$ (constantly many) further copies that then verify the corresponding problems.
Second, Condition~\ref{rew:c} in Lemma~\ref{lem:fo-r-sat1} imposes a global condition over all time points $i\in[0,n]$. Therefore, \atm additionally guesses, for each $\exists S(a)\in\Bphi$, at which time point~$i$ we have $\tdbfbs\models\rewPRef{\exists S(a)}$, where this rewriting is \wrt the world~$W=w_i$.
The elements of \Bphi for which a copy of \atm is responsible are then propagated along the branches of the computation tree, and split accordingly.

\changed{The ATM's behavior is specified in Algorithm~\ref{algo:atm1}, where the valuation set $\Vmc$, propositions $\pv$, top-level past formulas~$\Pmc$, and the past formulas~$\Pmc^v$ are constructed based on \atcq as described in the beginning of Section~\ref{sec:ltl-sep}.}
Because of the data complexity assumptions, all constructions depending only on~\atcq and~\aont are of constant size and encoded directly into the states of~\atm.
In addition to~\as, \Bphi, and a valuation $v\in\Vmc$, which are guessed at the beginning, this includes
\begin{itemize}
	\item  a set $\Bphi'\subseteq\Bphi$ that contains the elements of~\Bphi the current copy of the machine is responsible for (see Lemma~\ref{lem:fo-r-sat1}\ref{rew:c});
	\item two types~$T_l$ and~$T_r$ for the left-most and the right-most types of the current subsequence of $T_0,\dots,T_n$.
\end{itemize}
%
%
The (read-only) input tape of \atm contains only the FO structure \tdbfbs, which implicitly contains the number~$n$.
In each configuration, \atm stores the index~$i$ and level~$\ell$ as described in Example~\ref{ex:dlltcqs-atm} on its work tapes, which requires only a logarithmic number of bits.
\changed{
The different ATM configurations are described by the recursive Algorithm~\ref{algo:atm2} (\algofont{\algoATMRecursion}). 
At the end, there are $n+1$ copies of \atm, one for each index $i\in[0,n]$, and each of them knows only one type $T_i$, which induces a unique world $w_i=T_i\cap\pv$.
These copies execute the final tests described in Algorithm~\ref{algo:atm3} (\algofont{\algoATMFinal}), in line with Lemmas~\ref{lem:dlltcqs-fo-r-sat} and~\ref{lem:dlltcqs-atm-split}.}
\begin{algorithm}[tb]
\fontsize{8pt}{11pt}\selectfont
\KwIn{TCQ~\atcq, TKB~$\atkb=\langle\aont,(\afb_i)_{0\le i\le n}\rangle$,
	valuation $v$, worlds \as, sets of basic concept assertions $\Bphi$ and $\Bphi'$,
	types $T_l$ and $T_r$, level $\ell$, time point $i$%
}
\KwOut{\true if all recursively created ATM copies accept, otherwise \false}
\BlankLine
$\Bphi^{(l)}$, $\Bphi^{(r)}:= $ Guess a partition of~$\Bphi'$ into two sets\;
\If{$\ell > 1$}{
	$\ell:=\ell-1$\;
	
	 $(T^{(1)},T^{(2)}) :=$ Guess a t-compatible pair from \atypeset \tcp*{Lemma~\ref{lem:dlltcqs-atm-split}}
	\tcp{split into two copies and accept iff both accept}
	\Return{\textnormal{\algofont{\algoATMRecursion}%
			\big(\atcq, \atkb, $v$, \as, \Bphi, $\Bphi^{(l)}$, $T_l$, $T^{(1)}$, $\ell$, $i$\big)
	  }%
    \upshape{\textbf{and}}\;
		~\phantom{\textbf{return}}\textnormal{\algofont{\algoATMRecursion}%
				\big(\atcq, \atkb, $v$, \as, \Bphi, $\Bphi^{(r)}$, $T^{(2)}$, $T_r$, $\ell$, $i+2^\ell$\big)
		}
  }
}
\Else(\tcp*[f]{$\ell=1$}){
\lIf(\tcp*[f]{Lemma~\ref{lem:dlltcqs-atm-split}}){\upshape $(T_l,T_r)$ are not t-compatible}{\Return{\false}}
%
\tcp{split into two copies one last time}
\Return{\textnormal{\algofont{\algoATMFinal}%
		\big(\atcq, \atkb, $v$, \as, \Bphi, $\Bphi^{(l)}$, $T_l$, $i$\big)}
  \upshape{\textbf{and}} 
	\textnormal{\algofont{\algoATMFinal}%
			\big(\atcq, \atkb, $v$, \as, \Bphi, $\Bphi^{(r)}$, $T_r$, $i+1$\big)
  }
}
}
%
%
\caption{\algofont{\algoATMRecursion}}
\label{algo:atm2}
\end{algorithm}
\begin{algorithm}[tb]
\fontsize{8pt}{11pt}\selectfont
\KwIn{TCQ \atcq, TKB~$\langle\aont,\afbs\rangle$ with $\afbs:=(\afb_i)_{0\le i\le n}$,
	valuation $v$, worlds \as, sets of basic concept assertions \Bphi and $\Bphi'$, type $T_i$, time point~$i$%
}
\KwOut{\true if all final tests are successful, otherwise \false}
\BlankLine
$w_i:=T_i\cap{}\pv$\;
\lIf(\tcp*[f]{Lemma~\ref{lem:dlltcqs-atm-split}}){\upshape $w_i\notin\as$}{\Return{\false}} 

%
\lIf(\tcp*[f]{Lemma~\ref{lem:dlltcqs-fo-r-sat}\ref{rew:a}}){\upshape $\tdbfbs\not\models\rsat[w_i](i)$}{\Return{\false}}
%

\ForEach(\tcp*[f]{Lemma~\ref{lem:dlltcqs-fo-r-sat}\ref{rew:c}, ``if''-direction}){$S\in\NFRM(\aont)$, $a\in\NI(\atcq)$}{
\lIf{$\tdbfbs\models\rewPRef{\exists S(a)}$ 
	{w.r.t.~$w_i$} \upshape\textbf{and} $\exists S(a)\notin\Bphi$}{%
	\Return{\false}}
}

\ForEach(\tcp*[f]{Lemma~\ref{lem:dlltcqs-fo-r-sat}\ref{rew:c}, ``only if''-direction}){$\exists S(a)\in\Bphi'$}{
		\lIf{\upshape $\tdbfbs\not\models\rewPRef{\exists S(a)}$ 
			{w.r.t.~$w_i$}}{%
			{\Return{\false}}
}}
\Return{\true}
\caption{\algofont{\algoATMFinal}}
\label{algo:atm3}
\end{algorithm}
It is easy to show that 
a successful run of \atm indeed reflects the satisfiability of~\atcq w.r.t.~\atkb. 
\begin{restatable}{theorem}{thmdlltcqsDCAalogtime}
\label{thm:dlltcqs-dc-alogtime}
	TCQ entailment in~\DLLitehhorn is in \ALogTime in data
	complexity, even if $\NRR\neq\emptyset$.
\end{restatable}

\section{TCQ Entailment Beyond the Horn Fragment}
\label{sec:exdlltcqs}

Having established the good computational behavior of TCQ entailment in \changed{Horn} fragments of \DLLite, we now consider the more expressive \textit{krom} and \textit{bool} fragments.
Even without role inclusions, it turns out that TCQ entailment in these logics is as hard as for \ALC, which allows to express qualified existential restrictions (but no inverse roles).
With role inclusions, the complexity even increases to the same level as for $\ALC\Imc$ (\ALC with inverse roles) (see Table~\ref{tab:cc}).
As an auxiliary result, we first show that there is no difference between \DLLitekrom and \DLLitebool in our setting, as TCQs can be used to simulate CIs that are usually only expressible in \DLLitebool.

\subsection{Reducing {\it DL-Lite}\texorpdfstring{$_{\mathit{bool}}$}{bool} to  {\it DL-Lite}\texorpdfstring{$_{\mathit{krom}}$}{krom}}

We show that TCQs, together with CIs of the form $\top\sqsubseteq A\sqcup\overline{A}$, can simulate several kinds of CIs that go beyond \DLLitekrom, covering \DLLitebool and even parts of \changed{\ALCI}.
\changed{The description logic \ALCI supports} \emph{qualified existential / value restrictions} of the form $\exists R.C$ / $\forall R.C$, where $R$ is a role and $C$ a concept, with the following semantics:
\begin{align*}
  (\exists R.C)^\Imc &:= \{ x\in\Delta^\Imc \mid \text{there is } y\in\Delta^\Imc\text{ such that } (x,y)\in R^\Imc\text{ and }y\in C^\Imc \} \\
  (\forall R.C)^\Imc &:= \{ x\in\Delta^\Imc \mid \text{for all }y\in\Delta^\Imc\text{ such that } (x,y)\in R^\Imc\text{ implies }y\in C^\Imc \}
\end{align*}
\changed{That is, a qualified existential restriction $\exists R.C$ checks for the existence of an $R$-successor of type~$C$, whereas $\forall R.C$ requires that all $R$-successors are of type~$C$.}

In the following construction, we employ negated CQs to simulate complex CIs. We use (fresh) symbols~$\overline{A}$ to simulate the complements of concept names~$A$.
\begin{table}[tb]
	\centering
  \caption{Representing complex CIs in \DLLitekrom via TCQs.}
	\label{tab:dlltcqs-krom-transf}
	\begin{tabu}{|l|l|}
		\tabtopline
		CI & TCQ 
		\tabmidline
		$\exists R.A_1\sqsubseteq A_2$&
		$\lnot\exists x,y. R(x,y)\land A_1(y)\land\overline{A}_2(x)$\\
		$A_1\sqsubseteq \forall R.A_2$&
		$\lnot\exists x,y. A_1(x)\land R(x,y)\land\overline{A}_2(y)$\\
		$A_1\sqcap\dots\sqcap A_m\sqsubseteq A_{m+1}\sqcup\dots\sqcup A_{m+n}$&
		$\lnot\exists x. A_1(x)\land\dots\land
		A_m(x)\land\overline{A}_{m+1}(x)\land\dots\overline{A}_{m+n}(x)$ 
		\tabbotline
	\end{tabu}
\end{table}
\begin{restatable}{lemma}{lemKromTransf}\label{lem:krom-transf}
	Let $(C\sqsubseteq D,\lnot\acq)$ be one of the pairs of a CI and a TCQ given in Table~\ref{tab:dlltcqs-krom-transf}, and let \aint be a model of $\top\sqsubseteq A_i\sqcup\overline{A}_i$ and $A_i\sqcap\overline{A}_i\sqsubseteq\bot$ for all concept names~$A_i$ occurring in~$D$. Then, we have $\aint\models C\sqsubseteq D$ iff $\aint\models\lnot\acq$.
\end{restatable}
This means that, given a TCQ~\atcq and a TKB $\Kmc=\langle\Omc,(\Amc_i)_{0\le i\le n}\rangle$, we have
$$\langle\aont,(\Amc_i)_{0\le i\le n}\rangle\models\atcq \text{ iff }
\langle\aont',(\Amc_i)_{0\le i\le n}\rangle\models((\Boxm\Boxf\atcqtwo)\to\atcq),\text{ where}$$
\begin{itemize}
\item $\aont'$ is obtained from \aont by removing all CIs of the forms listed in Table~\ref{tab:dlltcqs-krom-transf} and adding the CIs required by Lemma~\ref{lem:krom-transf}, and
\item \atcqtwo is the conjunction of the negated CQs simulating the removed CIs.
\end{itemize}
With the same construction, \atcq is satisfiable w.r.t.\ $\langle\aont,(\Amc_i)_{0\le i\le n}\rangle$ iff $(\Boxm\Boxf\atcqtwo)\land\atcq$ is satisfiable w.r.t.\ $\langle\aont',(\Amc_i)_{0\le i\le n}\rangle$.
This means that we can use all CIs listed in Table~\ref{tab:dlltcqs-krom-transf} also in \DLLitekrom for our purposes. In particular, we have the following corollary.
\begin{corollary}\label{cor:dlltcqs-krombool-is-the-same}
TCQ entailment in \DLLitebool can be \ctwo{logspace-}reduced to TCQ entailment in
\DLLitekrom.
\end{corollary}

This means that it suffices to show our complexity upper bounds for \DLLitekrom, and the lower bounds for \DLLitebool.
But even more than that, we can use CIs with qualified existential restrictions on the left-hand side \changed{(or, equivalenty, value restrictions on the right-hand side)} to prove hardness results for TCQ entailment in \DLLitekrom.
We can even nest these concept constructors arbitrarily.
%
\begin{example}\label{ex:exdlltcqs-gci-trans}
The CI
$ A_1\sqcup A_2\sqcup \exists R_1.A_3 \sqsubseteq
A_4 \sqcup \forall R_1.(A_1 \sqcap \exists R_2) $
can be expressed by the following CIs, assuming $A'_5,\dots,A'_7$ to be fresh concept names:
\begin{align*}
&A_1\sqsubseteq A_4\sqcup A'_5,\ 
A_2\sqsubseteq A_4\sqcup A'_5,\ 
A'_6\sqsubseteq A_4\sqcup A'_5, \\
&\exists R_1.A_3\sqsubseteq A'_6,\ 
A'_5\sqsubseteq \forall R_1.A'_7,\ 
A'_7\sqsubseteq A_1,\ 
A'_7\sqsubseteq \exists R_2.
\end{align*}
These CIs can then, in turn, be simulated by negated CQs as described in Lemma~\ref{lem:krom-transf}.
\end{example}

As usual, we now consider a Boolean TCQ~\atcq and a TKB $\atkb=\langle\aont,(\afb_i)_{0\le i\le n}\rangle$ written in a DL between \DLLitekrom and \DLLitehbool, depending on the context, and investigate the combined and data complexity of TCQ entailment.

\subsection{\ctwo{Combined Complexity}}

Given Corollary~\ref{cor:dlltcqs-krombool-is-the-same}, we directly get two rather strong hardness results from atemporal query answering, \ie even \emph{without rigid names}, from \ExpTime-hardness of UCQ entailment in \DLLitebool and \TwoExpTime-hardness of UCQ entailment in \DLLitehbool~\cite{BoMMP-TODS16}.

\begin{corollary}\label{cor:dlltcqs-hkrom-cc:lb-w/o-rigid}\label{cor:dlltcqs-krom-cc:lb-w/o-rigid}
  TCQ entailment in \DLLitekrom is \ExpTime-hard in combined complexity, which increases to \TwoExpTime-hardness for~\DLLitehkrom, even if $\NRC=\emptyset$ and $\NRR=\emptyset$.
\end{corollary}
Since \DLLitehbool is a sublogic of $\ALC\Hmc\Imc$, for which TCQ entailment with rigid concepts and roles is in \TwoExpTime \cite[\citethm12]{BaBL-AI15}, the \TwoExpTime lower bound is already tight.
To match the \ExpTime lower bound for \DLLitebool without role inclusions, we first establish an auxiliary result for satisfiability of \changed{(atemporal)} conjunctions of CQ literals in \DLLitekrom.

\begin{lemma}\label{lem:dll-cq-literal-conj-sat}
Satisfiability of Boolean conjunctions of CQ literals w.r.t.\ \DLLitekrom KBs is in \ExpTime in combined complexity.
\end{lemma}
\begin{proof}
  By grounding the positive literals using fresh individual names, the problem can be reduced to UCQ non-entailment~\cite{BaBL-JWS15}.
  The complexity bound then follows from the fact that UCQ entailment w.r.t.\ so-called frontier-one disjunctive inclusion dependencies is in \ExpTime~\cite[\citethm8]{BoMP-IJCAI13}, and such dependencies can simulate \DLLitekrom CIs.
\end{proof}

Following the approach of Lemma~\ref{lem:tcq-sat-iff}, this allows us to show the following.

\begin{theorem}\label{thm:dlltcqs-krombool-cc:ub-w/o-rigid}
TCQ entailment in~\DLLitebool is in \ExpTime in combined complexity if $\NRC=\emptyset$ and $\NRR=\emptyset$.
\end{theorem}
\begin{proof}
By Corollary~\ref{cor:dlltcqs-krombool-is-the-same}, it suffices to describe a decision procedure for \DLLitekrom, which can be done by following Lemma~\ref{lem:tcq-sat-iff}.
By \cite[\citelem6.1]{BaBL-JWS15}, we can assume without loss of generality that the TKB is of the form $\langle\aont,\emptyset\rangle$, which means that we do not have to find a mapping~$\iota$.
Since there are no rigid names to enforce dependencies between time points, in order to check the r-satisfiability of a set $\as=\{\ax_1,\dots,\ax_k\}$, it suffices to check the satisfiability of $\chi_i$ for all $i\in[1,k]$ individually (see also~\cite{BaBL-JWS15}).
Hence, we can define~\as as the set of \emph{all} those sets $\ax_i$ for which~$\chi_i$ is satisfiable w.r.t.~\aont. According to Lemma~\ref{lem:dll-cq-literal-conj-sat}, this can be done in exponential time.
Moreover, t-satisfiability of $\pa\atcq$ w.r.t.~\as can also be checked in \ExpTime~\cite{BaBL-JWS15}, and hence we obtain the claim by Lemma~\ref{lem:tcq-sat-iff}.
\end{proof}


In the presence of rigid names, the complexity of TCQ entailment in \DLLitekrom increases.
\begin{theorem}\label{thm:dlltcqs-krom-cc:lb-rigid-concepts}
	TCQ entailment in~\DLLitekrom is \coNExpTime-hard w.r.t.\ combined complexity if $\NRC\neq\emptyset$, even if $\NRR=\emptyset$.
\end{theorem}
\begin{proof}
  The proof is by reduction from the satisfiability problem for $\EL_\bot$-LTL, which is already \NExpTime-hard if no role names are available (neither rigid nor flexible)~\cite{BoT-IJCAI15,BoTh-LTCS-15-07}; \ctwo{$\EL_\bot$ is the DL extending \EL by the $\bot$ constructor.}
  The formulas~$\atcq$ of that language are similar to TCQs, but instead of CQs they use assertions and CIs formulated in \EL (plus~$\bot$) as atomic formulas.
  However, if~\atcq contains no role names, we can assume that all assertions in~\atcq are of the form $A(a)$, where $A\in\NC$ and $a\in\NI$, which can hence be directly treated as CQs.
  Moreover, all CIs in~\atcq are of the form $\top\sqsubseteq A$, $A\sqsubseteq\bot$, or $A_1\sqcap\dots\sqcap A_m\sqsubseteq A_{m+1}$, which can be replaced by equivalent negated CQs according to Lemma~\ref{lem:krom-transf}, if we add certain \DLLitekrom CIs to the global ontology.
  We can hence obtain a TCQ~$\atcq'$ and an ontology~\Omc such that $\atcq'$ is satisfiable w.r.t.~$\langle\Omc,\emptyset\rangle$ iff the $\EL_\bot$-LTL formula~\atcq is satisfiable.
\end{proof}
%
For proving containment in \coNExpTime, we use a technique from~\cite{BaGL-TOCL12,BaBL-JWS15} that is again based on Lemma~\ref{lem:tcq-sat-iff} and additionally guesses an exponential set $\Tmc\subseteq 2^{\NRC(\Kmc)}$ that represents all combinations of rigid concept names that are allowed to occur in a model of the TCQ~\atcq \wrt the TKB~\Kmc.
Using this set, we can separate the satisfiability tests required for r-satisfiability in a similar fashion as in Lemma~\ref{lem:dlltcqs-iff-s-r-consistent}.

\begin{restatable}{theorem}{thmexdllRigidConUB}\label{thm:dlltcqs-krombool-cc:ub-rigid-concepts}
	TCQ entailment in~\DLLitebool is in \coNExpTime in combined complexity if $\NRR=\emptyset$, even if $\NRC\neq\emptyset$.
\end{restatable}


Finally, we prove \TwoExpTime-hardness of TCQ satisfiability in \DLLitekrom in the presence of rigid role names, by reducing the word problem of exponentially space-bounded alternating Turing machines.
Our reduction is based on the \TwoExpTime-hardness proof for \ALC-LTL in~\cite{BaGL-TOCL12}, which we adapt to our setting using ideas from~\cite{KRH-TOCL2013:horndls}.
Recall that containment in \TwoExpTime follows from the corresponding result for $\ALC\Imc$~\cite[\citethm12]{BaBL-AI15}.
\begin{restatable}{theorem}{thmKromRigidLB}\label{thm:krom-cc:lb-rigid-roles}
	TCQ entailment in~\DLLitekrom is \TwoExpTime-hard in combined complexity
	if $\NRR\neq\emptyset$.
\end{restatable}

\tikzstyle{dot}=[circle,fill=black,inner sep=0pt,minimum size=4.5pt]
\tikzset{node distance=1cm} 
\subsection{Data Complexity}


As the final result of this paper, we show that the data complexity of TCQ
entailment in \DLLitehkrom is in \coNP, by showing that satisfiability is
in~\NP.
We follow an approach similar to the r-complete tuples from
Section~\ref{sec:dlltcqs-r-sat-charact}; however, since the goal is~\NP instead
of \ALogTime, we do not need to be as careful with our constructions as in
Sections~\ref{sec:dlltcqs-r-sat-charact}, \ref{sec:dlltcqs-r-sat-rew},
and~\ref{sec:dlltcqs-dc}.
For example, we can simply guess the set~\as and mapping~$\iota$ in constant and
linear time, respectively.
Similarly, in order to separate the satisfiability tests for each
$i\in[0,\dots,n+k]$ as in Section~\ref{sec:dlltcqs-r-sat-charact}, we can guess
one \emph{flexible} ABox type per time point, instead of only one rigid ABox
type.
Of course, the individual ABox types need to agree on the rigid assertions.

\begin{definition}[ABox Type]
	\label{def:exdlltcqs-aboxtype}
	A \emph{(flexible) ABox type} for~\Kmc is a set~\Amc of assertions formulated
	over $\NI(\Kmc)\cup\NIA$, $\NC(\aont)$, and $\NR(\aont)$ such that
	$\lnot\aaxiom\in\Amc$ iff $\aaxiom\notin\Amc$.
\end{definition}

The set \NIA is defined similarly as in Section~\ref{sec:dlltcqs-r-sat-charact}: it
contains an individual name~$\aiel{x}$ for each $i\in[0,n+k]$ and variable~$x$
occurring in a CQ in~$\Q_{\iota(i)}$.
We again consider the ABoxes $\cqinst{\Q_{\iota(i)}}$ that contain the
assertions obtained from the CQs $\acq\in\Q_{\iota(i)}$ by replacing each
variable~$x$ by~$\aiel{x}$.
In contrast to Section~\ref{sec:dlltcqs-r-sat-charact}, we also distinguish the time points~$i$ here.
The reason is that, even if the same CQ is satisfied at two different time points, the elements that satisfy it may behave differently due to the nondeterminism inherent in \DLLitehkrom.

As in Section~\ref{sec:dlltcqs-r-sat-charact}, we must ensure that the
interactions between $\Jmc_0,\dots,\Jmc_{n+k}$, which are caused by the rigid
names, do not lead to the satisfaction of some $\acq_j\in\tcqcqs\atcq$ in the
unnamed part of some~$\Jmc_i$ although we have 
$p_j\in\overline{W_{\iota(i)}}$.
However, it is clear that we cannot guess the whole unnamed part of all the
interpretations $\Jmc_0,\dots,\Jmc_{n+k}$.
Instead, we consider these tree-shaped parts only up to a constant depth, and
moreover abstract from the actual individual names that are the roots of these
trees by considering only their general behavior as it is relevant to the
TCQ~\atcq.

To this end, we define \emph{types}, which sufficiently characterize the
interpretations of individual names and their unnamed successors.
A type captures the basic concepts satisfied at a named individual~$a$, as well
as relevant homomorphisms of CQs from $\tcqcqs\atcq$ into the unnamed successors
of~$a$.
However, it does not refer to the actual individual name itself, and is
therefore independent of the input ABoxes.
A \emph{temporal type} is a set of types, which describe the possible behaviors
of the unnamed parts of the interpretations over time.

\begin{definition}[Type]
	\label{def:type}
	A \emph{basic type}~\Bmc is a set of basic concepts from~$\BC(\aont)$ and
	their negations, such that $B\in\Bmc$ iff $\lnot B\not\in\Bmc$ for all
	$B\in \BC(\aont)$; it induces the set of assertions
	$\Amc_\Bmc(a):=\{ \ctwo{(\lnot)B}(a)\mid \ctwo{(\lnot)B}\in \Bmc\}$.
%
	A \emph{type}~\Tmc is a triple $(\Bmc,\Qmc,\Mmc)$ containing a basic type~\Bmc, a
	set $\Qmc\subseteq\tcqcqs\atcq$ of CQs, and a set
	$\Mmc\subseteq\bigcup_{\acq\in\tcqcqs\atcq}2^{\NT(\acq)}$ of term sets.
	A \emph{temporal type}~$\tau$ is a set of types.
	We denote the set of all temporal types by~\TTypes.
\end{definition}
%
The intuition behind a type $\Tmc=(\Bmc,\Qmc,\Mmc)$ is that it describes the
(negated) basic concepts \Bmc that are satisfied at some individual name~$a$,
the CQs $\Qmc\subseteq\tcqcqs\atcq$ that are satisfied by~$a$ and/or its unnamed
successors, and the sets of terms \Mmc that can be partially satisfied by~$a$
and its unnamed successors, \ie for which there is a partial homomorphism of
these terms and the associated atoms of a CQ $\acq\in\tcqcqs\atcq$ into the
unnamed part of an interpretation that is reachable from~$a$ via role
connections. Observe that the CQs in $\Qmc$ describe CQs satisfied somewhere in
the successor tree of $a$, whereas the sets in \Mmc are used to describe CQs
partially satisfied \emph{at $a$}.

We thus guess, for each $a\in\NI(\Kmc)\cup\NIA$ and $i\in[0,n+k]$, a
type~$\typeai{a}{i}$, which is a polynomial amount of information.
In addition, for each of the constantly many temporal types $\tau\in\TTypes$, we
guess so-called \emph{tree ABoxes}~$\tabox{\tau}{\Tmc}$ (of constant size) that
describe prototypical trees of unnamed successors for each type $\Tmc\in\tau$.
For instance, if \Tmc specifies a CQ $\acq\in\tcqcqs\atcq$ to be \emph{not}
satisfied, then \acq is not satisfied in~$\tabox{\tau}{\Tmc}$.
For a fixed~$\tau$, the ABoxes~$\tabox{\tau}{\Tmc}$ must contain the same
individual names and agree on the rigid assertions, which allows us to test the
satisfaction of the negative CQ literals at each of the $n+k+1$ time points
individually.
\begin{definition}[Tree ABoxes]
\label{def:tree-aboxes}
	Given a temporal type $\tau\in\TTypes$, ABoxes $(\Amc_\Tmc)_{\Tmc\in\tau}$
	are called \emph{tree ABoxes} (for~$\tau$) if, for all
	$\Tmc=(\Bmc,\Qmc,\Mmc)\in\tau$,
	\begin{enumerate}[label=(T\arabic*),leftmargin=3em,series=ta]
		\item\label{ta:rigid}
			we have $\NI(\Amc_\Tmc)=\NI(\Amc_{\Tmc'})$ and
			$\Amc_\Tmc\cap\ASr=\Amc_{\Tmc'}\cap\ASr$ for all $\Tmc'\in\tau$;
		\item\label{ta:elements}
			the elements of $\NI(\Amc_\Tmc)$ are of the form
			$u_{\lambda_1\dots\lambda_\ell}$, where each~$\lambda_i$ is either of the
			form~$R^{\Tmc'}$ for $\Tmc'\in\tau$ and $R\in\NFRM(\Omc)$, or simply a
			rigid role $R\in\NRRM(\Omc)$;
		\item\label{ta:complete}
			\ctwo{for all assertions~$\alpha$ of the form
				$B(u_\rho)$ or $S(u_\rho,u_{\rho\smash{R^{[\ |\Tmc']}}})$ with $B\in\BC(\Omc)$ and $S\in\NRM(\Omc)$}, we have either
		$\alpha\in\Amc_\Tmc$ or $\lnot\alpha\in\Amc_\Tmc$, \ctwo{and $\Amc_\Tmc$ contains no other assertions};%
\footnote{\ctwo{The notation $R^{[\ |\Tmc']}$ refers to the fact that $R$ can be either a rigid role without superscript, or a flexible role with some type~$\Tmc'$ as superscript.}}
		\item\label{ta:consistent}
			the knowledge base $\langle\Omc,\Amc_\Tmc\rangle$ is consistent;
		\item\label{ta:basic-type}
		\ctwo{$\Amc_\Bmc(u_\eps)\subseteq \Amc_\Tmc$};
		\item\label{ta:cqs}
			for all $\acq\in\tcqcqs\atcq$, we have $\acq\in\Qmc$ iff
			$\langle\Omc,\Amc_\Tmc\rangle\models\acq$;
		\item\label{ta:partial-homs}
			for all $S\in\bigcup_{\acq\in\tcqcqs\atcq}2^{\NT(\acq)}$, we have
			$S\in\Mmc$ iff there are a CQ $\acq\in\tcqcqs\atcq$ and a partial
			homomorphism $\ahom\colon\NT(\acq)\to\NI(\Amc)$ of~\acq into~\Amc with
			$\ctwo{u_\eps}\in\range(\ahom)$ and $S=\domain(\ahom)$, \ctwo{where the functions \range and \domain yield the range and domain of a given mapping, respectively.}
	\end{enumerate}
\end{definition}
Condition~\ref{ta:rigid} expresses that the tree ABoxes should respect the rigid
names.
Conditions~\ref{ta:elements} and~\ref{ta:complete} establish the tree shape of
the ABoxes.
Conditions~\ref{ta:complete} and~\ref{ta:consistent} further ensure that the ABoxes are
consistent w.r.t.~\Omc and that nondeterministic choices forced by~\Omc are
realized in them.
Finally, Conditions~\ref{ta:basic-type}--\ref{ta:partial-homs} reflect the
satisfaction of the types~\Tmc in each~$\Amc_\Tmc$.

However, so far we have said nothing about the size of tree ABoxes.
The idea is that they should contain just enough elements to allow us to expand
them into proper models of the TKB~\Kmc.
That is, an element $u_\rho$ without role successors in
$(\Amc_\Tmc)_{\Tmc\in\tau}$ should have an ancestor~$u_\sigma$ that satisfies
the same basic concepts, which allows us to continue the construction of a
model at~$u_\rho$ by copying the successors of~$u_\sigma$.
Moreover, since we want to preserve the satisfaction of queries from
$\tcqcqs\atcq$ in this process, we require that a whole subtree of depth
$d:=\max\{|\NT(\acqalpha)|\mid\acqalpha\in\tcqcqs\atcq\}$ repeats, which is the
maximum amount of elements that a match for a connected CQ from $\tcqcqs\atcq$
may require.
For tree ABoxes $(\Amc_\Tmc)_{\Tmc\in\tau}$ and
$u_\sigma,u_\rho\in\NI(\Amc_\Tmc)$, we say that $u_\sigma$ is an \emph{ancestor}
of~$u_\rho$ (and $u_\rho$ is a \emph{descendant} of~$u_\sigma$) if
$\rho=\sigma\rho'$ for some~$\rho'$.
Moreover, the \emph{subtree of depth~$d$} below~$u_\sigma$ is the tuple
$(\Cmc_\Tmc)_{\Tmc\in\tau}$ of ABoxes, where each~$\Cmc_\Tmc$ is obtained
from~$\Amc_\Tmc$ by restricting this ABox to the individual names from
$\{u_{\sigma\rho}\mid|\rho|\le d\}$.
We say that two such subtrees
$(\Cmc_\Tmc)_{\Tmc\in\tau},(\Dmc_\Tmc)_{\Tmc\in\tau}$ are \emph{isomorphic} if
there is a bijection~$f$ between $\NI(\Cmc_\Tmc)$ and $\NI(\Dmc_\Tmc)$ such that
each $\Dmc_\Tmc$ is equal to $\Cmc_\Tmc$ after replacing all individual names
according to~$f$.

\begin{definition}[Complete Tree ABoxes]
\label{def:complete-tree-aboxes}
	Tree ABoxes $(\Amc_\Tmc)_{\Tmc\in\tau}$ for~$\tau$ are called \emph{complete}
	if, for all  $u_\rho\in\NI(\Amc_\Tmc)$, exactly one of the following
	conditions is satisfied:
	\begin{enumerate}[resume*=ta]
		\item\label{ta:successors}
			for all $\Tmc=(\Bmc,\Qmc,\Mmc)\in\tau$ and $R\in\NRM(\Omc)$,
			$\exists R(u_\rho)\in\Amc_\Tmc$ iff
			$u_{\rho\smash{R^{[\ |\Tmc]}}}\in\NI(\Amc_\Tmc)$ and
			$R(u_\rho,u_{\rho\smash{R^{[\ |\Tmc]}}})\in\Amc_\Tmc$; or
		\item\label{ta:bounded}
			$u_\rho$ has an ancestor~$u_{\rho'}$ such that $\rho=\rho'\sigma$ with
			$|\sigma|=d$, and $u_{\rho'}$ has an ancestor~$u_{\rho''}$ such that the
			subtree of depth~$d$ below~$u_{\rho'}$ is isomorphic to the subtree of
			depth~$d$ below~$u_{\rho''}$; in this case, $u_\rho$ does not have any
			descendants in~$\NI(\Amc_\Tmc)$.
	\end{enumerate}
\end{definition}

Condition~\ref{ta:successors} ensures that complete tree ABoxes contain all role
successors required by existential restrictions, but only up to the first time
that a subtree of depth~$d$ repeats~\ref{ta:bounded}.
Condition~\ref{ta:bounded} implies that the size of complete tree ABoxes
is constant in data complexity.
To see that this holds, let $t$ be the number of
possible types; clearly
$t\le2^{|\BC(\aont)|+|\tcqcqs\atcq|+2^d\cdot|\tcqcqs\atcq|}$.
Then, the number of possible successors of an element of~$\NI(\Amc_\Tmc)$ is
$s\le t\cdot|\NRM(\aont)|$, and a subtree of depth~$d$ has at most $m\le s^d$
elements.
Since each individual name can satisfy different basic concepts and have
different role connections to each of its $s$ successors 
in every ABox~$\Amc_\Tmc$, $\Tmc\in\tau$, the number of non-isomorphic subtrees
of depth~$d$ is at most $k \le \big(2^{|\BC(\aont)|+s\cdot|\NRM(\aont)|}\big)^{m\cdot t}$.
Finally, this means that any path of length $k+d$ has at least two nodes with
isomorphic subtrees of depth~$d$ (we add~$d$ to ensure that the $k$th node still
has a full subtree of depth~$d$).
Since the quantities $d,t,s,m,k$ only depend on~\aont and~\atcq, 
the size and number of possible complete tree ABoxes is constant.

To summarize, we guess a tuple of additional information of the form
\[ \atuple=\Bigg(
  (\aboxtype{i})_{0\le i\le n+k},\ 
  \changed{\big\{(\typeai{a}{i})_{0\le i\le n+k}\big\}_{a\in\NI(\atkb)\cup\NIA},\ 
	\big\{(\tabox{\tau}{\Tmc})_{\Tmc\in\tau}\big\}_{\tau\in\TTypes}}
\Bigg), \]
where
\begin{itemize}
	\item $\aboxtype{i}$ are ABox types that contain~$\afb_i$
		and~$\Amc_{\Q_{\iota(i)}}$ and agree on the rigid assertions;
	\item $\typeai{a}{i}=(\Bmc^a_i,\Qmc^a_i,\Mmc^a_i)$ are types such that
		$\Amc_{\Bmc^a_i}(a)\subseteq\aboxtype{i}$; we denote the resulting temporal
		type of an individual name $a\in\NI(\atkb)\cup\NIA$ by
		$\ttypea{a} := \{ \typeai{a}{i} \mid i\in[0,n+k] \}$;
	\item $(\tabox{\tau}{\Tmc})_{\Tmc\in\tau}$ are sequences of complete tree
		ABoxes for~$\tau$ if we have $\tau=\ttypea{a}$ for some
		$a\in\NI(\atkb)\cup\NIA$, and $\tabox{\tau}{\Tmc}$ are empty otherwise.
\end{itemize}
Given a complete tree ABox~\Amc and an individual name~$a$, we construct
$\Amc[a]$ by replacing~$u_\eps$ by~$a$ and all other individual names~$u_\rho$
by~$u_{a\rho}$.
Given the tuple~\atuple as above, we define the ABoxes~$\tabox{\atuple}{i}$,
for each time point $i\in[0,n+k]$, as the union of all ABoxes
\ctwo{$\tabox{\tau}{\Tmc}[a]$ with $\tau=\ttypea{a}$ and $\Tmc=\typeai{a}{i}$}, for all individual names
$a\in\NI(\atkb)\cup\NIA$.

We can now characterize r-satisfiability similarly to
Definition~\ref{def:dlltcqs-r-complete}.

\begin{restatable}{lemma}{LemExdlltqcsRSat}
\label{lem:exdlltcqs-r-sat}
	\as is r-satisfiable w.r.t.~$\iota$ and~\atkb iff there is a tuple \atuple as
	above such that, for all $i\in[0,n+k]$:
	\begin{enumerate}[label=\textnormal{(\rccond\arabic*$'$)},leftmargin=*]
		\item\label{exdll-rc-consistent}
			$\KR[i]:=\langle\aont,\aboxtype{i}\cup\tabox{\atuple}{i}\rangle$ is
			consistent.
		\item\label{exdll-rc-negcqs}
			For all $p_j\in\overline{\ax_{\iota(i)}}$, we have
			$\KR[i]\not\models\acqalpha_j$.
	\end{enumerate}
\end{restatable}

As mentioned above, \as and $\iota$ can be guessed in polynomial time in data
complexity, and the LTL satisfiability test can be done in \PTime
\cite[\citelem4.12]{BaBL-JWS15}.
Similarly, the tuple~\atuple can be guessed in polynomial time; in particular, Definition~\ref{def:tree-aboxes} is independent of the input ABoxes.
Since both KB consistency and CQ non-entailment for \DLLitehkrom are decidable
in \NP \cite[\citethm8.2]{dllrelations}, \ref{exdll-rc-consistent}
and~\ref{exdll-rc-negcqs} can be decided nondeterministically in polynomial
time, since the size of each~$\tabox{\atuple}{i}$ only depends linearly on the
number of individual names in the input. By Lemma~\ref{lem:tcq-sat-iff}, TCQ
satisfiability is hence in \NP.

This result actually applies even to $\mathcal{ALCHI}$, which extends
\DLLitehbool by  qualified existential restrictions (value restrictions can be
eliminated); Lemma~\ref{lem:krom-transf} shows how to
simulate such restrictions on the left-hand side of CIs. However, it is
well-known that the presence of role inclusions also allows us to express
qualified existential restrictions on the right-hand side of
CIs~\cite{dllfamily}. In other words, in the context of TCQ entailment, there is
no difference between \DLLitehkrom, \DLLitehbool, and $\mathcal{ALCHI}$.

\begin{theorem}\label{thm:dc-dllitehkrom}
	TCQ entailment in $\mathcal{ALCHI}$ is in \coNP in data complexity, even if\linebreak $\NRR\neq\emptyset$.
\end{theorem}


\section{Conclusions}\label{ch:conclusions}

In this article, we have studied the complexity of TCQ entailment in \DLLite logics between \DLLitecore and \DLLitehhorn.
We have thus focused on a scenario that reflects the needs of the applications of today: the temporal queries are based on LTL, one of the most important temporal logics; the ontologies are written in standard lightweight logics; and the data allows to capture data streams \changed{\cite{KKMNNOS-ISWC16,KMM+-JWS17,AKK+-TIME17,BKR+-JAIR18}}.
Since the complexities we have shown for the Horn fragments of \DLLite are considerably better than known results for other DLs, including the lightweight DL \EL, and do not even depend on the rigid symbols considered, we have identified a fragment that is interesting for applications that need efficient reasoning.

In contrast, TCQ entailment in the Krom and Bool fragments of \DLLite turned out to be as complex as in more expressive DLs, such as \SHQ~\cite{BaBL-JWS15}. 
In particular, we have shown that TCQ entailment in expressive DLs such as $\mathcal{ALCHI}$ can be reduced to TCQ entailment in \DLLitehkrom.
While the combined complexity thus strongly depends on which symbols are considered to be rigid, we have shown the contrary for data complexity. More precisely, we have shown \coNP containment for the case with rigid roles, and thus closed an important gap.
Altogether, our results show that the features we have studied can often be considered ``for free''. For combined complexity, 
TCQ entailment w.r.t.\ a \DLLitehhorn TKB is in \PSpace, even if rigid symbols are considered; this matches the complexity of satisfiability in LTL, which is much less expressive given the fact that ontologies are not considered at all.
Similarly, the \coNP-containment we have shown for many expressive DLs for data complexity matches the complexity of conjunctive query entailment in these DLs. 

In future work, we want to study extensions of TCQs with operators from metric temporal logics and investigate extensions of \DLLitehhorn. On the other hand, a restriction of the set of temporal operators could yield even better results, such as first-order rewritability \cite{BoLT-JWS15}.
It would be also interesting to find out if there are efficient parallel implementations of TCQ answering in \DLLitehhorn, for which we have \ALogTime data complexity.%
\footnote{Containment in \ALogTime is considered as an indicator for the existence of efficient parallel implementations~\cite[\citethm6.27]{compcomplexity09}.}

\begin{acknowledgements}
We thank Carsten Lutz and the anonymous reviewers for their many helpful comments.
This work was supported by the German Research Foundation (DFG) within the Collaborative Research Centre~912 (HAEC), in the joint DFG-ANR project BA 1122/19-1 (GOASQ), and in the grant 389792660 as part of TRR~248 (see \url{https://perspicuous-computing.science}).
\end{acknowledgements}

\begin{appendix}
	\normalsize

\changed{
\section{Additional Background on \DLLite and LTL}
\label{app:prelim}
}

\subsection{Reasoning in Horn Fragments of \DLLite}
\label{sec:prelims-dls-horn}


The logics below \DLLitehhorn do not allow to express disjunction on the
right-hand side of CIs, which means that the CIs can be represented as
first-order Horn clauses.
Reasoning in such DLs is easier since it can be done using deterministic
algorithms, often based on \emph{canonical interpretations}.
We recall (and slightly adapt) the construction of the canonical interpretation
from~\cite{dllfamily,BotAC-DL10:dllhorn}, which is based on the standard
\emph{chase}~\cite{DNR-PODS08:corechase}.
This interpretation contains prototypical domain elements of the
form~$\uel\apath$ ($u$ for \enquote{unnamed}), where $\apath$ is a \emph{path}
$\apath:=aR_1\dots R_\ell$ with an indidivual name~$a$ and roles
$R_1,\dots,R_\ell$.
We assume that the KB does not already contain the symbols~$\uel\apath$.
The expression $|\apath|$ denotes the \emph{length}~$\ell$ of~\apath.

\begin{definition}[Canonical Interpretation]
	\label{def:dll-io}\index{canonical interpretation!dl-litehhorn@\DLLitehhorn}
	Let $\akb=\langle\aont,\afb\rangle$ be a \DLLitehhorn knowledge base.
	First, for all $A\in\NC$ and $P\in\NR$, define the interpretation $\aint^0$ as follows:
	\begin{align*}
   \Delta^{\aint^0}&:=\NI(\Kmc), \\
	A^{\aint^0}&:=\{ a\mid A(\indone)\in\afb \}, \\ 
	P^{\aint^0}&:=\{ (\indone,\indtwo)\mid P(\indone,\indtwo)\in\afb \}\cup{}
	\{  (\indone ,\uel{\indone P})\mid \exists P(\indone)\in\afb\}\cup{}
	\{  (\uel{\indone P^-},\indone)\mid \exists P^-(\indone)\in\afb\}.
	\end{align*}
	Then, for each $i\ge0$, do the following:
	for all $X\in\NC\cup\NR$, define $X^{\aint^{i+1}}:=X^{\aint^i}$,
	apply one of the following rules, and increment~$i$:
	\begin{itemize}
		\item If $\bigsqcap\abcset\sqsubseteq B\in\aont$ and
      $\el\in (\bigsqcap\abcset)^{\aint^i}$, then do the following:
       \begin{itemize}
          \item if $B\in\NC$, then add $\el$ to~$B^{\aint^{i+1}}$;
          \item if $B=\exists R$ and $\el\in\NI(\afb)$, then add
            $\uel{\el R}$ to $\Delta^{\aint^{i+1}}$ and
            $(\el,\uel{\el R})$ to $R^{\aint^{i+1}}$;
          \item if $B=\exists R$ and $\el=\uel{\apath}$, then add
            $\uel{\apath R}$ to $\Delta^{\aint^{i+1}}$ and
            $(\el,\uel{\apath  R})$ to $R^{\aint^{i+1}}$.
      \end{itemize}
      We also apply this rule to the CIs of the form
      $\exists R\sqsubseteq\exists R$, for all $R\in\NR(\aont)$, and
      $\exists S\sqsubseteq\exists R$ and $\exists S^-\sqsubseteq\exists R^-$,
      for all $S\sqsubseteq R\in\aont$, which we consider to be implicitly
      present in~\aont ($**$).
		\item If $S\sqsubseteq R\in\aont$ and
      $(\elone,\eltwo)\in S^{\aint^i}$, then add $(\elone,\eltwo)$ to
      $R^{\aint^{i+1}}$.
	\end{itemize}
	The \emph{canonical interpretation}~$\Imc_\akb$ is obtained as the limit of
	this (possibly infinite) sequence of rule applications, assuming that
   each rule that becomes applicable at some point is executed exactly once.
   We denote by $\adomp[\aint_\akb]\un:=\Delta^{\Imc_\akb}\setminus\NI(\Kmc)$
   the set  of all unnamed elements~$\uel\apath$ that are introduced in this
   process.
\end{definition}

Our assumption~($**$) about additional axioms in the ontology ensures that,
whenever there is an element $a\in\NI(\akb)\cap(\exists R)^{\aint^i}$ for
some $i\ge0$, then $a$ has an $R$-successor~$\uel{aR}$ in the canonical
interpretation, and similarly for the new unnamed elements.
This assumption simplifies some of our proofs (specifically, those of
Lemmas~\ref{lem:dlltcqs-idef} and~\ref{lem:dlltcqs-qi}).
%
This is also the main difference to the constructions
in~\cite{dllfamily,BotAC-DL10:dllhorn}, where $\uel{aR}$ would only be created
if $a$ does not already have an $R$-successor, \eg another named element.
Additionally, we are dealing with negated assertions and assertions about basic
concepts (\ie not just concept names) in ABoxes, which are not \changed{explicitly} considered
in~\cite{dllfamily,BotAC-DL10:dllhorn}.
However, it is straightforward to adapt the following results, which show that
the canonical interpretation~$\aint_\akb$ can be used for reasoning about~\akb.

\begin{lemma}[see \cite{BotAC-DL10:dllhorn}]
	\label{lem:dll-iomodel}\label{lem:iosatq}
	A \DLLitehhorn knowledge base~\akb is consistent iff $\Imc_\akb\models\akb$.
	Moreover, if \akb is consistent, then for every Boolean UCQ~$\aquery$, we
   have $\akb\models \aquery$ iff $\Imc_\akb\models\aquery$.
\end{lemma}

We now describe the behavior of the domain elements of $\aint_\akb$ in more
detail.
In particular, the basic concepts satisfied by an element
$\el\in\adom[\aint_\akb]$ are uniquely determined by the basic concepts
satisfied by~\el at the point where it was first introduced
into~$\adom[\aint^i]$.
\changed{For an unnamed element $u_{\rho R}$, there is only a single such basic concept, namely $\exists R^-$.}
\begin{lemma} \label{lem:dll-io-elements}\label{lem:dll-deltaun}
  Let $\el\in\adom[\aint_\akb]$, $i\ge0$ be the minimal index for which
  $\el\in\adom[\aint^i]$, and $\Bmc_\el$ be the set of all
  $B\in\BC(\aont)$ such that $e\in B^{\aint^i}$. Then, for all
  $B\in\BC(\aont)$, we have $\el\in B^{\aint_\akb}$ iff
  $\aont\models\bigsqcap\Bmc_\el\sqsubseteq B$.
  In particular, if $\el=\uel{\apath R}\in\adomp[\aint_\akb]\un$, then
  $\uel{\apath R}\in B^{\Imc_\akb}$ iff $\aont\models\exists R^-\sqsubseteq B$.
\end{lemma}

\subsection{Propositional Linear Temporal Logic}
\label{sec:prelims-ltl}

Propositional linear temporal logic (LTL) extends propositional logic with
modal operators to represent past and future moments. We implicitly
fix a finite signature $\altlsignature=\{p_1,\dots, p_\ell\}$ of propositional
variables.

\begin{definition}[Syntax of LTL]\label{def:ltl-syntax}
	The set of \emph{LTL formulas} is defined by the following rule, where
  $p\in\altlsignature$:
	$\altlform,\altlformtwo::= p\mid \neg\altlform\mid\altlform\land\altlformtwo
	\mid\Next\altlform\mid \Previous\altlform\mid\altlform\Until\altlformtwo\mid\altlform\Since\altlformtwo$.
\end{definition}
Again, more operators can be derived as in Table~\ref{tab:tcqs-derived-ops}.

%
\begin{definition}[Semantics of LTL]\label{def:ltl-semantics}
	An \emph{LTL structure} is an infinite sequence $\altlint = (w_i)_{i\ge0}$ of \emph{worlds} $w_i\subseteq \altlsignature$.
	Given an LTL formula \altlform and time point $i \ge 0$, the satisfaction relation $\altlint,i\models\altlform$ is defined by induction on the structure of~\altlform: for propositional variable $p\in\altlsignature$, $\altlint,i\models p$ holds if $p\in w_i$; for complex formulas, the corresponding condition of Table~\ref{tab:tcq-semantics} has to be satisfied.
	If $\altlint,0\models\altlform$, then \altlint is a \emph{model} of \altlform.
\end{definition}

%
We briefly review some known constructions and results for LTL satisfiability checking.
For a set \Fmc of LTL formulas, $\Clo(\Fmc)$ denotes the set of all subformulas occurring in~\Fmc, together with all their negations.
The set $\atypeset[\Fmc]$ contains all \emph{types} for~\Fmc, which are sets $\atype\subseteq\Clo(\Fmc)$ such that
\begin{itemize}
  \item for all $\lnot\altlform\in\Clo(\Fmc)$, we have $\lnot\altlform\in\atype$ iff $\altlform\notin\atype$;
  \item for all $\altlform\land\altlformtwo\in\Clo(\Fmc)$, we have $\altlform\land\altlformtwo\in\atype$ iff $\{\altlform,\altlformtwo\}\subseteq \atype$.
\end{itemize}
Each type~\atype uniquely defines a world $w=\atype\cap\altlsignature$.
A type~\atype is \emph{initial} if
\begin{itemize}
  \item for all $\Previous\altlform\in\Clo(\Fmc)$, we have $\Previous\altlform\notin\atype$;
  \item for all $\altlform\Since\altlformtwo\in\atype$, we have $\altlformtwo\in\atype$.
\end{itemize}
A pair $(\atype_1,\atype_2)\in\atypeset[\Fmc]\times\atypeset[\Fmc]$ is \emph{t-compatible} if
\begin{itemize}
  \item for all $\Previous\altlform\in\Clo(\Fmc)$, we have $\Previous\altlform\in\atype_2$ iff $\altlform\in\atype_1$;
  \item for all $\Next\altlform\in\Clo(\Fmc)$, we have $\Next\altlform\in\atype_1$ iff $\altlform\in\atype_2$;
  \item for all $\altlform\Since\altlformtwo\in\Clo(\Fmc)$, we have $\altlform\Since\altlformtwo\in\atype_2$ iff (i) $\altlformtwo\in\atype_2$, or (ii) $\altlform\in\atype_2$ and $\altlform\Since\altlformtwo\in\atype_1$.
  \item for all $\altlform\Until\altlformtwo\in\Clo(\Fmc)$, we have $\altlform\Until\altlformtwo\in\atype_1$ iff (i) $\altlformtwo\in\atype_1$, or (ii) $\altlform\in\atype_1$ and $\altlform\Until\altlformtwo\in\atype_2$.
\end{itemize}
These local properties are the basis for the following characterization, which says that every satisfiable LTL formula has a periodic model with a period of at most exponential length.

\begin{lemma}[{see~\cite{SiCl85:ltlpspace}}]\label{lem:ltl-periodic-model}
  An LTL formula \altlform is satisfiable iff there is a sequence of types
  $T_0,\dots,T_s,\dots,T_{s+p}$ for $\{\altlform\}$ such that
  \begin{itemize}
    \item $s$ and $p$ are bounded by an exponential function in the size of~\altlform;
    \item $T_0$ is initial and contains~\altlform;
    \item for all $i\in[0,s+p-1]$, the pair $(T_i,T_{i+1})$ is t-compatible;
    \item the pair $(T_{s+p},T_s)$ is t-compatible;
    \item for each $\altlform\Until\altlformtwo\in T_s$, there is an index $i\in[s,s+p]$ such that $\altlformtwo\in T_i$.
  \end{itemize}
\end{lemma}

The last condition says that each $\Until$-subformula has to be satisfied at some point within the period.
This result forms the basis of a \PSpace decision procedure for LTL satisfiability~\cite{SiCl85:ltlpspace}, which we describe in more detail in Section~\ref{sec:dlltcqs-cc}.

To conclude this section, we recall the \emph{separation theorem}, which was originally shown in~\cite{Gabbay87} using a \changed{\emph{strict} semantics for $\Until$ and $\Since$}, but also holds in our setting \changed{since the strict and non-strict variants of these operators are mutually expressible in the presence of $\Next$/$\Previous$}.
An LTL formula is called \emph{separated} if no future operators occur in the
scope of past operators and vice versa.
\begin{lemma}[{see~\cite{Gabbay87}}]
	\label{lem:sep}
	Every LTL formula is equivalent to a separated LTL formula.
\end{lemma}


\section{Proofs for Section~\ref{sec:dlltcqs-r-sat-charact}}
\label{app:r-complete}

\begin{lemma}\label{lem:dlltcqs-rSatEasyDir}
	If \as is r-satisfiable w.r.t.~$\iota$ and~\atkb, then there is an r-complete tuple w.r.t.~\as and~$\iota$.
\end{lemma}
\begin{proof}
	
	Let $\Jmc_0,\ldots,\Jmc_{n+k}$ be the interpretations over a domain \adom that exist according to the r-satisfiability of \as (see Definition~\ref{def:tqa-r-sat}).
	We assume w.l.o.g.\ that \adom contains $\NI(\atkb)$ and that all individual names are interpreted as themselves in all of these interpretations.
	
	%
	We first define the tuple $(\AR,\QR,\QRn,\RF)$ as follows:
	\begin{align*}
	\AR :={} 
	&\{\alpha\in\ASr(\Kmc) \mid \Jmc_0\models\alpha\};\\
	\QR :={}
	&\{\acq_j\in\tcqcqs\atcq\mid  p_j\in\ax \text{ for some } \ax\in\as\};\\
	\QRn :={} 
	&\{\acq_j\in\tcqcqs\atcq\mid p_j\not\in\ax \text{ for some }\ax\in\as\};\\
	\RF :={} 
	& \{\exists S(b) \mid{} S\in\NFRM(\aont),\ b\in\NI(\akb)\cup\NIA,\ i\in[0,n+k], \\
	&\phantom{\{}
	\langle\aont,\AR\cup\rigcons{\QR}\cup\Amc_{Q_{\iota(i)}}\cup\Amc_i\rangle
	\models\exists S(b)\}.
	\end{align*}
	
	We prove that the tuple is r-complete by showing that it satisfies all the conditions in Definition~\ref{def:dlltcqs-r-complete}. 
	It is easy to see that \AR is a rigid ABox type for~\aont, that \QR satisfies \cond\ref{def:dlltcqs-rc:cqcons}, and that \QRn complies with \cond\ref{def:dlltcqs-rc:qrn}. 
	Moreover, the rather straightforward definition of the tuple  seems to make it easy to show that each of the knowledge bases~\KR[i] has a model (\cond\ref{def:dlltcqs-rc:consistent}) which satisfies neither a CQ from $\tcqcqs{\atcq}\setminus\Q_{\iota(i)}$ (\cond\ref{def:dlltcqs-rc:negcqs}) nor a witness of the queries in \QRn (\cond\ref{def:dlltcqs-rc:witnesses}), 
	based on the given interpretations; however, special attention needs to be given to the UNA and \cond\ref{def:dlltcqs-rc:rf}. 
	The crucial point is that the given interpretations may satisfy CQs in a conjunction $\chi_{\iota(i)}$ by a homomorphism that can however not be used to satisfy the corresponding ABox $\Amc_{\Q_{\iota(i)}}$, because \ahom maps different variables to the same domain element, which are represented as named individuals from \NIA in the ABox; the same applies to the ABox $\Amc_{\RF}$ w.r.t.\ the elements from \NIA and \NIT.
	For \cond\ref{def:dlltcqs-rc:rf}, observe that the ``if''-direction does not directly yield that \emph{each} of the given interpretations satisfies the assertions in $\Amc_{\RF}$.
	
	The idea is therefore to extend the given interpretations $\Jmc_i$ in two steps. First, we construct models $\Imc'_i$ of $\langle\aont,\AR\cup\rigcons{\QR}\cup\Amc_{Q_{\iota(i)}}\cup\Amc_i\rangle$ that interpret the elements of \NIA and \NIT by using duplicates of elements from \adom. In order to overcome the issue with \cond\ref{def:dlltcqs-rc:rf}, we ensure that these interpretations still share one domain and also interpret the rigid symbols in the same way.
	This allows us to then adapt these interpretations in a second step for the elements from \NIT 
	in order to get models of the respective knowledge bases~\KR[i], which include $\Amc_{\RF}$.
	
	%
	For this extension, we consider different canonical interpretations for all CQs \textcolor{blue}{$\acqalpha_j\in\QR$ satisfied in one of the given interpretations: $\Imc_{\RF}$, the one of $\akb_{\RF}:=\langle\aont,\Amc_{\RF}\rangle$; $\Imc_{j}$, the one of $\akb_{j}:=\langle\aont,\cqinst{\acq_j}\rangle$; and 
	$\Imc^\mathsf{R}_{j}$, which collects the rigid consequences of $\Imc_{j}$} and is inductively defined according to Definition~\ref{def:dll-io} with the adaptation that all symbols $\aname\in\NC\cup\NR$ are initially interpreted as follows:
	\begin{align*}
	\aname^0 &:= \begin{cases}
	\aname^{\Imc_{j}} &\text{if $\aname\in\NRC\cup\NRR$},\\
	\emptyset &\text{otherwise}.
	\end{cases}
	\end{align*}
	Being defined in this way, $\Imc^\mathsf{R}_{j}$ behaves exactly as $\Imc_{j}$ w.r.t.\ the rigid
	names, but the interpretations of the flexible names only contain those tuples that are implied by the rigid information. 
	Note that the domain of $\Imc^\mathsf{R}_{j}$ is a subset of the domain of 
	$\Imc_{j}$ 
	since all elements~$\uel\apath$ that would be created by the iteration in Definition~\ref{def:dll-io} for
	$\Imc^\mathsf{R}_{j}$ are also created by the one for $\Imc_{j}$ and are hence contained in the initial interpretation above. 
	We consider the interpretations $\Imc^\mathsf{R}_{j}$ to make sure that all interpretations we construct also satisfy $\rigcons{\QR}$, even if they do not satisfy an ABox~$\cqinst{\acq}$.
	
	Observe that $\akb_j$ is consistent since $\acq_j\in\QR$ implies that there is an interpretation~$\Jmc_i$ that satisfies~\acq and \aont and is thus a model of~$\akb_j$ (if two variables are mapped by the homomorphism to the same domain element, we obtain a model respecting the UNA by creating a copy of this element that satisfies exactly the same concept names and participates in the same role connections as the original element).
	Moreover, all the given interpretations satisfy the rigid consequences of $\cqinst{\acq_j}$ w.r.t.~\aont (\ie particularly $\rigcons{\{\acq_j\}}$, see Definition~\ref{def:dlltcqs-consequences}) because, by assumption, they share one domain and respect the rigid names.
	The following properties of all $\acq_j\in\QR$ that are crucial for our construction:
	\begin{itemize}
		\item $\Imc_{j}$ can be homomorphically embedded into each~$\Jmc_i$ for which we have $p_j\in \ax_{\iota(i)}$ since there must be a homomorphism of~$\acq_j$ into~$\Jmc_i$. Hence, domain elements that satisfy at least the symbols satisfied by the elements of~$\Imc_{j}$ must exist.
		\item $\Imc^\mathsf{R}_{j}$ can be homomorphically embedded into all~$\Jmc_i$, because there must be an index $\ell\in[0,n+k]$ such that $\Jmc_\ell$ satisfies~$\acq_j$ (and~\aont), and the rigid consequences of~$\acq_j$, represented by $\Imc^\mathsf{R}_{j}$, are satisfied in all the given interpretations.
	\end{itemize}
	These facts imply that we can, as the first step, extend all given interpretations~$\Jmc_i$ to models~$\Imc_i'$ of~$\rigcons{\QR}$ and~$\Amc_{Q_{\iota(i)}}$, because they already contain elements that behave in the same way---at least \wrt the symbols that need to be satisfied to obtain such models:
	\begin{itemize}
		\item The common domain \adom is extended by the union of the domains of all $\Imc_{j}$ with $\acq_j\in\QR$ (including the domains of $\Imc^\mathsf{R}_{j}$). These domains may overlap in~\NI and~\NIA.
		\item The individual names from~\NIA are interpreted as themselves.
		\item For each $j\in[1,m]$ and $p_j\in \ax_{\iota(i)}$, all symbols are, on the domain of $\Imc_{j}$, interpreted exactly as in~$\Imc_{j}$. There are no role connections between the old and the new domains except between~$\NI(\akb)$ and the elements of~\NIA. \hfill ($\dagger$)
		\item If $p_j\in\overline{\ax_{\iota(i)}}$, then all symbols are, on the domain of $\Imc_{j}$, interpreted exactly as in~$\Imc^\mathsf{R}_{j}$.
	\end{itemize}
	Recall that \adom is not yet complete; it is still to be extended by the domain of $\Imc_{\RF}$.
	However, this definition already meets our requirement that $\Imc_i'\models\langle\aont,\AR\cup\rigcons{\QR}\cup\Amc_{Q_{\iota(i)}}\cup\Amc_i\rangle$ and is such that $\Imc_i'\models\chi_{\iota(i)}$ for all $i\in[0,n+k]$, which can be seen given the following observations:
	\begin{itemize}
		\item $\Imc_i'$ satisfies~\aont and~\AR. 
		This is because, given ($\dagger$), 
		the interpretation of symbols on the unnamed elements in the original domain \adom does not change (i.e., the interpretation of basic concepts on these elements remains the same); further, the new domain elements do not exhibit new behavior that was not already present in~$\Jmc_i$; and the latter also implies that the interpretation of basic concepts on the elements of~\NI does not change. 
		\item If $i\le n$, then $\Imc_i'$ satisfies $\Amc_i$ by the same reasons as in the previous item. Otherwise, $\Amc_i$ is trivially satisfied.
		\item $\Imc_i'$ is a model of $\rigcons{\QR}$ since all $\Imc^\mathsf{R}_{j}$, for $\acq_j\in\QR$, are part of $\Imc_i'$ and $\Imc^\mathsf{R}_{j}\models\rigcons{\{\acq_j\}}$.
		
		\item Similarly, $\Imc_i'$ satisfies $\Amc_{Q_{\iota(i)}}$ since that ABox consist exactly of the ABoxes $\Amc_{\acq_j}$
		with $p_j\in \ax_{\iota(i)}$, satisfied by
		$\Imc_{j}$, which are part of $\Imc_i'$.
		\item For each $p_j\in\overline{\ax_{\iota(i)}}$, we get $\Imc_i'\not\models\acq_j$ since any homomorphism of $\acq_j$ into
		$\Imc_i'$ would allow us to also find one into $\Jmc_i$, which contradicts the assumptions that $\Jmc_i\models\chi_{\iota(i)}$.
		Hence, $\Imc_i'\models\chi_{\iota(i)}$.
	\end{itemize}
	
	We come to the second part.
	Since the interpretations~$\Imc_i'$ are models of the knowledge bases $\langle\aont,\AR\cup\rigcons{\QR}\cup\Amc_{Q_{\iota(n+i)}}\cup\Amc_{n+i}\rangle$, every assertion $\exists S(b)\in\RF$ is satisfied in one of them by the definition of \RF. That interpretation thus also satisfies the rigid consequences described in~$\Amc_{\RF}$ since it is a model of \aont.
	But then this holds for all the interpretations because they interpret the rigid symbols on the named elements in the common domain (i.e., those in $\adom\cap\NI(\atkb)\cap\NIA$
	) in the same way and satisfy \AR and $\Amc_{\QR}$, which contain all the relevant rigid information. 
	We can thus extend the domain~\adom by the domain of~$\Imc_{\RF}$ and all interpretations~$\Imc_i'$ as follows. The names from~\NIT are interpreted by themselves and behave exactly in the same way as the corresponding elements that already exist in each interpretation since they describe the consequences of assertions in $\RF$;
	note that we may again have to copy elements if the UNA would be violated otherwise. Note that this extension does not introduce new role connections between the old and the new domains except between~$\NI(\akb)$ and the elements of~\NIA and \NIT, similar to~($\dagger$). We can therefore argue similarly as above that the final interpretations~$\Imc_i'$ are models as required for Condition~\ref{def:dlltcqs-rc:consistent}.
	
	We now use these constructed interpretations to show that $(\AR,\QR,\QRn,\RF)$ also satisfies Conditions~\ref{def:dlltcqs-rc:negcqs} and \ref{def:dlltcqs-rc:witnesses}, \ie it is r-complete.
	%
%
	For \cond\ref{def:dlltcqs-rc:negcqs}, we assume that there are an $i\in[0,n+k]$ and a $p_j\in\overline{\ax_{\iota(i)}}$ such that
	$\KR[i]\models\acq_j$, which yields $\Imc_i'\models\acq_j$. This directly contradicts the fact that $\Imc_i'\models\chi_{\iota(i)}$. 
	The proof for \cond\ref{def:dlltcqs-rc:witnesses} is also by contradiction. We assume
	that there are an index $i\in[0,n+k]$ and a rigid witness query~\acqtwo for some $\acq_j\in\QRn$ such that $\KR[i]\models\acqtwo$, and thus also $\Imc_i'\models\acqtwo$ holds.
	By the definition of~\QRn, there must be a $\ax_\ell\in\as$ with $\ell\in[1,k]$ such that $p_j\not\in\ax_\ell$, and thus $\Imc_{n+\ell}'\not\models\acq_j$. By
	Definition~\ref{def:dlltqcs-rigid-witness-query}, we know that $\Imc_{n+\ell}'\not\models\acqtwo$.
	But this contradicts the facts that \acqtwo contains only rigid names and that
	$\Imc_{n+\ell}'$ and $\Imc_i'$ respect the rigid names.
	%
\end{proof}

We now prove the \enquote{if}-direction of Lemma~\ref{lem:dlltcqs-iff-s-r-consistent}.
Let $\Imc_i$ and $\Jmc_i$, for $i\in[0,n+k]$, be defined as in the main text.
In the following for elements $\el\in\NI(\atkb)\cup\NIA\cup\NIT$, we use $\el^i$
to denote the corresponding element in
$\NI(\atkb)\cup\adomp[\aint_i]\aux\cup\adomp[\aint_i]\tree$; we thus consider
$a^i:=a$ for all $a\in\NI(\atkb)$.

We state the following fact for future reference, and 
then establish connections between $\Jmc_0,\dots,\Jmc_{n+k}$ and $\Imc_0,\dots,\Imc_{n+k}$.
\begin{fact}\label{fact:dlltcqs-deltas}
	The sets $\NI(\atkb)$, $\adomp[\aint_i]\aux$, $\adomp[\aint_i]\tree$, and $\adomp[\aint_i]\un$, for all $i\in[0,n+k]$, are pairwise disjoint.
\end{fact}
\begin{lemma}\label{lem:dlltcqs-idef-role}
	For all $i\in[0,n+k]$, $\elone,\eltwo\in\adom[\Imc_i]$, and $R\in\NR$, $(\elone,\eltwo)\in R^{\Jmc_i}$ iff $(\elone,\eltwo)\in R^{\Imc_i}$.
\end{lemma}
\begin{proof}
($\Leftarrow$) This direction follows directly from the definition of~$R^{\Jmc_i}$.
($\Rightarrow$) 
We focus on the definition of~$R^{\Jmc_i}$.
If $R$ is flexible, we need to consider two cases (i) and (ii); in both, we assume $i\neq j$: (i) $(\elone,\eltwo)\in S^{\Imc_j}$ for some
rigid subrole~$S$ of~$R$, and (ii) $(\elone,\eltwo)\in R^{\Imc_j}$ 
 and either
$\elone$ or $\eltwo$ has a witness w.r.t.~\aont.
Given $\elone,\eltwo\in\adom[\Imc_i]$, Fact~\ref{fact:dlltcqs-deltas} however implies $\elone,\eltwo\in\NI(\atkb)$ in both cases.
Hence, (ii) is impossible since named domain elements cannot have witnesses according to Definition~\ref{def:dlltcqs-witness}.
In case~(i), we get $S(\elone,\eltwo)\in\AR$ since $S\in\NRR$, that \AR is a rigid ABox type, and $\Imc_j\models\AR$. Then, $\Imc_i\models\AR$ implies $(\elone,\eltwo)\in S^{\Imc_i}$, and $\Imc_i\models\aont$ yields $(\elone,\eltwo)\in R^{\Imc_i}$ since we assume $\aont\models S\sqsubseteq R$.

If $R$ is rigid, we consider the case that $(\elone,\eltwo)\in R^{\Imc_j}$ for some $j\neq i$ and get $\elone,\eltwo\in\NI(\atkb)$, as above. Since $R\in\NRR$,  \AR is a rigid ABox type, and $\Imc_j\models\AR$, we must have $R(\elone,\eltwo)\in\AR$. $\Imc_i\models\AR$ then leads to $(\elone,\eltwo)\in R^{\Imc_i}$.
\end{proof}

The following is a direct consequence of the fact that the interpretation of roles in $\Jmc_i$ is based on the canonical interpretations, which are models of \aont, and Fact~\ref{fact:dlltcqs-deltas}.
\begin{lemma}
	\label{lem:dlltcqs-idefj-role}
	For all $i,j\in[0,n+k]$, $\elone\in\adom$, $\eltwo\in\adomp[\Imc_j]\un$, and $R\in\NRM$, we have that $(\elone,\eltwo)\in R^{\Jmc_i}$ implies $(\elone,\eltwo)\in R^{\Imc_j}$.
\end{lemma}

\changed{\begin{lemma}\label{lem:dll-flexible-basic-concept}
	For all $e\in\NI(\Kmc)\cup\NIA\cup\NIT$ and flexible basic concepts~$B$ with $e\in B^{\Jmc_i}$, either $e\in B^{\Imc_i}$ or there is a $\abcset\subseteq\BCr(\Omc)$ and a $j\in[0,n+k]$ with $e\in(\bigsqcap\abcset)^{\Imc_j}$ and $\Omc\models\bigsqcap\abcset\sqsubseteq B$.
\end{lemma}
\begin{proof}
	For flexible concept names, the claim follows directly from the definition of~$\Jmc_i$. It remains to consider~$B$ to be of the form $\exists R$ with $R\in\NRM(\Omc)$. If $e\notin B^{\Imc_i}$, then by the definition of~$\Jmc_i$, there exists a $j\in[0,n+k]$ for which one of the following cases applies:
	\begin{itemize}
		\item There is an $S\in\NRR$ such that $\Omc\models S\sqsubseteq R$ and $e\in(\exists S)^{\Imc_j}$.
		In this case, we can set $\abcset:=\{\exists S\}$ since $\Omc\models\exists S\sqsubseteq\exists R$.
		\item There exists $d\in\Delta^{\Imc_j}$ with $(e,d)\in R^{\Imc_j}$ and either $d$ or $e$ has a witness w.r.t.~\Omc.
		Since $e$ is a named domain element, Definition~\ref{def:dlltcqs-witness} yields $\wits\aont\el=\emptyset$, and thus $\wits\aont\eladd\neq\emptyset$, which implies that \eladd is of the form $\ujel{\el S}$ with $S\in\NRM(\aont)$ and there exists $\abcset\subseteq\BCr(\aont)$ with $\el\in(\bigsqcap \abcset)^{\Imc_j}$ and $\aont\models \bigsqcap \abcset\sqsubseteq \exists S$.
		Since $(\el,\ujel{\el S})\in R^{\Imc_j}$, by Definition~\ref{def:dll-io} we infer that $\aont\models S\sqsubseteq R$, which implies
		$\aont\models\bigsqcap\abcset\sqsubseteq\exists R$.
		\qedhere
	\end{itemize}
\end{proof}}

\lemdlltcqsIdef*
\begin{proof}
	\changed{For \ref{lem:dlltcqs-idef:ind}, $\el\in B^{\Imc_i}$ clearly implies $\el\in B^{\Jmc_i}$ since $\aname^{\Imc_i}\subseteq \aname^{\Jmc_i}$ for $\aname\in\NC\cup\NR$.
	To prove the converse, we first consider the case that $B$ is rigid.}
	By the definition of the rigid names in $\Jmc_i$, $\el\in B^{\Jmc_i}$ yields that there is a $j\in[0,n+k]$ such that $\el\in B^{\Imc_j}$. Since $\Imc_j$ and $\Imc_i$ are both models of the rigid ABox type~\AR, we get $B(\el)\in\AR$ and $\el\in B^{\Imc_i}$.
	
	\changed{For a flexible~$B$ with $e\in B^{\Jmc_i}$, by Lemma~\ref{lem:dll-flexible-basic-concept}, we either immediately get $e\in B^{\Imc_i}$, or there is a $\abcset\subseteq\BCr(\aont)$ with $\aont\models\bigsqcap\abcset\sqsubseteq B$ and a $j\in[0,n+k]$ such that $\el\in(\bigsqcap\abcset)^{\Imc_j}$.
	But this yields $B'(\el)\in\AR$ for all $B'\in\abcset$, which together with $\Imc_i\models\AR$ implies $\el\in(\bigsqcap\abcset)^{\Imc_i}$. Hence, $\Imc_i\models\aont$ leads to $\el\in B^{\Imc_i}$. This concludes the proof of  \ref{lem:dlltcqs-idef:ind}.}
	
	\ref{lem:dlltcqs-idef:auxr} is a direct consequence of
	Fact~\ref{fact:dlltcqs-deltas} and the definition of $\Jmc_i$.
	
	For \ref{lem:dlltcqs-idef:auxf} and the case that $B\in\NC$, the equivalence with one of \ref{lem:dlltcqs-idef:auxf1}--\ref{lem:dlltcqs-idef:auxf3} is covered by the definition of~$\Jmc_i$ if, for~\ref{lem:dlltcqs-idef:auxf3}, Lemma~\ref{lem:dll-deltaun} is taken into account. 
	
	It remains to consider $B$ to be of the form $\exists R$ with $R\in\NRM$.
	For the case that $i=j$, the claim can be restricted to Item~\ref{lem:dlltcqs-idef:auxf1} since the other two are subsumed by it; for~\ref{lem:dlltcqs-idef:auxf2}, this holds
	because $\Imc_i\models\aont$. Then, it is a direct consequence of
	Fact~\ref{fact:dlltcqs-deltas} and the definition of~$\Jmc_i$, because the 
	interpretation of the elements from $\adomp[\Imc_i]\aux\cup\adomp[\Imc_i]\tree\cup\adomp[\Imc_i]\un$ in $\Jmc_i$ is
	not influenced by any $\Imc_j$ with $j\neq i$.
	We consider the case $i\neq j$.
	\begin{itemize}
	\item For $\el\in\NIA\cup\NIT$, we only have to consider~\ref{lem:dlltcqs-idef:auxf2}.

	\changed{($\Rightarrow$)~This follows directly from Lemma~\ref{lem:dll-flexible-basic-concept}.}

	($\Leftarrow$)~Given $\el\in(\exists R)^{\Imc_j}$, there is an element of the form $\ujel{\el R}\in\adomp[\Imc_j]\un$ with $(\el,\ujel{\el R})\in R^{\Imc_j}$ by Definition~\ref{def:dll-io}. Since Definition~\ref{def:dlltcqs-witness} implies that $\abcset$ is a witness of $\ujel{\el R}$, the definition of~$\Jmc_i$ yields $(\el,\ujel{\el R})\in R^{\Jmc_i}$; that is, $\el\in(\exists R)^{\Jmc_i}$.
	
	\item Let $\el\in\adomp[\Imc_j]\un$.

	($\Rightarrow$)~Again, there are two options, by the definition of~$\Jmc_i$:
	\ctwo{%
	\begin{align}
		&\text{there is an $S\in\NRRM$ such that $\el\in(\exists S)^{\Imc_j}$ and $\aont\models S\sqsubseteq R$, or} \label{s:i} \\
		&\text{there is an $R$-successor~\eladd of \el in~$\Imc_j$, and either \el or \eladd has a witness w.r.t.~\aont.} \label{s:ii}
	\end{align}%
	}%
	For~\eqref{s:i}, we can set $\abcset:=\{\exists S\}$, as in Lemma~\ref{lem:dll-flexible-basic-concept}.
	For~\eqref{s:ii}, observe that we have $\el\in(\exists R)^{\Imc_j}$.
	%
	%
	If $\wits\aont\eladd$ is undefined or empty, then \ref{lem:dlltcqs-idef:auxf3} holds directly.
	Otherwise, we assume $\wits\aont\eladd\neq\emptyset$ and hence get $\eladd\in\adomp[\Imc_j]\un$ by Definition~\ref{def:dlltcqs-witness}.
	By Definition~\ref{def:dll-io}, we then have either
	\ctwo{%
	\begin{align}
		&\text{$\el=\ujel{\apath}$ and $\eladd=\ujel{\apath R}$, or} \label{s:ip} \\
		&\text{$\eladd=\ujel{\apath}$ and $\el=\ujel{\apath R^-}$.} \label{s:iip}
	\end{align}%
	}%
	For~\eqref{s:ip}, Definition~\ref{def:dlltcqs-witness} yields $\wits\aont\el\neq\emptyset$ (\ie \ref{lem:dlltcqs-idef:auxf3} holds) or that there is a $\abcset\subseteq\BCr(\aont)$
	such that $\el\in(\bigsqcap\abcset)^{\Imc_j}$ and
	$\aont\models \bigsqcap\abcset\sqsubseteq\exists R$ (\ie \ref{lem:dlltcqs-idef:auxf2} holds).
	For~\eqref{s:iip}, we immediately get that the witness of~\eladd is also a witness of~\el, again by Definition~\ref{def:dlltcqs-witness}.

	($\Leftarrow$)~Let $\el=\ujel{\apath}$. We start with~\ref{lem:dlltcqs-idef:auxf2}. If there is a set $\abcset\in\BCr(\aont)$ with $\el\in(\bigsqcap\abcset)^{\Imc_j}$ and $\aont\models \bigsqcap\abcset\sqsubseteq \exists R$, then
	Definition~\ref{def:dll-io} implies that the element
	$\ujel{\apath R}\in\adomp[\Imc_j]\un$ exists and that
	$(\el,\ujel{\apath R})\in R^{\Imc_j}$. Since $\abcset$ is further a witness of~$\ujel{\apath R}$ by Definition~\ref{def:dlltcqs-witness}, we get $(\el,\ujel{\apath R})\in R^{\Jmc_i}$, and hence $\el\in(\exists R)^{\Jmc_i}$ by the definition of~$\Jmc_i$.
	For~\ref{lem:dlltcqs-idef:auxf3},
	and thus $\el\in(\exists R)^{\Imc_j}$, we similarly get that $(\el,\ujel{\apath R})\in R^{\Imc_j}$ by Definition~\ref{def:dll-io}. Then, $\wits\aont\el\neq\emptyset$ yields
	$\el\in(\exists R)^{\Jmc_i}$, as in the previous case. \qedhere
	\end{itemize}
\end{proof}

\begin{lemma}
	\label{lem:dlltcqs-t-ai}
	For all $i\in[0,n+k]$, $\Jmc_i$ is a model of~$(\aont,\afb_i)$.
\end{lemma}
\begin{proof}
For every assertion $\aaxiom\in\Amc_i$, by Lemma~\ref{lem:dll-iomodel}, we have $\Imc_i\models\aaxiom$, and thus Lemmas~\ref{lem:dlltcqs-idef-role} and~\ref{lem:dlltcqs-idef}\ref{lem:dlltcqs-idef:ind} yield
$\Jmc_i\models\aaxiom$.

We consider a CI $B_1\sqcap\dots\sqcap B_m\sqsubseteq B\in\aont$ and element \el such that $\el\in B_1^{\Jmc_i}\cap\dots\cap B_m^{\Jmc_i}$; note that $B_1,\dots,B_m$ are basic concepts and $B$ is either a basic concept or $\bot$. 
For the case that $\el\in\NI(\atkb)$, Lemma~\ref{lem:dlltcqs-idef}\ref{lem:dlltcqs-idef:ind} yields $\el\in B_1^{\Imc_i}\cap\dots\cap B_m^{\Imc_i}$. Since $\Imc_i\models\aont$, this implies that $\el\in B^{\Imc_i}$, which is impossible if $B=\bot$. Otherwise, we get $\el\in B^{\Jmc_i}$, again by Lemma~\ref{lem:dlltcqs-idef}\ref{lem:dlltcqs-idef:ind}.

Let now  $\el\in\adomp[\Imc_j]\aux\cup \adomp[\Imc_j]\tree\cup \adomp[\Imc_j]\un$ for some $j\in[0,n+k]$.
If $i=j$, then we get the same conclusion as in the previous case since, given that $\Imc_j\models\aont$, Items~\ref{lem:dlltcqs-idef:auxf2} and~\ref{lem:dlltcqs-idef:auxf3} collapse to~\ref{lem:dlltcqs-idef:auxf1}. More precisely, we can argue analogously by referring to Lemma~\ref{lem:dlltcqs-idef}\ref{lem:dlltcqs-idef:auxr} and~\ref{lem:dlltcqs-idef:auxf} instead of Lemma~\ref{lem:dlltcqs-idef}\ref{lem:dlltcqs-idef:ind}.
In case that $i\neq j$, Lemma~\ref{lem:dlltcqs-idef} implies that, for each $B_\ell$ with $\ell\in[1,m]$, we have either that 
(i)~there is a $\abcset_\ell\subseteq\BCr(\aont)$ such that
$\el\in(\bigsqcap\abcset_\ell)^{\Imc_j}$ and
$\aont\models\bigsqcap\abcset_\ell\sqsubseteq B_\ell$ or that
(ii)~$\el\in B_{\smash\ell}^{\smash{\Imc_j}}\cap\adomp[\Imc_j]\un$ and
$\wits\aont\el\neq\emptyset$.
$\Imc_j\models\aont$, following from Lemma~\ref{lem:dll-iomodel}, thus leads to $\el\in B_\ell^{\smash{\Imc_j}}$ for~(i) and to $\el\in B^{\Imc_j}$ for both cases.
If $B=\bot$, this is again impossible.
If $B$ is rigid, then Lemma~\ref{lem:dlltcqs-idef}\ref{lem:dlltcqs-idef:auxr} yields
$\el\in B^{\Jmc_i}$, as required.
If $B$ is flexible and Case~(ii) applies to at least one~$B_\ell$ with $\ell\in[1,m]$, then
Item~\ref{lem:dlltcqs-idef:auxf3} of Lemma~\ref{lem:dlltcqs-idef}\ref{lem:dlltcqs-idef:auxf} yields the claim.
Otherwise, it is easy to see that we can define $\abcset:=\bigcup_{\ell=1}^m\abcset_\ell$ and have $\el\in(\bigsqcap\abcset)^{\Imc_j}$, $\aont\models\abcset\sqsubseteq B_1\sqcap\dots\sqcap B_m$, and $\aont\models\abcset\sqsubseteq B$.
Hence, the Item~\ref{lem:dlltcqs-idef:auxf2} of Lemma~\ref{lem:dlltcqs-idef}\ref{lem:dlltcqs-idef:auxf} applies, and we also get $\el\in B^{\Jmc_i}$.

It remains to consider role inclusions of the form $S\sqsubseteq R$. We consider a tuple $(\elone,\eltwo)\in S^{\Jmc_i}$, \ie there is a $j\in[0,n+k]$ such that $(\elone,\eltwo)\in S^{\Imc_j}$. Since $\Imc_j\models\aont$, we get $(\elone,\eltwo)\in R^{\Imc_j}$. For the case that $R$ is rigid, we immediately get $(\elone,\eltwo)\in R^{\Jmc_i}$. For the case that $R$ is flexible, we assume $(\elone,\eltwo)\notin R^{\Jmc_i}$. Then, we must have that $i\neq j$, that $R$ has no rigid subrole~$S'$ with $(\elone,\eltwo)\in S'^{\Imc_j}$, and that neither \elone nor \eltwo have a witness w.r.t.~\aont. By the second observation, $S$ cannot be rigid nor have a rigid subrole~$S'$ with $(\elone,\eltwo)\in (S')^{\Imc_j}$, because then $S'$ would be a rigid subrole of~$R$. But this means $(\elone,\eltwo)\notin S^{\Imc_j}$, which contradicts our assumption.
\end{proof}

%
%

\begin{lemma}
	\label{lem:dlltcqs-qi}
	For all $i\in[0,n+k]$, $\Jmc_i$ is a model of~\ctwo{$\chi_{\iota(i)}$}.
\end{lemma}
\begin{proof}
We show that $\Jmc_i$ is a model of every CQ literal in \ctwo{$\chi_{\iota(i)}$}.
Let \acqalpha first be a positive such literal.
Since $\afb_{\Q_{i}}$ contains an instantiation of~\acqalpha and $\Imc_i\models\afb_{\Q_{i}}$ by Lemma~\ref{lem:dll-iomodel}%
\footnote{In the remaining parts of the proof, we do not always explicitly refer to Lemma~\ref{lem:dll-iomodel} to justify the argument that $\Imc_i\models\KR[i]$ for all $i\in[0,n+k]$.} 
we know that there is a homomorphism~\ahom of~\acqalpha into~$\Imc_i$ that maps all variables in \acqalpha to elements of~$\adomp[\Imc_i]\aux$; that is, \ahom maps each such variable $x$ to $\ael{x}$.
By the fact that $\adomp[\Imc_i]\aux\subseteq\adom$ and Lemmas~\ref{lem:dlltcqs-idef-role} and~\ref{lem:dlltcqs-idef}, \ahom is then also a homomorphism of~\acqalpha into~$\Jmc_i$.	

Let now $\lnot\acqalpha$ be a negative literal in~$\chi_{i}$.
We proceed by contradiction and assume \ahom to be a homomorphism of~\acqalpha into~$\Jmc_i$.
First, observe the following. 
\begin{enumerate}[label=(O\arabic*)]
\item\label{obs:1} \cond\ref{def:dlltcqs-rc:qrn} implies $\acqalpha\in\QRn$, and hence \cond\ref{def:dlltcqs-rc:witnesses} and Lemma~\ref{lem:iosatq} yield that none of the rigid witness queries of \acqalpha is satisfied in any of the canonical interpretations.
\item\label{obs:2} \cond\ref{def:dlltcqs-rc:negcqs} implies $\KR[i]\not\models\acqalpha$, and thus Lemma~\ref{lem:iosatq} yields that $\Imc_i\models\lnot\acqalpha$.
\end{enumerate}
%
%
%
We derive contradictions to these observations by distinguishing the two cases~(I) and~(II) outlined in Section~\ref{sec:dlltcqs-r-sat-charact}.

\paragraph{(I)} 
Let first \ahom be such that it maps no terms into $\ctwo{\adomp\nam:=}\NI(\atkb)\cup\bigcup_{j=0}^{n+k} \adomp[\Imc_j]\aux\cup\adomp[\Imc_j]\tree$. Because of the UNA, we thus have $\NI(\acqalpha)\cap\NI(\atkb)=\emptyset$, which yields that $\NI(\acqalpha)=\emptyset$ since $\acqalpha$ does not contain names from \NIA or \NIT. 
%
Moreover, we can then assume that there is a single index~$j\in[0,n+k]$ such that \ahom maps all terms of~\acqalpha to elements of~$\adomp[\Imc_j]\un$, by the definition of~\adom.
To see this, note that \acqalpha is connected and that, for a role $R$, a tuple $(\elone,\eltwo)\in R^{\Jmc_i}$ without named individuals exists only if the elements belong to the same domain~$\adom[\Imc_j]$ by the definition of $R^{\Jmc_i}$ and Fact~\ref{fact:dlltcqs-deltas}. 

Given \ref{obs:2}, we then directly get a contradiction for the case that $j=i$ by Lemmas~\ref{lem:dlltcqs-idef-role} and~\ref{lem:dlltcqs-idef}, which imply that \ahom is a homomorphism of~\acqalpha into~$\Imc_i$.

For the case $j\neq i$, we show that there is a rigid witness query~\acqtwo for~\acqalpha such that $\Imc_j\models\acqtwo$, in contradiction \ctwo{to~\ref{obs:1}}.
Since \acqalpha is connected, by Lemma~\ref{lem:dlltcqs-idefj-role}, and by considering how the elements in $\adomp[\Imc_j]\un$ are related by roles within~$\Imc_j$ (see Definition~\ref{def:dll-io}), it is easy to see that there is a variable $x\in\NV(\acqalpha)$, for which $\ahom(x)=\ujel{\apath}$ is such that the length of \apath is minimal compared to the paths of all elements of $\range(\ahom)$; further, all other $\ahom(\vtwo)$ for $\vtwo\in\NV(\acqalpha)$ must then be of the form $\ujel{\apath\apathtwo}$ with $\apathtwo\in\NRM^*$.
Moreover, we know by the definition of~$\Jmc_i$ on the elements
of~$\adomp[\Imc_j]\un$ that $\pi$ is also a homomorphism of~$\varphi$
into~$\Imc_j$.

We now construct the rigid witness query~\acqtwo for~\acqalpha, distinguishing
two cases.
   We first consider the case that
    $\wits\aont{\ujel{\apath }}\neq\emptyset$.
    By Definition~\ref{def:dlltcqs-witness}, \apath must be of the form
    $\apath_0R_0\dots R_\ell$, and there are a set
    $\abcset\subseteq\BCr(\aont)$ such that
    $\aont\models\bigsqcap\abcset\sqsubseteq \exists R_0$ and an element
    $\ujel{\apath_0}\in\adomp[\Imc_j]\un$ that satisfies $\bigsqcap\abcset$
    in~$\Imc_j$. By Definition~\ref{def:dll-io} and
    Lemma~\ref{lem:dll-deltaun}, $\ujel{\apath}\in\adomp[\Imc_j]\un$ implies
    $\aont\models\exists R^-_i\sqsubseteq \exists R_{i+1}$,
    for all $i\in[0,\ell-1]$.
    Consider now the ABox $\Amc_\abcset:=\{B(a)\mid B\in\abcset\}$, where $a$ is an
    arbitrary individual name. By Definition~\ref{def:dll-io} and the above
    observations \ctwo{on~$\pi$}, the canonical model~$\Imc'$ of $\langle\Omc,\Amc_\abcset\rangle$
    is isomorphic to a subtree of~$\Imc_j$ starting at~$\ujel{\rho_0}$.
    In particular, it contains the elements
    $u_{aR_0},\dots,u_{aR_0\dots R_\ell}$ and similar elements corresponding to
    all other successors of $\ujel{\rho_0R_0}$ in~$\Imc_j$. Hence, $\Imc'$
    also satisfies~$\varphi$, namely via the homomorphism~$\pi'$ given by
    $\pi'(x):=u_{a\sigma}$ whenever $\pi(x)=\ujel{\rho_0\sigma}$, for
    each variable $x\in\NV(\varphi)$.
    By Lemma~\ref{lem:iosatq}, 
    $\langle\Omc,\Amc_\abcset\rangle\models\varphi$, which shows that the CQ
    $\exists x.\abcset(x)$ is a rigid witness query for~$\varphi$.
  
	In the remaining case that $\wits\aont{\ujel{\apath}}=\emptyset$, then we cannot find a rigid witness
    query so easily. Instead, we iteratively replace atoms of~$\varphi$ by
    rigid atoms until no more flexible atoms remain. We do this by induction on
    the length~$|\sigma|$ of the elements $\pi(x)=\ujel{\rho\sigma}$,
    $x\in\NV(\varphi)$. That is, we start with those variables~$x$ that are
    mapped to the \enquote{root} $\ujel{\rho}$, and then proceed along the
    tree-like structure of~$\adomp[\Imc_j]\un$. Initially, we set
    $\psi:=\varphi$, and maintain the invariants that $\pi$ is a homomorphism
    of~$\psi$ into~$\Imc_j$, and that every homomorphism of~$\psi$ into an
		interpretation~\Imc can be extended to a homomorphism of~$\varphi$
		into~\Imc (in the end, this implies that $\psi$ is a rigid witness query
		for~$\varphi$ that is satisfied by~$\Imc_j$).

		In the first step, we consider all concept atoms $A(x)$ in~$\psi$ such that
		$\pi(x)=\ujel{\rho}$. Since $\ujel{\rho}\in A^{\Jmc_i}$ and
		$\wits\aont{\ujel{\rho}}=\emptyset$, by Lemma~\ref{lem:dlltcqs-idef} either (i)~$A$ is rigid, or (ii)~there is a
		$\abcset\subseteq\BCr(\Omc)$ such that $\ujel{\rho}\in(\bigsqcap\abcset)^{\Imc_j}$
		and $\Omc\models\bigsqcap\abcset\sqsubseteq A$.
		In case~(i), we do not have to replace~$A(x)$ since it is already rigid, and
		in case~(ii) we replace it by~$\abcset(x)$, which satisfies our invariants.

		Next, we consider the role atoms $R(x,y)$ in~$\psi$ such that either~$x$
		or~$y$ is mapped by~$\pi$ to~$\ujel{\rho}$. Without loss of generality, we
		assume that $\pi(x)=\ujel{\rho}$ and $\pi(y)=\ujel{\rho S}$; if this is not
		the case, then we simply consider~$R^-$ instead of~$R$. If $\wits\aont{\ujel{\rho S}}\neq\emptyset$, then, since $\ujel{\rho}$ has
		no witnesses, there must exist some $\abcset\subseteq\BCr(\Omc)$ such that
		$\ujel{\rho}\in(\bigsqcap\abcset)^{\Imc_j}$ and
		$\Omc\models\bigsqcap\abcset\sqsubseteq\exists S$.
		But then the canonical model for $\langle\Omc,\Amc_\abcset\rangle$ already satisfies all
		atoms of~$\psi$ that are mapped by~$\pi$ to the
		elements below~$\ujel{\rho S}$, namely via the
		homomorphism~$\pi'$ defined by $\pi'(y):=u_{a\sigma}$ whenever
		$\pi(y)=\ujel{\rho\sigma}$. Moreover, $\pi'$ also satisfies $R(x,y)$ since,
		by Definition~\ref{def:dll-io}, $\Omc\models S\sqsubseteq R$. This
		means that we can replace all these atoms by the conjunction~$\abcset(x)$ (for
		all~$x$ with $\pi(x)=\ujel{\rho}$).
		It remains to consider the case that neither $\ujel{\rho}$ nor
		$\ujel{\rho S}$ have any witnesses.
		Since $(\ujel{\rho},\ujel{\rho S})\in R^{\Jmc_i}$, by the definition
		of~$\Jmc_i$, there is a rigid role~$S_2$ such that
		$\Omc\models S_2\sqsubseteq R$ and
		$(\ujel{\rho},\ujel{\rho S})\in\smash{S_2^{\Imc_j}}$. Hence, we can replace
		$R(x,y)$ by the rigid atom $S_2(x,y)$, which is satisfied in~$\Imc_j$ and
		implies $R(x,y)$.

		We have thus dealt with all atoms involving a variable~$x$ with
		$\pi(x)=\ujel{\rho}$.
		Assume now that we have continued this process up to
		some~$\ujel{\rho\sigma}$, and consider all atoms involving variables~$x$
		mapped to one of its direct successors~$\ujel{\rho\sigma S}$. By our
		construction above, we can assume that $\ujel{\rho\sigma S}$ has no
		witnesses; otherwise, we would already have replaced all atoms
		involving~$x$.
		Hence, we are in the same position as in the case above, and can continue
		replacing the atoms in~$\psi$ by rigid atoms in exactly the same way.
		We can do this until $\psi$ contains no more flexible atoms, and hence is a
		rigid witness query for~$\varphi$ \ctwo{that is satisfied in~$\Imc_j$}.

\ctwo{This contradicts~\ref{obs:1}, and hence} finishes the proof of Case~(I) of Lemma~\ref{lem:dlltcqs-qi}.

\paragraph{(II)}
In the remainder of the proof, let \ahom be such that at least one term is mapped into $\adomp\nam=\NI(\atkb)\cup\bigcup_{j=0}^{n+k}
\big(\adomp[\Imc_j]\aux\cup\adomp[\Imc_j]\tree\big)$.
We directly define a homomorphism~\ahomtwo of \acqalpha into $\Imc_i$ to contradict \ref{obs:2}, \ie $\Imc_i\models\lnot\acqalpha$.
This is done in three phases, by considering the terms \ahom maps to elements from $\adomp\nam$, those that are directly connected to the latter, and all others. After phase one, all terms considered are thus mapped to elements from $\bigcup_{j=0}^{n+k} \adomp[\Imc_j]\un$ by \ahom.

\medskip
\ctwo{\noindent\emph{Phase~1)}} For all $t\in\NT(\acqalpha)$, where $\ahom(t)\in\adomp\nam$, and assuming $\ahom(t)=e^j$, 
let $\ahomtwo(t):=e^i$.
%
We first prove an auxiliary result and subsequently show that, for the terms mapped so far, \ahomtwo is a homomorphism of \acqalpha into $\Imc_i$.
\begin{claim}\label{eq:match-basic-concepts}
For all $B\in\BC(\aont)$, $t\in\NT(\acqalpha)$ with $\ahom(t)\in\adomp\nam$, $\ahom(t)\in B^{\Jmc_i}$ implies $\ahomtwo(t)\in B^{\Imc_i}$.
\end{claim}
\begin{proof}[Proof of the claim]
%
By Lemma~\ref{lem:dlltcqs-idef}, $\ahom(t)\in B^{\Jmc_i}$ implies two options: either
(i)~$\ahom(t)\in B^{\Imc_i}$, or (ii)~$\ahom(t)\in\adomp[\Imc_j]\aux\cup\adomp[\Imc_j]\tree$ and there
is a $\abcset\subseteq\BCr(\aont)$ with $\ahom(t)\in(\bigsqcap\abcset)^{\Imc_j}$ and $\aont\models\bigsqcap\abcset\sqsubseteq B$.
In Case~(i), we have $\ahomtwo(t)=\ahom(t)$ by definition, hence the claim holds.
In Case~(ii), the claim follows if $i=j$, as in Case~(i); otherwise, we distinguish the following two cases.

The first case is for $\ahom(t)\in\adomp[\Imc_j]\tree$; then $\ahom(t)$ is of the form
$a_{b\apath}^j$. Since $a_{b\apath}$ then occurs in some ABox by Definition~\ref{def:dll-io}, only $\Amc_{\RF}$ contains assertions on elements of \NIT, and $\Amc_{\RF}$ is the same w.r.t.\ all time points, the element $a_{b\apath}^i$ also exists.
By Lemma~\ref{lem:dll-io-elements} and the definition of~$\Amc_{\RF}$,
$\ahom(t)\in B^{\Imc_j}$ implies that $B$ subsumes the conjunction of all
rigid basic concepts satisfied by $\uel{b\apath}$ in some
$\Imc_{\langle\aont,\{\exists S(b)\}\rangle}$ where $\exists S(b)\in\RF$.
Another application of Lemma~\ref{lem:dll-io-elements} then yields that $B$
is also satisfied by $\ahomtwo(t)=a_{b\apath}^i$ in~$\Imc_i$.

In the second case, we consider $\ahom(t)\in\adomp[\Imc_j]\aux$; then $\ahom(t)$ is of the form $a_x^j$. 
Let $\acqalphatwo\in\tcqcqs\atcq$ be the (unique) CQ containing the variable~$x$.
By the definition of $\Imc_j$ and \cond\ref{def:dlltcqs-rc:cqcons}, the existence of the element~$a_x^j$ implies that $\acqalphatwo\in\QR$.
Hence, the element~$a_x^i$ must also exist, and
$\ahomtwo(t)=a_x^i$ is well-defined.
Since $\ahom(t)\in(\bigsqcap\abcset)^{\Imc_j}$; $\abcset\subseteq\BCr(\aont)$; and $\rigcons{\QR}$ and $\Amc_{\RF}$ contain all rigid assertions on~$a_x$, taking \aont into account (i.e., in particular, all following from $\Amc_{\Q_{\iota(j)}}$ are also in $\rigcons{\QR}$); Definition~\ref{def:dll-io} and
Lemma~\ref{lem:dll-io-elements} yield that the elements of \abcset are
implied by the conjunction of
\begin{itemize}
\item all basic concepts for which there are assertions on $a_x$ in $\AS(\langle\aont,\Amc_{\{ \acqalpha\}}\rangle)$ and
\item all rigid concepts $\exists R$ for which there is an assertion
$\exists S(a_x)\in\RF$ with $\aont\models S\sqsubseteq R$.
\end{itemize}
But, for all concepts of the latter form, \cond\ref{def:dlltcqs-rc:rf} yields that
$\exists R(a_x)$ is (already) implied by some KB
$\langle\aont,
\AR\cup\rigcons{\QR}\cup\Amc_{\Q_{\iota(i')}}\cup\Amc_{i'}\rangle$,
and must hence be contained in the set of basic concepts from item one, because $\AR$ and $\Amc_{i'}$ do not contain assertions on $a_x$. 
Denoting the set of these basic concepts by $\BC(\acqalphatwo,a_x)$, we thus obtain 
$\aont\models\bigsqcap\BC(\acqalphatwo,a_x)\sqsubseteq\bigsqcap\abcset$. Given $\acqalphatwo\in\QR$, the definition of $\rigcons{\QR}$ yields
$B'(a_x)\in\rigcons{\QR}$ for all $B'\in\abcset$.
From $\Imc_i\models\rigcons{\QR}$, we obtain $a_x^i\in(\bigsqcap\abcset)^{\Imc_i}$ which together with $\Imc_i\models\aont$ leads to $a_x^i\in B^{\Imc_i}$, as required. 
\qedhere
\end{proof}
We continue with the proof of Lemma~\ref{lem:dlltcqs-qi}.

As a consequence, all concept atoms $A(t)$ in~\acqalpha with $\ahom(t)\in\adomp\nam$ are satisfied by~\ahomtwo in~$\Imc_i$.
We next show that this also holds for the role atoms that only contain such terms.
Let hence $R(t,t')\in\acqalpha$ be such that $\ahom(t),\ahom(t')\in\adomp\nam$.
If $\ahom(t)$ and $\ahom(t')$ are both contained in $\adom[\Imc_i]$, which especially holds for the elements in $\NI(\atkb)$, the claim
follows immediately from Lemma~\ref{lem:dlltcqs-idef-role} and the fact that
$\ahomtwo(t)=\ahom(t)$ and $\ahomtwo(t')=\ahom(t')$.
Otherwise, both $\ahom(t)$ and $\ahom(t')$ belong to some
$\NI(\atkb)\cup\adomp[\Imc_j]\aux\cup\adomp[\Imc_j]\tree$ for a fixed $j\in[0,n+k]$ such that $j\neq i$; note that there are no role connections between elements of different sets 
$\adomp[\Imc_j]\aux\cup\adomp[\Imc_j]\tree$ and $\adomp[\Imc_{j'}]\aux\cup\adomp[\Imc_{j'}]\tree$ in $\Jmc_i$ by definition.
\ctwo{We can thus distinguish the following cases:
(A)~$R$ is rigid, $\ahom(t)$ or $\ahom(t')$ is contained in
$\adomp[\Imc_j]\aux$, and neither is contained in
$\adomp[\Imc_j]\tree$; (B)~$R$ is rigid and one of $\ahom(t)$ and $\ahom(t')$ is contained in
$\adomp[\Imc_j]\tree$; and (C)~$R$ is flexible.}

\ctwo{In Case~(A),}
$(\ahom(t),\ahom(t'))\in R^{\Jmc_i}$ implies $(\ahom(t),\ahom(t'))\in R^{\Imc_j}$.
By Definition~\ref{def:dll-io}, considering relations between named elements, and the fact that \AR and \atkb do not contain assertions on elements from \NIA, there must be an assertion
$S(\tau(t),\tau(t'))\in\rigcons{\QR}\cup\Amc_{Q_{\iota(j)}}$ such that
$\aont\models S\sqsubseteq R$; $\tau(t)=e$ if $\ahom(t)=e^j$. By the definition of $\rigcons{\QR}$ (see Definition~\ref{def:dlltcqs-consequences}), we get
$R(\tau(t),\tau(t'))\in\rigcons{\QR}$.
Hence, $\Imc_i\models\Amc_{\QR}$ implies $(\ahomtwo(t),\ahomtwo(t'))\in R^{\Imc_i}$.

\ctwo{In Case~(B),}
we argue similarly to the previous case.
Again, from $(\ahom(t),\ahom(t'))\in R^{\Jmc_i}$, we obtain $(\ahom(t),\ahom(t'))\in$ $R^{\Imc_j}$. By Definition~\ref{def:dll-io}, since only $\Amc_{\RF}$ contains assertions on elements from \NIT, and because $\Amc_{\RF}$ contains all rigid assertions on the elements of \NIT that follow by \aont (see the definition of $\Amc_{\RF}$), there must be an assertion $R(\tau(t),\tau(t'))\in\Amc_{\RF}$.
From $\Imc_i\models\Amc_{\RF}$, we then get $(\ahomtwo(t),\ahomtwo(t'))\in R^{\Imc_i}$.

\ctwo{In Case~(C),}
given $(\ahom(t),\ahom(t'))\in R^{\Jmc_i}$ and the fact that witnesses are not defined for elements of~$\adomp\nam$, there must be a rigid role~$S$ such that $(\ahom(t),\ahom(t'))\in S^{\Imc_j}$ and $\aont\models S\sqsubseteq R$ by the definition of ${\Jmc_i}$.
As in the respective previous case (i.e., based on the kind of $\ahomtwo(t)$ and $\ahomtwo(t')$ here), it follows that $(\ahomtwo(t),\ahomtwo(t'))\in S^{\Imc_i}$.
Since $\Imc_i\models\aont$, we then obtain $(\ahomtwo(t),\ahomtwo(t'))\in R^{\Imc_i}$.

It remains to define \ahomtwo for the variables of~\acqalpha that are mapped by~\ahom into $\bigcup_{j=0}^{n+k}\adomp[\Imc_j]\un$. 
Since the relations in $\Jmc_i$ are based on those in the canonical interpretations, Definition~\ref{def:dll-io} implies that all variables that occur in role atoms together with a term mapped to an element $e\in\adomp\nam$ by \ahom are of the form $\ujel{eS_0}$, and the role atom must be an $R$-atom such that $\aont\models S_0\sqsubseteq R$. 
Moreover, since \acqalpha is connected, if any variable is mapped to an element $\ujel{eS_0\dots S_\ell}\in\adomp[\Imc_j]\un$, 
then there is a variable that is mapped to $\ujel{eS_0}$, one directly connected to it and mapped to $\ujel{eS_0S_1}$, etc. We hence can proceed as follows.

\medskip
\ctwo{\noindent\emph{Phase~2)}} We consider all $y\in\NV(\acqalpha)$ for which there is an atom $R(t,y)\in\acqalpha$ where $\ahom(t)\in\adomp\nam$ and $\ahom(y)\in\adomp[\Imc_j]\un$, and assume $\ahom(t)=e^j$, $\ahom(y)=\ujel{eS_0}$, $S_0\in\NRM$, and $\aont\models S_0\sqsubseteq R$. Recall that $\ahomtwo(t)=e^i$. 
The goal is to choose an element of $\adom[\Imc_i]$ as value for $\ahomtwo(y)$ so that \ahomtwo can be (extended to) a homomorphism of \acqalpha into $\Imc_i$. For now, we however only show that our definition of $\ahomtwo(y)$ satisfies $(\ahomtwo(t),\ahomtwo(y))\in R^{\Jmc_i}$ for all these role atoms. The remaining atoms that contain $y$ are then covered in Part 3).

If $i=j$, then we can directly define $\ahomtwo(y):=\ahom(y)$ by Lemmas~\ref{lem:dlltcqs-idef-role} and~\ref{lem:dlltcqs-idef}.
Otherwise, we distinguish the following two cases. Note that $(e^j,\ujel{eS_0})\in R^{\Jmc_i}$ implies $(e^j,\ujel{eS_0})\in R^{\Imc_j}$ by Lemma~\ref{lem:dlltcqs-idefj-role}.

If $\wits\aont{\ujel{eS_0}}\neq\emptyset$, then $(e^j,\ujel{eS_0})\in R^{\Imc_j}$ implies $(e^j,\ujel{eS_0})\in R^{\Jmc_i}$ by the definition of $\Jmc_i$.
%
We then get $e^i=\ahomtwo(t)\in(\exists S_0)^{\Imc_i}$ by Claim~\ref{eq:match-basic-concepts}. 
According to Definition~\ref{def:dll-io}, the element $\uiel{eS_0}$ then exists and the pair $(e^i,\uiel{eS_0})$ is related in $\Imc_i$
as $(e^j,\ujel{eS_0})$ is related in $\Imc_j$. 
And the latter tuple is interpreted in the same way in $\Jmc_i$. 
We can thus set $\ahomtwo(y):=\uiel{eS_0}$.

If $\wits\aont{\ujel{eS_0}}=\emptyset$, then the definition of~$\Jmc_i$ yields that, for all atoms $R(t,y)$ as above, there is a rigid role~$S$ such that $\aont\models S_0\sqsubseteq S$, $\aont\models S\sqsubseteq R$, and $(e^j,\ujel{eS_0})\in S^{\Jmc_i}$.
Note that $S_0$ must be flexible since otherwise $\exists S_0$ would be a witness for $\ujel{eS_0}$. 
Furthermore, since $(e^j,\ujel{eS_0})\in S_0^{\Imc_j}$, 
Lemma~\ref{lem:dll-io-elements} yields that the assertion
$\exists S_0(e)$ is a consequence of the basic concepts obtained from the assertions involving~$e$ in~\KR[j]. 
\ctwo{We now show that this is still the case if $\Amc_{\RF}$ is disregarded, by a case distinction on whether~$e$ belongs to $\NI(\atkb)$, $\NIA$, or $\NIT$.}

If $e\in\NI(\atkb)$, then any (rigid) basic concept assertion on $e$ that is a consequence of $\Amc_{\RF}$ and~\Omc must be contained in \AR, since \AR is a rigid ABox type and $\Imc_j$ is a model of both these ABoxes. 
Since $\Amc_{\RF}$ does not contain flexible assertions, $\exists S_0(e)$ is also a consequence of \KR[j] if $\Amc_{\RF}$ is disregarded.

If $e\in\NIA$ is of the form $e=a_x$ and \acqalphatwo is the CQ in which $x$ appears, then $\exists S_0(e)$ is a consequence of the assertions in $\rigcons{\QR}\cup\Amc_{\Q_{\iota(j)}}\cup\Amc_{\RF}$ (i.e., again, taking \aont into account). 
For $\Amc_{\RF}$, we thus consider a rigid concept assertions $\exists R'(a_x)$, $R'\in\NRM$, for which there is a flexible assertion $\exists S'(a_x)\in\RF$ such that $\aont\models S'\sqsubseteq R'$. By~\cond\ref{def:dlltcqs-rc:rf}, all those assertions $\exists R'(a_x)$ follow however from some set of assertions $\rigcons{\QR}\cup\Amc_{\Q_{\iota(j')}}$, $j'\in[0,n+k]$, and hence from $\rigcons{\QR}$. 
By the definition of $\rigcons{\QR}$, they are thus contained in $\rigcons{\QR}$, which means that $\Amc_{\RF}$ can be disregarded.

If $e\in\NIT$, then $\exists S_0(e)$ must similarly follow from ABox assertions by Lemma~\ref{lem:dll-io-elements}; particularly, it follows exclusively from $\Amc_{\RF}$ (and~\aont) because elements from \NIT do not occur in other ABoxes. Since $\Amc_{\RF}$ contains only rigid assertions, the corresponding rigid basic concepts constitute a witness for $\ujel{eS_0}$, which contradicts our assumption and yields $e\not\in\NIT$.

\ctwo{In all three cases, we have shown} the entailment required to apply the ``only if''-direction of \cond\ref{def:dlltcqs-rc:rf} to infer
that $\exists S_0(e)\in\Amc_{\RF}$.
Since $\Imc_i\models\Amc_{\RF}$, $(e^i,a_{eS_0}^i)\in S^{\Imc_i}$ holds for all rigid roles~$S$ as above. Given $\Imc_i\models\aont$, the above assumptions on \aont, and $\wits\aont{\ujel{eS_0}}=\emptyset$, $(e^i,a_{eS_0}^i)$ satisfies all the 
role atoms $R(t,y)$ in $\Imc_i$ that are mapped to $(e^j,\ujel{e S_0})$ by~\ahom. 
We can therefore define $\ahomtwo(y):=a_{e S_0}^i$.

It thus remains to consider the satisfaction of the atoms we left out in Phase~2) and, in particular, the other variables of~\acqalpha mapped by \ahom to elements of $\bigcup_{j=0}^{n+k}\adomp[\Imc_j]\un$. As described above, we can assume them to be related in a tree structure and also to the elements we focused on in 2).

\medskip
\ctwo{\noindent\emph{Phase~3)}} We finish the definition of \ahomtwo using an induction over the structures of unnamed elements in the image of the homomorphism~\ahom, starting with the elements we considered in Phase~2). 
For all variables $y$ with $\ahom(y)\in\adomp[\Imc_i]\un$, we can obviously set $\ahomtwo(y):=\ahom(y)$. We therefore only consider the case that $\ahom(y)\in\adomp[\Imc_j]\un$ and $j\neq i$ in the following. Note that this is valid for the induction by Fact~\ref{fact:dlltcqs-deltas} and the interpretation of roles in $\Jmc_i$, which show that elements from different sets $\adomp[\Imc_i]\un$ and $\adomp[\Imc_j]\un$ cannot be related in $\Jmc_i$.

Given the latter observations, we can also maintain the following invariant while finishing the construction of~\ahomtwo for the remaining variables $y$ (i.e., we do not have to satisfy the invariant at all for variables mapped by \ahom to elements from $\adomp[\Imc_i]\un$ since \ahomtwo is already defined for all variables directly connected to them); let $m$ denote the number of variables for which \ahomtwo is defined already at the respective moments in the induction: If $\ahom(y)=\ujel{\apath S_1}$, then either
\ctwo{%
\begin{align}
	&\text{$\wits\aont{\ujel{\apath S_1}}=\emptyset$, $\ahomtwo(y)=a_{\apath S_1}^i$, and $|\apath|< m$, or} \label{t:i} \\
	&\text{$\wits\aont{\ujel{\apath S_1}}\neq\emptyset$, $\ahomtwo(y)$ is of the form $\uiel{\apathtwo S_1}$, and $|\apathtwo|< m$.} \label{t:ii}
\end{align}%
}
As induction hypothesis, we assume that the (partial) definition of \ahomtwo satisfies all role atoms that only contain variables for which it is already defined and the invariant.
It can readily be checked that our definitions from Phase~2), which represent the base case, satisfy both these requirements.

To show that \ahomtwo is a homomorphism of \acqalpha into $\Imc_i$ 
for the variables mapped so far, it remains to consider the concept atoms. 
We assume $\ahom(y)=\ujel{\apath S_1}\in\adomp[\Imc_j]\un$ and $\ahomtwo(y)$ to be defined already and consider all concept atoms $A(y)\in\acqalpha$.
Since $\ujel{\apath S_1}\in A^{\Jmc_i}$, by Lemma~\ref{lem:dlltcqs-idef} either
\ctwo{%
\begin{align}
  &\text{there is a $\abcset\subseteq\BCr(\aont)$ with
$\ujel{\apath S_1}\in(\bigsqcap\abcset)^{\Imc_j}$ and
$\aont\models\bigsqcap\abcset\sqsubseteq A$, or} \label{t:ip} \\
  &\text{$\ujel{\apath S_1}\in A^{\Imc_j}$ and
$\wits\aont{\ujel{\apath S_1}}\neq\emptyset$.} \label{t:iip}
\end{align}%
}%
If Case~\eqref{t:ii} applies, meaning $\ahomtwo(y)=\uiel{\apathtwo S_1}$, then two applications of Lemma~\ref{lem:dll-deltaun} yield that $\aont\models\exists S_1^-\sqsubseteq A$ and
$\ahomtwo(y)=\uiel{\apathtwo S_1}\in A^{\Imc_i}$.
Otherwise, \eqref{t:ip} must hold because of $\ujel{\apath S_1}\in A^{\Jmc_i}$ by the definition of $\Jmc_i$. $\Imc_i\models\aont$ then implies $\ujel{\apath S_1}\in A^{\Imc_j}$, and Lemma~\ref{lem:dll-deltaun} yields $\aont\models\exists S_1^-\sqsubseteq\bigsqcap\abcset$. From $\ahomtwo(y)=a_{\apath S_1}^i$, we then get $a_{\apath S_1}^i\in(\bigsqcap\abcset)^{\Imc_i}$ by the definition of~$\Amc_{\RF}$.
Given $\Imc_i\models\aont$, we conclude that $\ahomtwo(y)\in A^{\Imc_i}$.

To continue the definition of~\ahomtwo, we consider an element $\ujel{\apath S_1S_2}\in\range(\ahom)$
and all role atoms $R(y,z)\in\acqalpha$ with $\ahom(y)=\ujel{\apath S_1}$ and
$\ahom(z)=\ujel{\apath S_1S_2}$. 
We hence can assume that \acqalpha contains a variable $x$ such that $\ahom(x)=\ujel{\apath S_1}$ for which \ahomtwo has been defined already. 
%
%
For all these role atoms, 
we have that $\aont\models S_2\sqsubseteq R$ by the definition of $\Jmc_i$ and Definition~\ref{def:dll-io}.
\ctwo{Once again, we make a case distinction whether $\wits\aont{\ujel{\apath S_1S_2}}$ is empty or not.}

If $\wits\aont{\ujel{\apath S_1S_2}}=\emptyset$, then 
$\wits\aont{\ujel{\apath S_1}}=\emptyset$ holds, and hence Case~\eqref{t:i} of our invariant applies, meaning $\ahomtwo(y)=a_{\apath S_1}^i$.
In addition, for every of the role atoms $R(y,z)$ under consideration, there is a rigid role~$S$ such that $\aont\models S\sqsubseteq R$ and $(\ujel{\apath S_1},\ujel{\apath S_1S_2})\in S^{\Imc_j}$ by the definition of~$\Jmc_i$. Definition~\ref{def:dll-io} then yields $\aont\models S_2\sqsubseteq S$.
Together with $\aont\models\exists S_1^-\sqsubseteq \exists S_2$ and the given bound on the length of \apath, this implies that the element $a_{\apath S_1S_2}^i$ exists in $\Amc_{\RF}$, in all the rigid assertions $S(a_{\apath S_1}^i,a_{\apath S_1S_2}^i)$. Thus, $\Imc_i\models\Amc_{\RF}$ and $\Imc_i\models\aont$, because of the fact that $\aont\models S\sqsubseteq R$ yield that
%
%
$(a_{\apath S_1},a_{\apath S_1S_2})$ satisfies all 
all relevant role atoms $R(y,z)$.
We set $\ahomtwo(z):=a_{\apath S_1S_2}$ for all such variables $z$ and obtain Case~\eqref{t:i} of our invariant.

\ctwo{In the remaining case that $\wits\aont{\ujel{\apath S_1S_2}}\neq\emptyset$, we know that} 
$(\ujel{\apath S_1},\ujel{\apath S_1S_2})\in S_2^{\Jmc_i}$ since the pair
is contained in $S_2^{\Imc_j}$ and also
$\aont\models\exists S_1^-\sqsubseteq\exists S_2$ by the definitions of the two interpretations.
\ctwo{We make one final case distinction on whether $\wits\aont{\ujel{\apath S_1}}\neq\emptyset$ is empty or not.}

If $\wits\aont{\ujel{\apath S_1}}\neq\emptyset$, then \eqref{t:ii} implies that $\ahomtwo(y)$ is of the form $\ahomtwo(y)=\uiel{\apathtwo S_1}$.
Given $\Imc_i\models\aont$, the element $\uiel{\apathtwo S_1S_2}$ exists, and $(\uiel{\apathtwo S_1},\uiel{\apathtwo S_1S_2})$ satisfies all role atoms
$R(y,z)$ of the above form in~$\Imc_i$ by Definition~\ref{def:dll-io}.
Hence, the definition $\ahomtwo(z):=\uiel{\apathtwo S_1S_2}$ for all such variables~$z$ maintains the invariant (Case~\eqref{t:ii}).

\ctwo{Finally,} if $\wits\aont{\ujel{\apath S_1}}=\emptyset$, then there must be a set $\abcset\subseteq\BCr(\aont)$ with
$\aont\models\bigsqcap\abcset\sqsubseteq\exists S_2$ and
$\ujel{\apath S_1}\in(\bigsqcap\abcset)^{\Imc_j}$ by Definition~\ref{def:dlltcqs-witness}, which implies
$\aont\models\exists S_1^-\sqsubseteq\bigsqcap\abcset$ by
Lemma~\ref{lem:dll-deltaun}.
Since \eqref{t:i} yields that $\ahomtwo(y)$ is of the form $a_{\apath S_1}^i$, the definition of~$\Amc_{\RF}$ implies $\abcset(a_{\apath S_1})\subseteq\Amc_{\RF}$, and $\Imc_i\models\Amc_{\RF}$ then yields
$\ahomtwo(y)\in(\bigsqcap\abcset)^{\Imc_i}$.
Together with $\Imc_i\models\aont$, we get
$\ahomtwo(y)\in(\exists S_2)^{\Imc_i}$.
\ctwo{We can thus argue as in the previous case and set $\ahomtwo(z):=\uiel{a_{\apath S_1}S_2}$.}

This concludes the construction of~$\ahomtwo$ and shows that it is a homomorphism
of~$\acqalpha$ into~$\Imc_i$, which contradicts~\ref{obs:2}.
\end{proof}

\section{Proofs for Section~\ref{sec:dlltcqs-r-sat-rew}}

	\lemFOrSattwo*
	\begin{proof}
		($\Leftarrow$) This direction is trivial.
		($\Rightarrow$)
		We assume $(\AR,\QR,\QRn,\RF)$ to be an r-complete tuple, define
		$$ \Bphi:=\{ B(a)\in\AR\cup\RF \mid B\in\BC(\aont),a\in\NI(\atcq) \}.$$
		 and show that the tuple $(\ARs,\QRs,\QRsn,$ $\RFs)$ is r-complete as well.
		%
		We focus on the conditions in Definition \ref{def:dlltcqs-r-complete}.
		Our tuple obviously satisfies Conditions~\ref{def:dlltcqs-rc:cqcons} and~\ref{def:dlltcqs-rc:qrn} by construction.
		%

		For Conditions~\ref{def:dlltcqs-rc:consistent}, \ref{def:dlltcqs-rc:negcqs}, and~\ref{def:dlltcqs-rc:witnesses}, we describe a model of \KRis that can be homomorphically embedded into the canonical interpretation of the consistent KB~$\KR[i]$ that exists for the given tuple $(\AR,\QR,\QRn,\RF)$, since it satisfies \cond\ref{def:dlltcqs-rc:consistent}.
		Observe that all positive assertions contained in
		one of the ABoxes of \KRis must also be contained in $\KR[i]$:
		\begin{itemize}
			\item $\rigcons{\QRs}\subseteq\rigcons{\QR}$ follows from the facts that both \QRs and \QR satisfy \cond\ref{def:dlltcqs-rc:cqcons} and \QRs is the minimal set satisfying that condition.
			\item $\Amc_{\RFaux}\cup\Amc_{\RFphi}\cup\Amc_{\RFother}\subseteq\Amc_{\RF}$ is a consequence of the following observations. 
			By definition, every $\exists S(b)\in\RFaux$ is a consequence of a KB $\langle\aont, \Amc_{Q_j}\rangle$, $\ax_j\in\as$. Since the given tuple satisfies \cond\ref{def:dlltcqs-rc:rf} and we have $\iota(n+j)=j$ in that definition, the assertion is also contained in~\RF.
			$\Amc_{\RFphi}\subseteq\Amc_{\RF}$ follows from the construction.
			Each $\exists S(b)\in\RFother$ follows, 
			by definition, from \ARs together with some~$\Amc_i$ (and~\aont). Since this entailment does not depend on~\Bphi or $\rigcons{\QRs}$, and the given tuple satisfies \cond\ref{def:dlltcqs-rc:consistent}, we know that \AR must contain all assertions relevant for this entailment; since \cond\ref{def:dlltcqs-rc:rf} is satisfied, $\exists S(b)$ is thus also contained in~\RF.
			
			\item 
			All positive assertions in \ARs have to be positive in \AR, too, by the
			definition of \ARs and the observations in the previous items. More
			precisely, $\ASr(\atkb)\cap\big(\Bphi\cup\rigcons{\QRs}\big)\subseteq\AR$,
			and all other positive assertions in~\ARs are implied by these initial
			assertions together with some~$\Amc_i$, $i\in[0,n]$. Since each~$\KR[i]$
			is consistent by assumption, the rigid ABox type \AR also contains the
			latter rigid consequences.
		\end{itemize}
		
		Hence, any difference between $\KR[i]$ and \KRis (i.e., focusing on the assertions in \KRis and disregarding additional assertions in $\KR[i]$) must be due to negative rigid assertions in \ARs that occur positively in \AR (because \AR is a rigid ABox type) and may cause the inconsistency of \KRis. By providing a model for \KRis, we show that such assertions cannot exist.
		Since the given tuple satisfies \cond\ref{def:dlltcqs-rc:consistent} and $\KR[i]$ contains 
		all positive assertions occurring in \KRis, the KB \ctwo{$\KRisp$}, obtained from \KRis by 
		dropping the negative assertions, is also consistent. We focus on the canonical interpretation~\Imc of that KB and show that it also satisfies \KRis. We consider negative role and basic concept assertions in \KRis.
		\begin{itemize}
			\item 
			Let $\lnot R(\indone,\indtwo)\in\ARs$.
			We prove $\Imc\not\models R(\indone,\indtwo)$ by contradiction, assuming that some of the ABoxes in $\KRisp$ contains a role assertion $S(\indone,\indtwo)$ such that $\aont\models S\sqsubseteq R$. We thus consider the positive assertions in $\ARs$, $\rigcons{\QRs}$, $\Amc_{Q_{\iota(i)}}$, $\Amc_{\RFs}$, and~$\Amc_i$.
			
			If $S$ is rigid, then we can disregard $\Amc_{Q_{\iota(i)}}$ and $\Amc_{\RFs}$, since all rigid assertions in the former ABox are also contained in $\rigcons{\QRs}$ and because $\Amc_{\RFs}$ does not contain assertions on two elements of~$\NI(\atkb)$.
			Hence, by definition, the assertion $R(a,b)$ is contained in~\ARs. Since \ARs is a rigid ABox type (\ie exactly one of $R(\indone,\indtwo)$ and $\lnot R(\indone,\indtwo)$ is contained in it), that contradicts the assumption.
			
			If $S$ is flexible, then $S(\indone,\indtwo)$ occurs in $\Amc_i$ or $\Amc_{Q_{\iota(i)}}$, which implies that $\ax_{\iota(i)}\in\as$. By the definition of~\ARs and $\rigcons{\QRs}$, based on \QRs and Definition~\ref{def:dlltcqs-consequences}, we then get $R(\indone,\indtwo)\in\ARs$ or $R(\indone,\indtwo)\in\rigcons{\QRs}$, which also implies $R(\indone,\indtwo)\in\ARs$ and thus a contradiction to the assumption.
			
			\item Let $\lnot B(a)\in\ARs$. If $a\in\NI(\atcq)$, then $\lnot B(a)\in\AR$ by the definitions of~\Bphi and~\ARs. Since \AR is a rigid ABox type and the given tuple satisfies \cond\ref{def:dlltcqs-rc:consistent}, Lemma~\ref{lem:dll-iomodel} yields $\Imc'\not\models B(a)$, assuming
			$\Imc'$ to be the canonical interpretation of \KR[i]. By our above observation on the 
			positive assertions in \KRis, this interpretation must also satisfy $\KRisp$. Hence, $B(a)$ cannot be a consequence of that KB, and Lemma~\ref{lem:iosatq} yields $\Imc\not\models B(a)$.
			
			If $a\not\in\NI(\atcq)$, then we again proceed by contradiction and assume $\Imc\models B(a)$. Lemma~\ref{lem:dll-io-elements} then yields that there are positive assertions about~$a$ in $\ARs\cup\Amc_{\RFother}\cup\Amc_j$, $j\in[0,n]$, that together imply~$B$; the other ABoxes do not contain assertions on such names. By that lemma, we can also disregard $\Amc_{\RFother}$ since \ARs already contains all relevant assertions, \ctwo{\ie if \smash{$\Amc_{\RFother}$} contains a rigid role assertion $R(a,e)$, then it must follow from some flexible concept assertion $\exists S(a)$ entailed by $\langle\aont,\ARs\cup\Amc_i\rangle$ for some~$i$, and hence $\exists R(a)$ is already included in~\ARs.}
			Then, Lemma~\ref{lem:ars-other-consequences} implies that we can actually focus on \ARs alone and obtain $B(a)\in\ARs$. This contradicts the assumption since \ARs is a rigid ABox type.
		\end{itemize}
		Since \Imc is a model of $\KRisp$ by Lemma~\ref{lem:dll-iomodel} and we have shown that it satisfies all negative assertions in \ARs, it is also a model of \KRis. Hence, our tuple satisfies  \cond\ref{def:dlltcqs-rc:consistent}.
		
		If one of Conditions~\ref{def:dlltcqs-rc:negcqs} and~\ref{def:dlltcqs-rc:witnesses} is contradicted, then Lemma~\ref{lem:iosatq} yields that there is a homomorphism of the CQ that causes the contradiction into \Imc.
		Again, the above observation that the positive assertions contained in
		\KRis must be contained in $\KR[i]$ is important.
		By Definition~\ref{def:dll-io} and the semantics, every such homomorphism into \Imc is also a homomorphism into the canonical interpretation of the positive part of $\KR[i]$. This contradicts the assumption that $\KR[i]$ satisfies Conditions~\ref{def:dlltcqs-rc:negcqs} and~\ref{def:dlltcqs-rc:witnesses}, again by 
		Lemma~\ref{lem:iosatq}.
		
		It remains to consider \cond\ref{def:dlltcqs-rc:rf}, and we make a case distinction between the three parts of~\RF. Observe that, w.r.t.\ the ABoxes considered in that condition, the individual names occurring in \RFaux can only occur within 
		$\rigcons{\QRs}\cup\bigcup_{\ax\in\as}\Amc_{Q_{\ax}}$, and those in \RFother only in $\ARs\cup\bigcup_{0\le i\le n}\Amc_i$.
		\begin{itemize}
			\item We consider the assertions in \RFaux.
			($\Rightarrow$) For every $\exists S(a_x)\in\RFaux$, and thus $a_x\in\NIA$, the definition of~\RFaux directly yields that there is a $\ax_j\in\as$ such that $\langle\aont,\Amc_{Q_j}\rangle\models\exists S(a_x)$. This solves the claim given that \cond\ref{def:dlltcqs-rc:rf} considers $\iota$ to be such that $j=\iota(n+j)$.
			
			($\Leftarrow$) If there is a world $\ax\in\as$ such that $\langle\aont,\rigcons{\QRs}\cup\Amc_{Q_\ax}\rangle\models\exists S(a_y)$, $a_x\in\NIA$, then Lemma~\ref{lem:dll-io-elements} implies that $\exists S(a_x)$ can only follow from assertions involving~$a_x$. But $a_x$ can be associated to a unique query $\acqalpha_j\in\tcqcqs\atcq$ that contains the variable~$x$ and corresponding ABox $\Amc_{\acqalpha_j}$; no other such ABoxes contains assertions on $a_x$. This implies $\acqalpha_j\in\QRs$. By the definition of \QRs, there is a $\ax'\in\as$ with $p_j\in \ax'$ and, in particular, $\Amc_{Q_{\ax'}}$ implies all
			assertions on~$a_x$ in $\rigcons{\QRs}$. This shows that $\exists S(a_x)$
			already follows from $\Amc_{Q_{\ax'}}$ (and~\aont), which yields $\exists S(a_x)\in\RFaux$.
			
			\item For \RFother, its definition directly yields the claim.

			\item We consider the assertions in \RFphi.
			Since the given tuple satisfies \cond\ref{def:dlltcqs-rc:rf} and, by the definition of \RFphi, \RF and \RFphi coincide \wrt $\NI(\atcq)$, we have that there is $i\in[0,n]$ such that $\langle\aont,\AR\cup\rigcons{\QR}\cup\Amc_{Q_{\iota(i)}}\cup\Amc_i\rangle\models \exists S(a)$ iff $\exists S(a)\in\RFphi$ for all $a\in\NI(\atcq)$.
			
			($\Leftarrow$) This direction then directly follows from the above observations that all positive assertions in \ARs occur in \AR and $\rigcons{\QRs}\subseteq\rigcons{\QR}$.
			
			($\Rightarrow$) Since \AR is a rigid ABox type, the fact that the given tuple 
			satisfies \cond\ref{def:dlltcqs-rc:consistent} yields that all basic concept assertions that can be derived from assertions in \AR or $\rigcons{\QR}$ are also contained in \AR; note that these ABoxes both contain only rigid assertions. Moreover, by definition, \Bphi contains all rigid basic concept assertions on elements of $\NI(\atcq)$ from~\AR. Hence, Lemma~\ref{lem:dll-io-elements} yields that 
			$\exists S(a)\in\RFphi$ implies that there is an $i\in[0,n]$ with
			$\langle\aont,\ARs\cup\rigcons{\QRs}\cup\Amc_{Q_{\iota(i)}}\cup\Amc_i\rangle\models \exists S(a).$
		\end{itemize}
		Thus, \cond\ref{def:dlltcqs-rc:rf} is also satisfied.
	\end{proof}

\LemPrefThreeARs*
\begin{proof}
	By Lemma~\ref{lem:dll-iomodel}, we have
	$\langle\aont,\ARs\cup\Amc_i\rangle\models\ground(\acq)$ iff
	$\langle\aont,\ARs_{\ctwo{\NBC}}\cup\Amc_i\rangle\models\ground(\acq)$.
	Thus, it is sufficient to show the claim for all $\ARs_j$ by induction
	on~$j\in[0,\ctwo{\NBC}]$.
	Moreover, since our rewriting is based on $\dllhhornPerfRef\acq\aont$ by
	replacing all atoms individually, by Definition~\ref{def:db-int},
	Lemma~\ref{lem:dllhhorn-perfect-ref-correct}, and the substitution lemma for
	first-order logic it suffices to show that
	\[
		\ahom(\aatom(\vec{t}))\in\ARs_j\cup\Amc_i
		\text{ iff }
		\tdbfbs\models\ahom(\alpha^j(\vec{t},i))
	\]
	holds for all $i\in[-1,n]$, all atoms $\alpha(\vec{t})$, where
	$\alpha^j(\vec{t},i)$ is the corresponding replacement used for the definition
	of $\prefth{j}(i)$, and all functions
	$\ahom\colon\NV(\alpha(\vec{t}))\to\NI(\Kmc)$.

	For the base case $j=0$, this is easy to see by Definition~\ref{def:tdb-int}
	and a case analysis on whether
	$\ahom(\aatom(\vec{t}))\in \ARs_0 =
		\ASr(\atkb) \cap \big(\Bphi\cup\rigcons{\QRs}\big)$
	or $\ahom(\aatom(\vec{t}))\in\Amc_i$.

	Assume that the claim holds for an arbitrary $j\in[0,\ctwo{\NBC}-1]$.
	Then the case $\ahom(\aatom(\vec{t}))\in\Amc_i$ is again captured by $\Asf(\vec{t},i)$ and, by induction, the satisfaction of
	$ \pi\big(\exists p.\prefth[\alpha(\vec{t})]{j}(p)\big) $
	in \tdbfbs is equivalent to the existence of a $p\in[0,n]$ such that
$ \langle\Omc,\ARs_{j}\cup\Amc_p\rangle\models\ahom(\alpha(\vec{t})), $
	which is exactly the condition for $\ahom(\alpha(\vec{t}))$ being included in
	$\ARs_{j+1}$.
\end{proof}

\lemdlltcqsRews*
\begin{proof}
	By Lemma~\ref{lem:dllhhorn-perfect-ref-correct},
	$\langle\aont,\AKRs\cup\Amc_i\rangle$ is inconsistent iff we have
	$\db{\AKRs\cup\Amc_i}\models\dllhhornQUnsat\aont$, and it entails \acq iff
	$\db{\AKRs\cup\Amc_i}\models\dllhhornPerfRef\acq\aont$.
	It is thus sufficient to show that, for any Boolean CQ
	$\omega=\exists x_0.\ldots\exists x_{\ell-1}.\psi$,
	we have $\db{\AKRs\cup\Amc_i}\models\omega$ iff
	$\tdbfbs\models\omega_1\lor\dots\lor\omega_{2^\ell-1}$, where the latter disjunction represents our rewriting of $\omega$.

	($\Rightarrow$) We assume that there is a homomorphism~$\ahom$ of~$\omega$
	into $\db{\AKRs\cup\Amc_i}$, and show that \ahom is also a homomorphism of one of the formulas~$\omega_k$ into \tdbfbs.
	We choose the disjunct
	\begin{align*}
	\omega_k:={}&\exists'x_0.\dots\exists'x_{\ell-1}.
	\rep{\acqtwo}\land\acqtwo_{\mathsf{filter}}\text{ where}\\
	k={}&b_0\cdot 2^0+\ldots+b_j\cdot 2^j+\ldots+b_{\ell-1}\cdot 2^{\ell-1}
	\end{align*}
	with $b_j=0$ iff $\ahom(x_j)\in\NIA\cup\NIT$, for all $j\in[0,\ell-1]$.
	Moreover, for each $j$ with $b_j=0$, we consider the disjunct of~$\omega_k$
	in which~$x_j$ is equal to
	\begin{itemize}
		\item the prototype $a_{[S]\rho}$, if
			$\pi(x_j)=a_{b\rho}\in\NI(\Amc_{\exists S(b)})\setminus\{b\}$ with
			$b\in\NI(\Kmc)\setminus\NI(\atcq)$, or
		\item the auxiliary individual name $\pi(x_j)\in\NIA\cup\NITm$.
	\end{itemize}

	We show that, the formula
	$\ahom\big(\rep{\acqtwo}\land\acqtwo_{\mathsf{filter}}\big)$ is satisfied in
	\tdbfbs.
	First, consider a conjunct $\repo{s=t}$ of~$\psi_\filter$, for which there
	must exist two role atoms $R(x_j,s)$, $S(x_j,t)$ in~$\omega$ such that $x_j$ is a prototype element in the disjunct of~$\omega_k$ we
	consider, and $x_k,x_p$ are not.
	By construction, this means that
	$\pi(x_j)\in\NI(\Amc_{\exists S(b)})\setminus\{b\}$ for some
	$b\in\NI(\Kmc)\setminus\NI(\atcq)$, and that this does not hold for~$s$
	and~$t$.
	Since $R(\pi(x_j),\pi(s))$ and $S(\pi(x_j),\pi(t))$ are satisfied in
	$\db{\AKRs\cup\Amc_i}$, the elements of $\NI(\Amc_{\exists S(b)})\setminus\{b\}$ only
	occur in $\Amc_{\RFother}$, and the only other individual name that occurs
	in~$\Amc_{\exists S(b)}$ is~$b$, we must have $\pi(s)=b=\pi(t)$, \ie the
	formula $\repo{s=t}$ is satisfied by~$\pi$ in \tdbfbs.

	It remains to consider each atom~$\aatom(\vec{t})$ in~$\omega$ and show that
	$\ahom\big(\rep{\aatom(\vec{t})}\big)$ is satisfied in \tdbfbs.
	By assumption, $\ahom(\aatom(\vec{t}))\in\AKRs\cup\Amc_i$, and we now consider the specific parts of $\AKRs\cup\Amc_i$ that $\aatom(\vec{t})$ is mapped into.
	\begin{itemize}
		\item If $\ahom(\aatom(\vec{t}))\in\ARs\cup\Amc_i$, then 
			$\vec{t}$ cannot contain a variable~$x_j$ with $b_j=1$, since these ABoxes
			only contain individual names in~$\NI(\Kmc)$.
			If $\aatom$ is rigid, then due to
			$\langle\aont,\ARs\cup\Amc_i\rangle\models\ahom(\aatom(\vec{t}))$
			the part $\repone{\aatom(\vec{t})}=\prefth[\aatom(\vec{t})]{}(i)$
			of the formula $\rep{\aatom(\vec{t})}$ is satisfied in \tdbfbs by
			Lemma~\ref{lem:pref3-ars}.
			If $\aatom$ is flexible, then we must have
			$\ahom(\aatom(\vec{t}))\in\Amc_i$, and thus a corresponding disjunct $\B(t_1,i)$
			or $\R(t_1,t_2,i)$ (depending on the shape of~\aatom) is satisfied in
			\tdbfbs by Definition~\ref{def:tdb-int}.
		\item If
			$\ahom(\aatom(\vec{t})) \in
				\rigcons{\QRs}\cup\Amc_{Q_\ax}\cup\Amc_{\RFaux}\cup\Amc_{\RFphi}$,
			then there is a disjunct of $\reptwo{\aatom(\vec{t})}$ that is of the form
			$\repo{t_1=\ahom(t_1)}$ or
			$\repo{t_1=\ahom(t_1)}\land\repo{t_2=\ahom(t_2)}$,
			which is obviously satisfied under~\ahom.
		\item If $\ahom(\aatom(\vec{t}))\in\Amc_{\RFother}$, then either
			(i)~$\pi(\vec{t})$ contains only elements from 
			$\NI(\Amc_{\exists S(b)})\setminus\{b\}$ with 
			$b\in\NI(\Kmc)\setminus\NI(\atcq)$, or (ii)~it contains both one element
			from this set and~$b$ itself.
			In both cases, we have $\exists S(b)\in\RFother$.
			Moreover, by our construction, in~$\omega_k$ the elements of
			$\NI(\Amc_{\exists S(b)})\setminus\{b\}$ are considered via the
			corresponding prototypes.
			In case~(i), we hence have $\aatom(\vec{t})\in\Amc_{\exists S}$, and thus
			\[ \repthree{\aatom(\vec{t})} =
				\exists x.\rep[x]{\exists S} =
				\exists x.\exists p.\prefth[\exists S(x)]{}(p) \land
					\bigwedge_{a\in\NI(\atcq)}(x\neq a)
			\]
			can be satisfied by mapping~$x$ to~$b$ and~$p$ to~$j$ such that
			$\langle\aont,\ARs\cup\Amc_j\rangle\models\exists S(b)$; note that such a time point must exist
			by the definition of~\RFother. This is correct due to
			Lemma~\ref{lem:pref3-ars}.
			In case~(ii), we can similarly show that
			$\repthree{\aatom(\vec{t})}=\rep[t_1]{\exists S}$ is satisfied since
			$\pi(t_1)=b$ holds without loss of generality (if $\pi(t_2)=b$, we can
			consider the equivalent atom $\alpha^-(t_2,t_1)$ instead of
			$\alpha(t_1,t_2)$).
	\end{itemize}

	($\Leftarrow$)
	We assume that \tdbfbs satisfies one of the disjuncts~$\omega_k$,
	$k\in[0,2^\ell-1]$, via a homomorphism~\ahom, and extend~\ahom to the
	variables~$x_j$ with $b_j=1$ in such a way that $\ahom(\omega)$ is satisfied
	in \db{\AKRs\cup\Amc_i}.
	We consider a satisfied disjunct of~$\omega_k$ that corresponds to some
	assignment of the variables~$x_j$ with~$b_j=1$ to elements of
	$\NIA\cup\NITm\cup\NIP$.
	If $b_j=1$ and $x_j$ is considered to be an element of $\NIA\cup\NITm$ in this
	disjunct, then we set $\ahom(x_j)$ to this element.

	If $b_j=1$ and $x_j\in\NIP$ in this disjunct, then we consider the largest
	connected subset of atoms in~$\omega$ that contains~$x_j$ and for which at
	least one variable in each atom is considered to be an element of~\NIP.
	Since, for each of these atoms, the corresponding formula in
	$\repthree{\aatom(\vec{t})}$ must be satisfied, we know that they can all be
	mapped into an ABox of the form~$\Amc_{\exists S}$.
	If this mapping does not involve the root of that ABox, all these rewritings
	are of the form $\exists x.\rep[x]{\exists S}$, which by
	Lemma~\ref{lem:pref3-ars} and the definition of~$\RFother$ implies that there
	is an element~$b\in\NI(\Kmc)\setminus\NI(\atcq)$ such that the above atoms can
	be satisfied in~$\Amc_{\exists S(b)}$ instead of~$\Amc_{\exists S}$.
	In this case, we arbitrarily choose one such~$b$ and let $\ahom$ map each such
	variable~$x_j$ to the corresponding element in~$\NI(\Amc_{\exists S(b)})$.
	Otherwise, for at least one of the atoms, a formula of the form
	$\rep[t_1]{\exists S}$ must be satisfied in \tdbfbs, which means that
	$\pi(t_1)$ is an element of $\NI(\Kmc)\setminus\NI(\atcq)$.
	Similar arguments as above yield that the ABox $\Amc_{\exists S(\pi(t_1))}$ is
	part of $\Amc_{\RFother}$.
	Moreover, by $\psi_\filter$, any terms~$t_1$ occurring in such a way
	in~$\omega$ must be mapped by~$\pi$ to the \emph{same} element~$b$ of
	$\NI(\Kmc)\setminus\NI(\atcq)$.
	Hence, we can extend~$\pi$ to the variables above by considering the
	corresponding elements in~$\Amc_{\exists S(b)}$.

	To see that this definition of~$\ahom$ is correct, consider an arbitrary atom
	$\aatom(\vec{t})$ in~$\omega$ and its rewriting $\rep{\aatom(\vec{t})}$ in the
	disjunct of~$\omega_k$ that we considered above, \ie it is satisfied in
	\tdbfbs under~\ahom.
	If this is the case because $\repthree{\aatom(\vec{t})}$ is satisfied in
	\tdbfbs, we have argued \ctwo{in the previous paragraph} that $\ahom(\aatom(\vec{t}))$ is
	satisfied in \db{\AKRs\cup\Amc_i}.
	If $\reptwo{\aatom(\vec{t})}$ is satisfied, then we directly obtain
	$\ahom(\aatom(\vec{t})) \in
		\rigcons{\QRs}\cup\Amc_{Q_\ax}\cup\Amc_{\RFaux}\cup\Amc_{\RFphi}$.
	Since all of these ABoxes are part of $\AKRs\cup\Amc_i$, the claim follows.

	Finally, consider the case that $\repone{\aatom(\vec{t})}$ is satisfied in
	\tdbfbs.
	This can only be the case if no element of $\vec{t}$ is considered to be in
	$\NIA\cup\NITm\cup\NIP$, \ie $\vec{t}$ contains only variables or individual
	names from $\NI(\atcq)$.
	If $\alpha$ is flexible, the definition of the rewriting and
	Definition~\ref{def:tdb-int} yield that $\ahom(\aatom(\vec{t}))\in\Amc_i$.
	Similarly, if $\alpha$ is rigid, we obtain $\ahom(\aatom(\vec{t}))\in\ARs$
	from Lemmas~\ref{lem:pref3-ars} and~\ref{lem:ars-other-consequences}.
\end{proof}



\section{Proofs for Section~\ref{sec:dlltcqs-dc}}

\dlltcqsDCLB*
\begin{proof}
It is well-known that every finite monoid~\amonoid (i.e., a finite, closed set having an associative binary operation and an identity element) can be directly translated (in logarithmic time) to a deterministic finite automaton (DFA) that decides the word problem for that monoid, by treating the elements of \amonoid as states and considering transitions according to the associative operation.%
%
\footnote{We refer the reader to \cite{BaISt90:nc1uniformity} for details about monoids, groups, and the word problem in that context.} 
Moreover, for some such monoids (e.g., the group S5), this problem is complete for \DLogTime-uniform \NCone under \DLogTime-uniform \ACzero reductions \cite[\citecor10.2]{BaISt90:nc1uniformity}; and \DLogTime-uniform \NCone equals \ALogTime \cite[\citelem7.2]{BaISt90:nc1uniformity}. 

We hence can establish \ALogTime-hardness by considering an arbitrary DFA~\aautom 
and reducing its word problem to TCQ entailment in logarithmic time.
For that, we adapt a construction of \cite[\citethm9]{AKKRWZ-IJCAI15:omtqs}. 

Let $\aautom$ be a tuple of the form $(\automQ,\automA,\automT,\automqs,\automF)$, specifying the set of states~\automQ, the alphabet~\automA, the transition relation~\automT, the initial state~\automqs, and the set of final states~\automF.
Because we consider data complexity, the task is to specify 
a TCQ~$\atcq_{\aautom}$ based on~\aautom and an ABox sequence~$\afbs_w$ based on an arbitrary input word $w\in\automA^*$ such that
\aautom accepts $w$ iff $\langle\emptyset,\afbs_w\rangle\models\atcq_{\aautom}.$
%
%
We consider concept names $A_\sigma$ and $Q_q$ for all characters $\sigma$ of the input
alphabet~\automA and states $q\in\automQ$, respectively, and
define the following TCQ:
\[ \atcq_{\aautom}:=
  \Boxm\Big(\bigwedge_{q\rightarrow_{\sigma}q'\in\automT}\big(\big(Q_q(a)\land A_\sigma(a)\big)
    \rightarrow\Next Q_{q'}(a)\big)\Big)
      \rightarrow\bigvee_{q_f\in\automF} Q_{q_f}(a). \]
      
For a given input word $w = \sigma_0 \dots\sigma_{n-1} $, we then define the sequence $\afbs_w=(\afb_i)_{0\le i< n}$ of ABoxes as follows: $\afb_0:=\{Q_{\automqs}(a)\}$ and, for all $i\in[0,n]$, $\afb_i:=\{ A_{\sigma_i}(a)\}$.
It is easy to see that this reduction can be computed in logarithmic time.%

Given that the semantics of TCQ entailment focus on time point $n$, it can readily be checked that the model of $\langle\emptyset,\afbs_w\rangle$ that satisfies the premise of~$\atcq_{\aautom}$ at $n$ represents the run of~\aautom on~$w$. Observe that there is only one such model relevant for entailment since \aautom is deterministic.
Hence, \aautom accepts $w$ iff all models of $\langle\emptyset,\afbs_w\rangle$ that
satisfy the premise also satisfy the disjunction~$\bigvee_{q_f\in\automF}Q_{q_f}(a)$ at $n$. This is equivalent
to the entailment $\langle\emptyset,\afbs_w\rangle\models\atcq_{\aautom}$.
\end{proof}

\thmdlltcqsDCAalogtime*
\begin{proof}
  The ATM \atm accepts the input~$n$ and~\tdbfbs (in logarithmic time) iff there are sets
	$\as=\{\ax_1,\dots,\ax_k\}\subseteq 2^{\pv}$ and
	$\Bphi\subseteq\{B(a)\mid B\in\BC(\aont), a\in\NI(\atcq)\}$, a valuation $v\in\Vmc$, and types $\atype_0,\dots,\atype_n$ as follows, where $w_i:=\atype_i\cap\pv$:
	\begin{itemize}
		\item $\atype_0$ is initial and $\Pmc^v\subseteq\atype_n$;
		\item for every $i\in[0,n]$, the pair $(\atype_i,\atype_{i+1})$ is
		t-compatible;
		\item $w_n\in\atmfut{\as}{v}$;
    \item for every $i\in[0,n]$, we have $w_i\in\as$;
		\item for every $i\in[0,n]$, we have $\tdbfbs\models\rsat[w_i](i)$;
		\item for all $S\in\NFRM(\aont)$ and $a\in\NI(\atcq)$, we have:\\
		$\exists S(a)\in\Bphi$ iff there is an $i\in[0,n]$ such that
		$\tdbfbs\models\rewPRef{\exists S(a)}$ (w.r.t.~$w_i$). 
    \item for every $\ax\in\as$, we have $\tdbfbs\models\rsat(-1)$;
	\end{itemize}
	By Lemmas~\ref{lem:ltl-periodic-model} and~\ref{lem:dlltcqs-atm-split}, the first four points are equivalent to the existence of a set~\as and worlds~$w_i$ as above and an LTL-structure~\altlint as follows:
	\begin{itemize}
		\item \altlint only contains worlds from~\as,
		\item \altlint starts with $w_0,\dots,w_n$, 
		\item $\altlint,n\models\pa{\atcq}$.
	\end{itemize}
	Moreover, because of the condition that each $w_i$ is an element of~\as, the sequence
	$w_0,\dots,w_n$ can equivalently be expressed by a mapping
	$\iota\colon[0,n]\to[1,k]$ such that that $w_i=\ax_{\iota(i)}$ for all $i\in[0,n]$.
  Finally, by Definition~\ref{def:tcqs-t-sat} and Lemmas~\ref{lem:tcq-sat-iff}, \ref{lem:dlltcqs-iff-s-r-consistent}, \ref{lem:fo-r-sat2}, and~\ref{lem:fo-r-sat1}, the above items characterize the satisfiability of~\atcq w.r.t.~\Kmc.
	The claim now follows from the fact that the class \ALogTime is closed under
	complement (see~\cite[\citethm2.5]{alternation}).
\end{proof}


\section{Proofs for Section~\ref{sec:exdlltcqs}}

\lemKromTransf*
\begin{proof}
	($\Rightarrow$) We assume $\aint\not\models\lnot\acq$, which yields $\aint\models\acq$, and hence that there is a corresponding homomorphism by Definition~\ref{def:tcqs-semantics}.
	Observe that the atoms in the CQ~\acq always refer to the concepts and roles of the corresponding CI $C\sqsubseteq D$ in the same way, so that $C$ and $\lnot D$ are modeled in the CQ. Thus, the shape of~\acq together with our assumption that \aint satisfies the CIs w.r.t.\ $\overline{D}$ and the semantics of the constructor $\forall$%
	yield that there is an element~\el in the domain of \aint such that $\el\in C^\aint$ and $\el\not\in D^\aint$. This directly yields $C^\aint\not\subseteq D^\aint$, and thus $\aint\not\models C\sqsubseteq D.$
	($\Leftarrow$) The proof for this direction is by dual arguments.
\end{proof}

\thmexdllRigidConUB*
\begin{proof}
	We consider \DLLitekrom (see Corollary~\ref{cor:dlltcqs-krombool-is-the-same}) and use Lemma~\ref{lem:tcq-sat-iff} for checking satisfiability of \atcq w.r.t.\ \atkb.
  As in the proof of Theorem~\ref{thm:dlltcqs-krombool-cc:ub-w/o-rigid}, we can assume \atkb to be of the form $\langle\aont,\emptyset\rangle$, since integrating the ABoxes into the TCQ does not influence combined complexity.
  This means that $\iota$ is irrelevant.
  We can then guess a set~$\as\subseteq 2^{\pv}$ in exponential time and check t-satisfiability of \pa{\atcq} w.r.t.\ this set in \ExpTime~\cite{BaBL-JWS15}.

  For r-satisfiability, we adapt a technique from~\cite{BaBL-JWS15,BaGL-TOCL12}.
	We guess a set $\rigctypes\subseteq2^{\NRC(\aont)}$, which specifies the combinations of rigid concept names that are allowed to be satisfied by domain elements in the models of the conjunctions~$\chi_i$, and a mapping $\tau\colon\NI(\atcq)\to\rigctypes$ that fixes the rigid concepts for each individual occurring in~\atkb---similar to the rigid ABox types we considered in previous sections.
	Based on $\tau$, we define a polynomial-sized ontology~$\aont_\tau$ and CQ~$\acqtwo_\tau$ as follows:
	\begin{align*}
	\aont_\tau &:= \{A_{\tau(a)}\equiv C_{\tau(a)} \mid a\in\NI(\atcq)\} 
	\cup\{  \top\sqsubseteq A\sqcup\overline{A},\ A\sqcap\overline{A}\sqsubseteq\bot \mid A\in\NRC(\aont) \},\\
	\acqtwo_\tau &:= \bigwedge_{a\in\NI(\atcq)}A_{\tau(a)}(a)
	\end{align*}
	where 
	$\equiv$ is an abbreviation for both $\sqsubseteq$ and $\sqsupseteq$, and
	$C_\rigctype$ with $\rigctype\subseteq\NRC(\aont)$ is defined as
	$C_\rigctype := \bigsqcap_{A\in \rigctype}A\sqcap
	\bigsqcap_{A\in\NRC(\aont)\setminus \rigctype}\overline{A}$.
	We further say that an interpretation~$\Jmc=(\adom[\Jmc],\cdot^\Jmc)$ \emph{respects~\rigctypes} if
	\[ \rigctypes=\{\rigctype\subseteq\NRC(\aont) \mid
	\text{$
		(C_\rigctype)^\Jmc\neq\emptyset$}\}. \] 
	In \cite[\citelem6.2]{BaBL-JWS15}, it is shown that \as is r-satisfiable w.r.t.~\atkb iff there are a set \rigctypes and mapping $\tau$ as above such that each conjunction
	$\chi_i\land\acqtwo_\tau$ with $i\in[1,k]$ has a model w.r.t.\ $\aont\cup\aont_\tau$ that
	respects~\rigctypes.
	The proof considers the DL \SHQ, but similarly holds for \DLLitekrom.
	
	Although it seems that the \NExpTime result now directly follows from Lemma~\ref{lem:dll-cq-literal-conj-sat} stating that satisfiability of conjunctions of CQ literals can be decided in exponential time, this is not the case. The restriction that \rigctypes should be respected causes an exponential blowup.
	For that reason, we consider the proof of~\cite[\citethm8]{BoMP-IJCAI13}, which addresses UCQ entailment, in more detail.
	In that paper, an exponentially large looping tree automaton is constructed that recognizes exactly those (forest-shaped) canonical models of the KB---in a wider sense---that do not satisfy the given UCQ.
	We integrate the check that the interpretations respect \rigctypes into the automaton.
	To this end, we restrict the state set to consider only models where every
	domain element satisfies some $C_\rigctype$ with $\rigctype\in\rigctypes$.
	To ensure that each $\rigctype\in\rigctypes$ is represented somewhere in the model, we check $|\rigctypes|$ variants of this automaton for emptiness, each of which considers an ABox of the form $\{A(a)\mid A\in \rigctype\}\cup\{\overline{A}(a)\mid A\in\NRC(\aont)\setminus\rigctype\}$,
	where $a$ is a fresh individual name.
	The disjoint union of all resulting interpretations is still a model of
	the original KB that does not satisfy the UCQ.
	It can readily be checked that this modified procedure for deciding UCQ non-entailment is sound and complete given the result of~\cite[\citethm8]{BoMP-IJCAI13}.
	Satisfiability of the conjunctions $\chi_i\land\acqtwo_\tau$ can thus be decided in exponential time, because the constructed automata are of exponential size and emptiness of looping tree automata can be decided in polynomial time~\cite[\citethm2.2]{VaWo86:looptreeautomataptime}.
\end{proof}

\thmKromRigidLB*
\begin{proof}[Proof Sketch]
We adapt a reduction proposed in~\cite{BaGL-TOCL12} (see the proof of Theorem~4.1), where the word problem for exponentially space-bounded ATMs is reduced to the satisfiability problem in $\ALC$-LTL with global CIs and rigid names. 
While the assertions in the proposed $\ALC$-LTL formula can be directly viewed as conjuncts of a TCQ, the global CIs cannot all be transferred into a \DLLitekrom ontology since \ALC is much more expressive than \DLLitekrom. However, we show how some of the critical CIs can be adapted to comply with the shapes given in Table~\ref{tab:dlltcqs-krom-transf}, and how the remaining ones---with qualified existential restrictions on the right-hand-side---can be replaced by equivalent new constructions. %
%
The latter are inspired by~\cite{KRH-TOCL2013:horndls} (see Section~6.2 in that paper).
	
	As in the original proof,
	we assume w.l.o.g.\ 
	that the ATM never moves to the left when it is on the left-most tape cell;
	that it has an accepting state~$q_a$ and a rejecting state~$q_r$, designating accepting and rejecting configurations, respectively;
	that any configuration where the state is neither $q_a$ nor $q_r$ has at least one successor configuration;	
and that all computations of the ATM are finite (see \cite[\citethm2.6]{alternation}).
	%
	We disregard transitions that do not move the head ($N$).
	Further, we \ctwo{assume for simplicity} that the length of every computation on a word
	$w\in\Sigma^k$ is bounded by $2^{2^k}$, and that every configuration in such a
	computation can be represented using $\le 2^k$ symbols, plus one to represent
	the state.%
	\footnote{\ctwo{Strictly speaking, the space bound is $2^{p(k)}$ for a polynomial~$p$, but omitting~$p$ does not affect the proof.}}
	
	According to \cite[\citecor3.5]{alternation}, there is an exponentially
	space-bounded alternating TM $\atm= (Q,\Sigma, \Gamma, q_0, \Delta)$
	whose word problem is \TwoExpTime-hard;
	\ctwo{$Q$ represents the set of \emph{states}, which is partitioned into the sets $Q_\exists$ of \emph{existential states} and $Q_\forall$ of \emph{universal states}; 
	$\Sigma$ is the \emph{input alphabet};
	$\Gamma$ is the \emph{work alphabet} of the TM, containing the \emph{blank symbol}~$B$ and all symbols from $\Sigma$;
	$q_0$ is the \emph{initial state};
	and $\Delta$ denotes the \emph{transition relation}.}
	We show that this problem can be reduced
	to TCQ satisfiability
	in \DLLitekrom with rigid role names.
	
	To this end, let $w = \sigma_0 \dots\sigma_{k-1} \in\Sigma^*$ be an arbitrary 
	input word given to \atm. We next construct a TCQ $\atcq_{\atm,w}$ and a TKB 
	$\langle\aont_{\atm,w},(\Amc_0)\rangle$ in \DLLitekrom 
	such that \atm accepts $w$ iff $\phi_{\atm,w}$  is satisfiable w.r.t.\  
	$\langle\aont_{\atm,w},(\Amc_0)\rangle$.
	We use two counters modulo $2^k$, $A$ and $A'$.
	We can consider a tree describing the computations of the ATM in that one path describes one computation. The individual configurations are represented explicitly, one after the other, and each as a chain, such that every tree node represents one of the $2^k$ tape cells of a configuration; these cells are numbered by the \emph{rigid} counter $A$.
	Each tree node or cell is represented by an individual in the reduction and, since these individuals are related by rigid roles, the computation tree ``exists'' at all time points; the time points are numbered by the $A'$ counter.
	Branching models the universal transitions. Different from usual computation trees, the tree however splits at the node representing the cell under the head of the machine; the remaining parts of the configuration where the splitting occurs are replicated in each of the subtrees.
	Before 
	specifying the TCQ and ontology, we introduce all symbols we use below:
	\begin{itemize}
		\item A single named individual $a$ identifies the root of the tree.
		\item Rigid role names $R_{q,\varrho,M}$ where $q\in Q$, 
		$\varrho\in\Gamma$, and $M\in\{L ,R \}$ represent the edges of the tree.
		We collect all these role names in the set~\Rmc.
		
		Note that these roles represent the major difference to the reduction of
		\cite{BaGL-TOCL12}, where a single rigid role fulfills this purpose, but is
		used within qualified existential restrictions on the right-hand side of CIs.
				\item Rigid concept names $A_0,\ldots, A_{k-1}$ are used to model the bits of
				a binary counter numbering the tape cells in the configurations.
			
				\item Rigid concept names $I$ and $H$ point out special cells. In particular,
				$I$ is satisfied by the nodes representing the initial configuration, and $H$
				is satisfied by all nodes representing a tape cell that is located (anywhere)
				to the right of the head in the current configuration.
				
		\item 
		A rigid concept name, for each element in $Q\cup\Gamma$, represents
		the tape content, the current state, and the head position in each 
		configuration in the tree: if \atm is in a state $q$ and the head is 
		on the $i$-th tape cell, then the individual (tree node) representing this 
		cell satisfies the concept name $q$; we correspondingly represent the symbols in 
		$\Gamma$.
						\item The rigid concept names $T_{q,\varrho,M}$, for all $q\in Q$,
						$\varrho\in\Gamma$, and $M\in\{L ,R \}$, are satisfied by an individual,
						representing a cell, if the head is on the left neighboring cell and the ATM
						executes the transition $(q,\sigma,M)$ in the described configuration.
						
	\end{itemize}
	We use the temporal dimension to synchronize successor configurations in
	accordance with the chosen transition in order to model the change in the tape
	contents, the head position, and the state from one configuration to the next:
	\begin{itemize}
		\item 
		Flexible concept names $A'_0,\ldots, A'_{k-1}$ are used to model a counter in
		the temporal dimension.
		Its value is incremented (modulo 
		$2^{k}$) similar to the counter $A_0,\ldots, A_{k-1}$ but along the temporal dimension and, at every time point, all individuals of the domain share the value of this counter.
		It is used for the synchronization of successor configurations: if the $A'$ 
		counter has value $i$, then
		the symbol in the $i$-th tape cell of any configuration (where $i$ is not the 
		head position) is 
		propagated to the $i$-th tape cell of its successor configuration.
		Similarly, the state is propagated from the cells~$c$ directly right of the
		head position,
		each pointing out a specific transition (via the symbols $T_{q,\varrho,M}$),
		to the corresponding cells of the successor configurations (i.e., these cells 
		have the same position on the tape as~$c$ for right-moves and otherwise lie 
		two to the left).
		\item We further use a flexible concept name, for each element in
		$Q\cup\Gamma$, which as above is distinguished from the rigid version by a prime.
		Considering a fixed time point, these names are used for the propagation 
		of the state $q$ or cell content $\sigma$ of a cell $c$ to the corresponding
		cell in the successor configuration(s).
		This propagation happens via the right neighboring cells of that configuration,
		which then satisfy $q'$ and $\sigma'$, respectively, at the time point whose
		$A'$-counter corresponds to the $A$-counter at~$c$.
	\end{itemize}
	We may further use concept names of the form $\overline{A}$ for given concept
	names $A$ as detailed in Lemma~\ref{lem:krom-transf}.

	In the remainder of the proof, we define the TCQ $\atcq_{\atm,w}$ and the TKB
	$\langle\aont_{\atm,w},(\Amc_0)\rangle$ by describing the conjuncts of
	$\atcq_{\atm,w}$ and listing the CIs contained in $\aont_{\atm,w}$.
	To enhance readability, we may use CIs that are not in \DLLitekrom, but can be
	transformed as described in the beginning of this section (see Table~\ref{tab:dlltcqs-krom-transf} and Example~\ref{ex:exdlltcqs-gci-trans}).
	We first express the tree structure in general.
	
		We enforce all elements to have some successor except if they 
		satisfy $q_a$ or $q_r$.
		Since the only elements satisfying a symbol from $Q$ are the ones 
		representing the position of the head, the tree generation thus is only 
		stopped if we meet a halting configuration:
		$$\overline{q_a}\sqcap\overline{q_r}\sqsubseteq
		 \bigsqcup_{R_\delta\in\Rmc}\exists R_\delta.$$
		Using a big disjunction over all possible roles, we can correctly represent the nondeterminism of the machine.
		
		The $A$-counter is incremented alongside the tree modulo $2^k$ and modeled using the following CIs for all $i\in[0, k-1]$:
		\begin{align*}
		\bigsqcap_{0\le j\le i}A_j\sqsubseteq& \bigsqcap_{R_\delta\in\Rmc}\forall R_\delta.\overline{A}_i,\\
		\bigsqcap_{0\le j< i}A_j\sqcap\overline{A}_i\sqsubseteq&\bigsqcap_{R_\delta\in\Rmc}\forall R_\delta.{A}_i,\\
		\Big(\bigsqcup_{0\le j<i}\overline{A}_j\Big)
		\sqcap{A}_i\sqsubseteq& \bigsqcap_{R_\delta\in\Rmc}\forall R_\delta.{A}_i,\\
		\Big(\bigsqcup_{0\le j<i}\overline{A}_j\Big)
		\sqcap\overline{A}_i\sqsubseteq&\bigsqcap_{R_\delta\in\Rmc}\forall 
		R_\delta.\overline{A}_i.
		\end{align*}
		For example, if the bits $A_0$, \dots, $A_i$ are all true in the current tape
		cell, then in the successor cell these bits are all false.
		
		We thus have described a sequence of configurations where we can address
		single tape cells in all the configurations using the A-counter. The latter restarts every time it has reached $2^k-1$, and thus with each new configuration.

		The counter is initialized with value $0$ at $a$. Hence, all elements representing the first tape cell in
		some configuration in the tree satisfy the auxiliary concept name $\C{A=0}$,
		defined as follows:
		\[ \C{A=0}\equiv\overline{A}_0\sqcap\ldots\sqcap \overline{A}_{k-1}. \]
		Below, we use additional concept names of the form~$\C{A=i}$, for
		(polynomially many) different values~$i$, which we assume to be defined similarly. \ctwo{Moreover, we assume that $\C{A\neq i}$ is defined as the negation of $\C{A=i}$ as described in Lemma~\ref{lem:krom-transf}.}
		
		We further add the assertion
		\[ \C{A=0}(a) \]
		to $\Amc_0$. Since the names $\overline{A}_0,\ldots,\overline{A}_{k-1}$ are rigid, this assertion must be satisfied at every time point.

		We now enforce basic conditions which help to ensure that the tree actually 
		represents a successful computation of \atm on $w$.
%
		To formulate these conditions, we use the rigid concept name $H$ to identify
		the tape cells that are to the right of the head:
		\[ \Big(H\sqcup\bigsqcup_{q\in Q}q\Big)\sqcap\C{A\neq 2^k-1}
		\sqsubseteq \bigsqcap_{R_\delta\in\Rmc}\forall R_\delta.H. \]
		Thus, the propagation stops at tree levels whose elements represent the last cell in a configuration, since these elements satisfy $\C{A=2^k-1}.$
		
		
		There is only one head position per configuration:
		$$H\sqsubseteq\bigsqcap_{q\in Q}\overline{q}.$$
		Note that we do not have to consider the elements representing the cells left
		to the head since, if such a cell satisfies a concept name from $Q$, then all
		its successors in the tree are enforced to satisfy $H$.
		
		Each tape cell is associated with at most one state (which, at the same time,
		represents the position of the head):
		\[ \top\sqsubseteq\bigsqcap_{q_1,q_2\in Q,q_1\neq q_2}
		\overline{q_1}\sqcup\overline{q_2}. \]
		
		Each tape cell contains exactly one symbol:
		\[ \top\sqsubseteq\bigsqcup_{\sigma\in\Gamma}\Big(\sigma\sqcap
		\bigsqcap_{\sigma'\in\Gamma\setminus\{\sigma\}}\overline{\sigma'}\big). \]
	
	Before specifying the remaining, more intricate conditions for the 
	synchronization of the configurations, we describe the first configuration 
	in the tree (starting at $a$) as the initial configuration.
	
		In particular, we mark the corresponding elements by
		adding the assertion $I(a)$ to $\Amc_0$ and by propagating the concept alongside the first configuration 
		as follows:
		$$I\sqcap\C{A\neq 2^k-1}\sqsubseteq \bigsqcap_{R_\delta\in\Rmc}\forall R_\delta.I.$$
		
		The first configuration is modeled by adding the assertion $q_0(a)$ to $\Amc_0$
		and by considering the following CIs for all $i\in[0,k-1]$:
		\begin{align*}
		I \sqcap \C{A=i}&\sqsubseteq \sigma_i,\\ 
		I \sqcap \C{A=k}&\sqsubseteq B,\\
		I \sqcap B\sqcap \C{A\neq 2^k-1}&\sqsubseteq 
		\bigsqcap_{R_\delta\in\Rmc}\forall R_\delta.B
		\end{align*}
		where $w=\sigma_0\dots\sigma_{k-1}$ is the input word.
	
	We finally come to the most involved part, the synchronization of the 
	configurations, which includes the modeling of the transitions.
	%
		
		We first introduce the $A'$-counter, which is incremented along the
		temporal dimension.
		For every possible value of this counter, there is a time point where $a$
		belongs to the concepts from the corresponding subset of
		$\{A'_0,\ldots,A'_{k-1}\}$. This is expressed using the following conjunct
		of~$\atcq_{\atm,w}$:
		\[ \Boxf\bigwedge_{0\le i<k}
		\Big(
		\Big(\bigwedge_{0\le j<i}A'_j(a)\Big)
		\leftrightarrow
		\big(A'_i(a)\leftrightarrow\Next\lnot A'_i(a)\big)
		\Big). \]
		This formula expresses that the $i$-th bit of the $A'$-counter is flipped
		from one world to the next iff all preceding bits are true. Thus, the
		value of the $A'$-counter at the next world is equal to the value at the
		current world incremented by one.
		
	Note that it is not necessary to initialize this counter to $0$ in $\Amc_0$; we only
		need to know that all possible counter values are represented at some time
		point.
		
		The value of the $A'$-counter is always shared by all individuals:
		$$\Boxf
		\Big(
		\bigwedge_{0\le i<k}\exists x.A'_i(x)\to\lnot\exists 
		x.\overline{A}'_i(x)
		\Big).$$
		
		For the application of the $A'$-counter, we introduce the abbreviation $\C{A=A'}$ describing the equality of the two counters:
		\begin{align*}
		\C{A_i=A+i'}
		&\equiv
		\left(A_i\sqcap A'_i\right)\sqcup
		\left(\overline{A}_i\sqcap \overline{A}'_i\right),\\
		\C{A=A'}&\equiv \bigsqcap_{0\le i<k}\C{A_i=A_i'}.
		\end{align*}
		Furthermore, we define similar abbreviations as follows (note that we consider addition modulo $2^k$):
		\begin{align*}
		\C{A=A'+1}\equiv {}&
		\bigsqcap_{0\le j< k}\left(A'_{j}\sqcap \overline{A}_{j}\right)
		\sqcup\\&\bigsqcup_{0\le i<k}\left(
		\bigsqcap_{0\le j< i}\left(A'_{j}\sqcap \overline{A}_{j}\right)
		\sqcap\overline{A}'_{i}\sqcap {A}_{i}
		\sqcap\bigsqcap_{i+1\le j<k}\C{A_j=A_j'}
		\right),\\
		\C{A=A'+2}\equiv{}&
		\C{A_0=A_0'}\sqcap\left(\bigsqcap_{1\le j<k}\left(A'_{j}\sqcap 
		\overline{A}_{j}\right)\sqcup
		\right.\\&\left.
		\bigsqcup_{1\le i<k}\left(
		\bigsqcap_{1\le j< i}\left(A'_{j}\sqcap \overline{A}_{j}\right)
		\sqcap\overline{A}'_{i}\sqcap {A}_{i}
		\sqcap\bigsqcap_{i+1\le j<k}\C{A_j=A_j'}
		\right)\right).
		\end{align*}
		We now can use the temporal dimension to propagate information from one level
		of the tree to the next one as outlined above,
		and hence specify the transitions.
		
		Symbols not under the head are copied:
		\begin{align*}
		\sigma\sqcap\bigsqcap_{q\in Q}\overline{q}\sqcap \C{A=A'}
		&\sqsubseteq \bigsqcap_{R_\delta\in\Rmc}\forall R_\delta.\sigma',\\
		\sigma'\sqcap\C{A\neq A'}
		&\sqsubseteq \bigsqcap_{R_\delta\in\Rmc}\forall R_\delta.\sigma',\\
		\sigma'\sqcap\C{A=A'}
		&\sqsubseteq \sigma.
		\end{align*}
		
		To describe the transitions,
		we explicitly store chosen transitions with the help of the rigid 
		concepts $T_{p,\varrho,M}$, by enforcing them to be satisfied by the elements 
		representing the cells directly right-neighbored to the head position.
		Recall that there may be several such cells; we are now at the point where we specify the branching of the tree.
		Hence, we model the transitions for all $q\in Q$ and $\sigma\in\Gamma$
		using the following CIs:
		\begin{align*}
		q\sqcap\sigma 
		&\sqsubseteq \bigsqcup_{\delta\in\Delta(q,\sigma)}\exists R_\delta,
		\quad\text{if $q\in Q_\exists$,} \\
		q\sqcap\sigma
		&\sqsubseteq \bigsqcap_{\delta\in\Delta(q,\sigma)}\exists R_\delta,
		\quad\text{if $q\in Q_\forall$,} \\
		q\sqcap\sigma
		&\sqsubseteq\bigsqcap_{\delta\in\Delta(q,\sigma)}\forall R_\delta.T_\delta.
		\end{align*}
		Observe that our main adaptation of the proof in~\cite{BaGL-TOCL12} is that we, instead of considering a single role $R$, deal with all those in $\Rmc$ and, instead of considering one $R$-successor per successor configuration $\delta$, consider an $R_\delta$-successor. This enables us to simulate the qualified existential restriction of the form $\bigsqcap_{\delta\in\Delta(q,\sigma)}\exists R. T_\delta$ on the right-hand side of a CI in the original proof, via the last two of the above CIs.
		
		The (possible) replacement of the symbols under the head is described with 
		the help of the transition concepts $T_{q,\sigma,M}$ for all $q\in Q,$ 
		$\sigma\in\Gamma,$ and $M\in\{L,R\}$:
		$$T_{q,\sigma,M}\sqcap \C{A=A'+1}\sqsubseteq
		\bigsqcap_{R_\delta\in\Rmc}\forall R_\delta.\sigma'. $$
		Recall that the transition concepts are only enforced to hold at the cell to
		the right of the current head position (hence the $+1$).
		
		The state information is similarly propagated for all $q\in Q$ and
		$\sigma\in\Gamma$ as follows:
		\begin{align*}
		T_{q,\sigma,R}\sqcap \C{A=A'}
		&\sqsubseteq \bigsqcap_{R_\delta\in\Rmc}\forall R_\delta. q', \\
		T_{q,\sigma,L}\sqcap \C{A=A'+2}
		&\sqsubseteq \bigsqcap_{R_\delta\in\Rmc}\forall R_\delta. q', \\
		q'\sqcap\C{A\neq A'}
		&\sqsubseteq \bigsqcap_{R_\delta\in\Rmc}\forall R_\delta.q',\\
		q'\sqcap\C{A=A'}
		&\sqsubseteq q.
		\end{align*}
	
	We lastly enforce the computation to be an accepting one by disallowing the state $q_r$ entirely using the CI $q_r\sqsubseteq\bot.$
	Note that this is correct since we assume all the computations of \atm to be terminating. 
	This finishes the definition of the Boolean TCQ~$\atcq_{\atm,w}$ and the global 
	ontology~$\aont_{\atm,w}$, which consist of the conjuncts and CIs specified above.
	We further collect all assertions in the ABox $\Amc_0$.
	Given our descriptions above, it is easy to see that the size of $\atcq_{\atm,w}$, $\aont_{\atm,w},$ and $\Amc_0$ is
	polynomial in~$k$. Moreover, it can readily be checked that our constructions are equivalent to those in the proof of \cite[\citethm4.1]{BaGL-TOCL12}. 
	Hence, $\atcq_{\atm,w}$ is satisfiable w.r.t.\ 
	$\langle\aont_{\atm,w},(\Amc_0)\rangle$ iff \atm accepts~$w$.
\end{proof}

\LemExdlltqcsRSat*
\begin{proof}
	($\Rightarrow$)
	Let $\Jmc_0,\dots,\Jmc_{n+k}$ be the interpretations over the common
	domain~$\Delta$ that exist by the r-satisfiability of~\as w.r.t.~$\iota$.
	As in the proof of Lemma~\ref{lem:dlltcqs-iff-s-r-consistent}, we can assume
	that they interpret the individual names in~\NIA in such a way that
	each~$\Jmc_i$ satisfies~$\Amc_{\Q_{\iota(i)}}$.
	We can thus already fix the ABox types
	$\aboxtype{i}:=\{\alpha\mid\Jmc_i\models\alpha\}$ for all $i\in[0,n+k]$, where
	$\alpha$ ranges over all assertions formulated over $\NI(\atkb)\cup\NIA$,
	$\NC(\Omc)$, and $\NR(\Omc)$ (see Definition~\ref{def:exdlltcqs-aboxtype}).
	To find the types~\typeai{a}{i}, we first unravel the interpretations~$\Jmc_i$
	into tree-shaped models~$\Jmc^\ast_i$.
	However, in contrast to classical (atemporal) unraveling techniques, we need
	to construct a common domain for all time points, and hence need to unravel the
	interpretations $\Jmc_0,\dots,\Jmc_{n+k}$ simultaneously.

	For this purpose, we iteratively extend the interpretations in a sequence of interpretations
	$\Jmc^j_0,\dots,\Jmc^j_{n+k}$ over a common domain~$\Delta^j$, for
	increasing $j\ge 0$.
	At each step~$j$ of this construction, we also maintain a function
	$g\colon\Delta^j\to\Delta$ that maps the domain elements of the tree-shaped
	interpretations to the ``original'' domain elements from~$\Delta$, such that, \ctwo{for all $d\in\Delta^j$ and $i\in[0,n+k]$, $d$ satisfies the same basic concepts in~$\Jmc_i^j$ as $g(d)$ does in the original interpretation~$\Jmc_i$}.
	We start with the domain $\Delta^0:=\{u_a\mid a\in\NI(\atkb)\cup\NIA\}$, the
	function~$g$ that maps each~$u_a$ to~$a^{\Jmc_0}$, and the interpretations
	$\Jmc^0_i$, $i\in[0,n+k]$, defined such that $a^{\Jmc^0_i}=u_a$ for all
	$a\in\NI(\atkb)\cup\NIA$, and otherwise uniquely determined
	by~$\aboxtype{i}$.
	However, the elements~$u_a$ may not yet satisfy all existential
	restrictions in~\aboxtype{i}.
	At each index $j> 0$ and for each domain element
	$u_\rho\in\Delta^j$ added in the previous step, 
	we therefore introduce a number of fresh
	(i)~flexible role successors~$u_{\rho R^i}$ for all $i\in[0,n+k]$ and
	$R\in\NFRM(\Omc)$, and (ii)~rigid role successors~$u_{\rho R}$ for
	$R\in\NRRM(\Omc)$
%
if $g(u_\rho)\in(\exists R)^{\Jmc_i}$. 
Then, there must also be a $d\in\Delta$ such that
	$(g(u_\rho),d)\in R^{\Jmc_i}$, and we can add $u_{\rho R^i}$ to $\Delta^{j+1}$ and
	set $g(u_{\rho R^i}):=d$. We also add the pair $(u_\rho,u_{\rho R^i})$ to all
	$S^{\Jmc^{j+1}_{i'}}$ for which we have $(g(u_\rho),d)\in S^{\Jmc_{i'}}$, for
	all $i'\in[0,n+k]$; hence, 
	$u_\rho\in(\exists R)^{\Jmc^{j+1}_i}$, meaning that the existential restriction
	at~$u_\rho$ is satisfied.
	We further interpret the basic concepts on~$u_{\rho R^i}$ in~$\Jmc^{j+1}_{i'}$
	as on~$d$ in~$\Jmc_{i'}$ (pending the introduction of role successors for
	$u_{\rho R^i}$).
	For all rigid roles~$R\in\NRRM(\Omc)$ (case~(ii)), we proceed in the same way,
	but only choose one successor~$u_{\rho R}$ for all time points.
	The unraveled interpretations~$\Jmc^\ast_i$, $i\in[0,n+k]$, over the domain
	$\Delta^\ast$ are obtained as the limit of this construction.
	It is easy to show that these interpretations still satisfy all the properties
	required for Definition~\ref{def:tqa-r-sat}; in particular, none of the
	negative CQ literals in the conjunctions~$\chi_{\iota(i)}$ can become
	violated by unraveling.
	Moreover, by construction, each~$\Jmc^\ast_i$ still satisfies~$\aboxtype{i}$.
	We denote by~$\Delta^\ast|_a$ the \emph{subtree of~$\Delta^\ast$ starting
	at~$u_a$}, \ie the set of all domain elements of the form~$u_{a\rho}$, and
	by~$\Jmc^\ast_i|_a$ the interpretation~$\Jmc^\ast_i$ restricted to this
	subtree.

	We now extract the types
	$\typeai{a}{i}:=\Tmc(\Jmc^\ast_i|_a)$ of these
	subtrees, which are triples of the form $(\Bmc^a_i,\Qmc^a_i,\Mmc^a_i)$,  where
	\begin{itemize}
		\item $\Bmc^a_i$ contains exactly those (negated) basic concepts~\ctwo{$(\lnot)B$} for
		which $\Jmc^\ast_i|_a\models \ctwo{(\lnot)B}(a)$,
		\item $\Qmc^a_i$ contains all CQs $\acq\in\tcqcqs\atcq$ that are satisfied
		in~$\Jmc^\ast_i|_a$, and
		\item $\Mmc^a_i$ contains all $S\subseteq\NT(\acq)$ for which there is a
		partial homomorphism~$\pi$ of~$\acq$ into~$\Jmc^\ast_i|_a$ with
		$u_a\in\range(\pi)$ and $S=\domain(\pi)$.
	\end{itemize}
	We thus immediately obtain the temporal types
	$\ttypea{a}=\{\typeai{a}{i}\mid i\in[0,n+k]\}$ for all individual names
	$a\in\NI(\atkb)\cup\NIA$.

	The next step is to make the structure of the models independent of the time
	points $i\in[0,n+k]$, by grouping the time points according to the
	types~$\typeai{a}{i}$, of which there are constantly many (for a fixed~$a$).
	Hence, we define a mapping $f_a\colon\ttypea{a}\to[0,n+k]$ that arbitrarily
	chooses, for each $\Tmc\in\ttypea{a}$, one representative time point
	$f_a(\Tmc)$ such that $\Tmc=\typeai{a}{f_a(\Tmc)}$.
	We now replace each subtree $\Jmc^\ast_i|_a$ in~$\Jmc^\ast_i$ by the subtree
	$\Jmc^\ast_{f_a(\typeai{a}{i})}|_a$ that is representative for the
	type~$\typeai{a}{i}$; that is, we interpret all concept (role) names on (pairs
	of) elements from~$\Delta^\ast|_a$ in the same way as at time point
	$f_a(\typeai{a}{i})$.
	These changes are \enquote{local} to each individual name~$a$, \ie the other
	individual names~$b$ do not affect the subtrees below~$a$.
	We denote the resulting interpretations by~$\Jmc^\lozenge_i$, $i\in[0,n+k]$,
	which have the same domain $\Delta^\lozenge=\Delta^\ast$ as before.
	It is easy to show that the types of the time points remain the same as
	in~$\Jmc^\ast_i$, \ie we have $\Tmc(\Jmc^\lozenge_i|_a)=\Tmc(\Jmc^\ast_i|_a)$.
	Moreover, \aont and $\aboxtype{i}$ clearly remain satisfied at each time point
	$i\in[0,n+k]$, which implies that the positive CQ literals
	in~$\chi_{\iota(i)}$ are also satisfied.
	For a negative CQ literal~$\lnot\acq$ in $\chi_{\iota(i)}$, assume that it
	becomes satisfied when replacing~$\Jmc^\ast_i|_a$ with
	$\Jmc^\lozenge_i|_a=\Jmc^\ast_{f_a(\typeai{a}{i})}|_a$.
	Since $i$ and~$f_a(\typeai{a}{i})$ have the same type~$\typeai{a}{i}$, and in
	particular agree on the second component~$\Qmc^a_i$ of that type, the
	homomorphism that maps~$\acq$ into~$\Jmc^\lozenge_i$ must map some terms to
	elements outside of~$\Delta^\ast|_a$.
	However, this also means that the set of all terms of~$\acq$ that are mapped
	by~$\pi$ into~$\Delta^\ast|_a$ is included in the third component of the type,
	$\Mmc^a_i$.
	Thus, there is a similar partial homomorphism~$\pi'$ also
	into~$\Jmc^\ast_i|_a$, which can be merged with the remainder of~$\pi$ to
	obtain a homomorphism of~$\acq$ into~$\Jmc^\ast_i$, contradicting the fact
	that this interpretation satisfies~$\chi_{\iota(i)}$.

	We can now remove the domain elements $u_{a\rho R^i}$ (and all their
	successors) that refer to any time point~$i$ that is not contained in
	$\range(f_a)$, since those elements were introduced only to satisfy
	existential restrictions at~$i$, which are now replaced by those of time point
	$f_a(\typeai{a}{i})$.
	We further assume that all time points~$i$ occurring in the names of the
	remaining domain elements from~$\Delta^\lozenge$ are replaced by the
	corresponding types~$\typeai{a}{i}$, \eg $u_{aR^i}$ is replaced by
	$u_{aR^{\typeai{a}{i}}}$.
	Hence, our domain already has the form required for the tree ABoxes in
	Definition~\ref{def:tree-aboxes}.
	The only remaining obstacle to obtain models of~$\KR[i]$ that we can use
	for~\ref{exdll-rc-consistent} and~\ref{exdll-rc-negcqs} is the fact that
	different named elements~$a$ and~$b$ having the same temporal type
	$\ttypea{a}=\ttypea{b}$ may still have non-isomorphic
	subtrees~$\Jmc^\lozenge_i|_a$ and~$\Jmc^\lozenge_i|_b$ of unnamed successors.

	We follow a similar approach as above, and take a partial function
	$h\colon\TTypes\to\NI(\atkb)\cup\NIA$ that maps each temporal type
	$\tau\in\TTypes$ to an arbitrary individual name~$a$ with $\ttypea{a}=\tau$,
	if such an individual exists, and that is undefined on all other temporal types
	(\ie those that are not realized by any individual name in our
	interpretations).
	We now replace all subtrees $\Delta^\lozenge|_a$ by
	$\Delta^\lozenge|_{h(\ttypea{a})}[a]$, which denotes the set obtained from
	$\Delta^\lozenge|_{h(\ttypea{a})}$ by replacing the name~$h(\ttypea{a})$
	with~$a$, \ie each~$u_{h(\ttypea{a})\rho}$ becomes~$u_{a\rho}$.
	We denote the resulting domain by~$\Delta^\dagger$.
	Correspondingly, we interpret the concept and role names
	on~$\Delta^\dagger|_a=\Delta^\lozenge|_{h(\ttypea{a})}[a]$
	in corresponding interpretations~$\Jmc^\dagger_i|_a$ as on $\Delta^\lozenge|_{h(\ttypea{a})}$
	in~$\Jmc^\lozenge_j|_{h(\ttypea{a})}$, where
	$j=f_{h(\ttypea{a})}(\typeai{a}{i})$ 
	is the time point corresponding to
	$\typeai{a}{i}$ in the subtree belonging to the individual~$h(\ttypea{a})$;
	this is well-defined since $a$ and~$h(\ttypea{a})$ have the same temporal
	type.
	For example, we have $(u_a,u_{aR^\Tmc})\in S^{\Jmc^\dagger_i}$ iff
	$(u_{h(\ttypea{a})},u_{h(\ttypea{a})R^\Tmc})\in S^{\Jmc^\lozenge_j}$ etc.
	We first show that the types remain the same, \ie we have $\Tmc(\Jmc^\dagger_i|_a)=\Tmc(\Jmc^\lozenge_i|_a)$ for all $a\in\NI(\atkb)\cup\NIA$ and $i\in[0,n+k]$.
	Since $\Jmc^\dagger_i|_a=\Jmc^\lozenge_j|_{h(\ttypea{a})}[a]$, where $j$ is
	defined as above, and
	$\Tmc(\Jmc^\lozenge_j|_{h(\ttypea{a})})=\Tmc(\Jmc^\lozenge_i|_a)$, it is clear
	that the basic type remains the same.
	Moreover, any (partial) homomorphism~$\pi$ of some $\acq\in\tcqcqs\atcq$ into
	$\Jmc^\lozenge_j|_{h(\ttypea{a})}$ that does not include the individual
	name~$h(\ttypea{a})$ in its domain is clearly valid also in~$\Jmc^\dagger_i$
	(after renaming), and vice versa.
	But if~$\pi$ does refer to~$h(\ttypea{a})$ (or~$a$), then it follows from
	$\Tmc(\Jmc^\lozenge_j|_{h(\ttypea{a})})=\Tmc(\Jmc^\lozenge_i|_a)$ that
	$a=h(\ttypea{a})$, since the individual name must be mentioned explicitly in
	either the second or third component of that type.
	In this case, it follows that $\Jmc^\dagger_i|_a=\Jmc^\lozenge_i|_a$, and the
	claim trivially holds.
	It is easy to show that the interpretations~$\Jmc^\dagger_i$ still
	satisfy~\Omc and~$\aboxtype{i}$.
	The fact that $\Jmc^\dagger_i\models\chi_{\iota(i)}$ follows as for
	$\Jmc^\lozenge_i$ from what we have shown above, namely that the types remain
	the same.

	Our interpretations now have the following properties.
	For each temporal type $\tau\in\TTypes$ that is realized in
	$\Jmc^\dagger_0,\dots,\Jmc^\dagger_{n+k}$, there is a representative
	individual name~$a=h(\tau)$ such that $\ttypea{a}=\tau$ and all other
	individual names~$b$ with the same temporal types have
	subtrees~$\Delta^\dagger|_b$ that behave in the same way as
	$\Delta^\dagger|_a$ (after reshuffling the time points to match the types,
	using~$f_a$).
	Similarly, for each type $\Tmc\in\tau$, there is a representative time point
	$i=f_a(\Tmc)$ such that $\typeai{a}{i}=\Tmc$ and all other time points~$j$
	with the same type have subtrees~$\Jmc^\dagger_j|_a$ that are isomorphic
	to~$\Jmc^\dagger_i|_a$.
	Hence, we can define the complete tree ABoxes
	$(\tabox{\tau}{\Tmc})_{\Tmc\in\tau}$ by introducing a new individual
	name~$u_\rho$ for each~$u_{a\rho}$ that appears in~$\Jmc^\dagger_i|_a$, and
	collecting all assertions about these individual names that hold
	in~$\Jmc^\dagger_i|_a$---until the termination condition~\ref{ta:bounded} of
	Definition~\ref{def:complete-tree-aboxes} is satisfied.
	It is easy to show that $(\tabox{\tau}{\Tmc})_{\Tmc\in\tau}$ are actually tree
	ABoxes for~$\tau$ according to Definition~\ref{def:tree-aboxes}.
	Moreover, we have constructed~$\Jmc^\dagger_i$, $i\in[0,n+k]$, in such a way
	that they satisfy~\aont, \aboxtype{i}, \tabox{\atuple}{i},
	and~$\chi_{\iota(i)}$, which means that~\ref{exdll-rc-consistent}
	and~\ref{exdll-rc-negcqs} are satisfied.

	($\Leftarrow$)
	Let \atuple be a tuple for which~\ref{exdll-rc-consistent}
	and~\ref{exdll-rc-negcqs} are satisfied.
	We iteratively construct tree-shaped interpretations $\Jmc_0,\dots,\Jmc_{n+k}$
	to satisfy Definition~\ref{def:tqa-r-sat}.
	To start, we define these interpretations up to the point where they are
	uniquely determined by the ABox types~\aboxtype{i} and the renamed complete
	tree ABoxes in~\tabox{\atuple}{i}.
	It remains to extend these interpretations to actual (possibly infinite)
	models of~\Omc, $\Amc_i$, and~$\chi_{\iota(i)}$.
	The only thing we have to be careful about in this process is to not
	accidentally satisfy some CQ that occurs negatively in~$\chi_{\iota(i)}$.
	Hence, consider a path from a root~$a$ to a leaf~$u_{a\rho}$, in these tree
	ABoxes, for which $u_{a\rho}$ is lacking some necessary role successors.
	By~\ref{ta:bounded}, we know that there exist an ancestor~$u_{a\rho'}$
	of~$u_{a\rho}$ such that $\rho=\rho'\sigma$ with $|\sigma|=d$, and an
	ancestor~$u_{a\rho''}$ of~$u_{a\rho'}$ such that the subtree of depth
	$d=\max\{|\NT(\acqalpha)|\mid\acqalpha\in\tcqcqs\atcq\}$ below $u_{a\rho''}$
	is isomorphic to the subtree of depth~$d$ below~$u_{a\rho'}$ (in case that
	$\rho''=\eps$, the element $u_{a\rho''}$ is actually equal to~$a$).
	We now extend the subtree below~$u_{a\rho'}$ by a copy of the subtree
	below~$u_{a\rho''}$, which in particular introduces the required role
	successors of~$u_{a\rho}$.
	Note that no existing assertions are replaced by this operation, due to the
	requirements that the subtrees are isomorphic up to depth~$d$, and the fact that 
	$u_{a\rho'}$ does not contain any successors beyond depth~$d$,
	by~\ref{ta:bounded}.
	We can continue this process indefinitely to obtain the desired models.
	Moreover, a CQ~$\acq$ that occurs negatively in~$\chi_{\iota(i)}$ can never
	become satisfied by the copied domain elements, because it could only
	be mapped into a subtree of maximal depth~$v$, which, by our construction,
	must have an isomorphic image in the original tree ABoxes
	in~\tabox{\atuple}{i}, which is a contradiction to~\ref{exdll-rc-negcqs}.
\end{proof}

\end{appendix}
\bibliographystyle{spmpsci}
\bibliography{ref}

\end{document}